\documentclass[12pt]{article}
\usepackage{geometry, booktabs, subcaption}
\usepackage{titlesec}
\usepackage[T1]{fontenc}
\usepackage{authblk}
\usepackage{amssymb,amsmath, amsthm}
\usepackage{graphicx}
\graphicspath{{IMG/}}

\newtheorem{theorem}{Theorem}%
\newtheorem{corollary}{Corollary}
\newtheorem{lemma}{Lemma}
\newtheorem{proposition}[theorem]{Proposition}%
\newtheorem{remark}{Remark}%
\newtheorem{definition}{Definition}

\usepackage{xr}
\usepackage{xcolor}
\usepackage{amsmath}
\usepackage{subcaption}
\usepackage{amsfonts,dsfont,mathrsfs}
\usepackage{booktabs,multirow,array}
\usepackage{bm}
\usepackage{cleveref}
\usepackage{placeins}
\usepackage{natbib}
\usepackage[linesnumbered,ruled,vlined]{algorithm2e}
\makeatletter
\renewcommand{\algocf@captiontext}[2]{#1\algocf@typo. \AlCapFnt{}#2} % text of caption
\def\@algocf@capt@plain{top}
\renewcommand{\algocf@makecaption}[2]{%
	\addtolength{\hsize}{\algomargin}%
	\sbox\@tempboxa{\algocf@captiontext{#1}{#2}}%
	\ifdim\wd\@tempboxa >\hsize%     % if caption is longer than a line
	\hskip .5\algomargin%
	\parbox[t]{\hsize}{\algocf@captiontext{#1}{#2}}% then caption is not centered
	\else%
	\global\@minipagefalse%
	\hbox to\hsize{\box\@tempboxa}% else caption is centered
	\fi%
	\addtolength{\hsize}{-\algomargin}%
}
\makeatother

\def\R{\mathds{R}}

\newcommand{\Wcal}{{\mathcal W}}
\newcommand{\push}{\#}

\renewcommand{\S}{\mathbb S}
\newcommand{\dd}{\mathrm d}
\newcommand{\quant}{F^{-}}
\newcommand{\mutilde}{{\widetilde{\mu}}}

\newcommand{\nutilde}{{\widetilde{\nu}}}

\newcommand{\Ttilde}{\widetilde{T}}
\newcommand{\mubar}{\bar{\mu}}

\newcommand{\virgolette}[1]{{``#1''}}

\newcommand{\dist}{\mathrm d_R}
\DeclareMathOperator*{\argmin}{argmin}

\allowdisplaybreaks[3]

\title{Wasserstein Principal Component Analysis for Circular Measures}

\author[1]{Mario Beraha}
\author[2]{Matteo Pegoraro}

\affil[1]{Department of Economics and Statistics, University of Torino}

\affil[2]{Department of Mathematical Sciences, Aalborg University}

\begin{document}

\maketitle

%\corresp{\href{mario.beraha@unito.it}{mario.beraha@unito.it}, \href{matteop@math.aau.dk}{matteop@math.aau.dk}}

\begin{abstract}
We consider the 2-Wasserstein space of probability measures supported on the unit-circle, and propose a framework for Principal Component Analysis (PCA) for data living in such a space. We build on a detailed investigation of the optimal transportation problem for measures on the unit-circle which might be of independent interest. In particular, we derive an expression for optimal transport maps in (almost) closed form and propose an alternative definition of the tangent space at an absolutely continuous probability measure, together with the associated exponential and logarithmic maps.
PCA is performed by mapping data on the tangent space at the Wasserstein barycentre, which we approximate via an iterative scheme, and for which we establish a sufficient a posteriori condition to assess its convergence. Our methodology is illustrated on several simulated scenarios and a real data analysis of measurements of optical nerve thickness.
\end{abstract}
\textbf{Keywords}:
Optimal Transport, Circular Measures, PCA, Weak Riemannian Structure, Distributional Data Analysis

\maketitle

\section{Introduction}

The analysis of complex data, such as high-dimensional, functional, compositional, or manifold-valued data, is an emerging trend in the statistical literature.
Such complex data are in fact routinely collected by medical imaging, genomic analyses, earth sciences etc..
To achieve meaningful analyses, it is fundamental that the space in which data take values is endowed with the right mathematical structure to capture the variability of the phenomena under investigation.
For instance, when analysing functional data, it is often the case that one wants to consider functions defined up to reparametrisation of the domain. This leads to the problem of ``alignment'' \citep[see, e.g.,][]{aneurisk_jasa}. Similarly, when analysing compositions, it is not suitable to embed the analysis in the usual euclidean space since the data are constrained on the unit-dimensional simplex, which makes operations such as addition and scalar multiplication meaningless unless care is taken \citep{menafoglio}.
At the same time, some metrics might be more suitable than others to compare two datapoints, think for instance at the sup or $L_2$ norm for functions.

In this paper, we focus on distributional data analysis, that is, a particular case of analysis of complex data, where datapoints are probability measures. Specifically,
we consider the 2-Wasserstein space of probability measures supported on the unit-circle $\S_1 := \{(x, y) \in \R^2: x^2 + y^2 = 1\}$, and propose a framework for Principal Component Analysis (PCA) for data living in such a space.
PCA is popular among practitioners as it produces both a set of orthogonal directions, usually interpreted as the main directions of variability in the dataset, and a map from the space where data live onto the space generated by such directions.
Hence, PCA can be used to visually interpret the variability in the dataset and to reduce the dimensionality of the data, by projecting data on their scores. In particular, classical multivariate statistical tools, such as linear regression or clustering, can be carried out by working on the scores.
In the context of data living on nonlinear spaces, this latter feature is particularly appealing, as it allows using out-of-the-box tools directly on the PCA scores.

Our investigation stems from the analysis in \cite{ali2021circular}, where measurements of the optical nerve head, obtained via Optical Coherence Tomography (OCT), are studied in connection to the development and progression of optic neuropathies such as glaucoma. 
The OCT produces a circular scan of the eye measuring neuroretinal rim (NRR) thickness, so that
each datapoint can be considered as a function supported on $\S_1$. These are then normalised to eliminate undesired variability introduced by different magnitudes, so that data can be considered as probability densities on $\S_1$.
A clustering pipeline on the coefficients of the Fourier series expansion of the densities is then developed, thus taking into account the circular nature of the support but overlooking the compositional nature of the data.

Analysing probability densities with methodologies from functional data analysis has been questioned in recent years, as this overlooks the constrained nature of such objects. See, e.g., \cite{menafoglio} and the references therein. Similarly, dealing with data supported on $\mathbb \S_1$ by mapping the circle to a subset of the real line might produce misleading results, as they depend on the map chosen. See, for instance, \Cref{fig:ex_data}, where we compare the Wasserstein distances between two datapoints in the dataset in \cite{ali2021circular} when seen as measures on the real line, for two different choices of the maps that ``unroll'' the circle onto $[0, 2\pi)$.
In \Cref{sec:eye} we discuss how the tools developed in this paper lead can be used for analysing the OCT measurements data.
In particular, by means of our Wasserstein PCA for measures on $\S_1$, we show how the principal directions lead to interpreting the main sources of variability in the data. Moreover, we assess the effectiveness of the dimensionality reduction pipeline by showing how a hierarchical clustering algorithm fit on the scores can effectively divide data into groups with different shapes of the optical nerve.

\begin{figure}
	\centering
		\includegraphics[height=6cm]{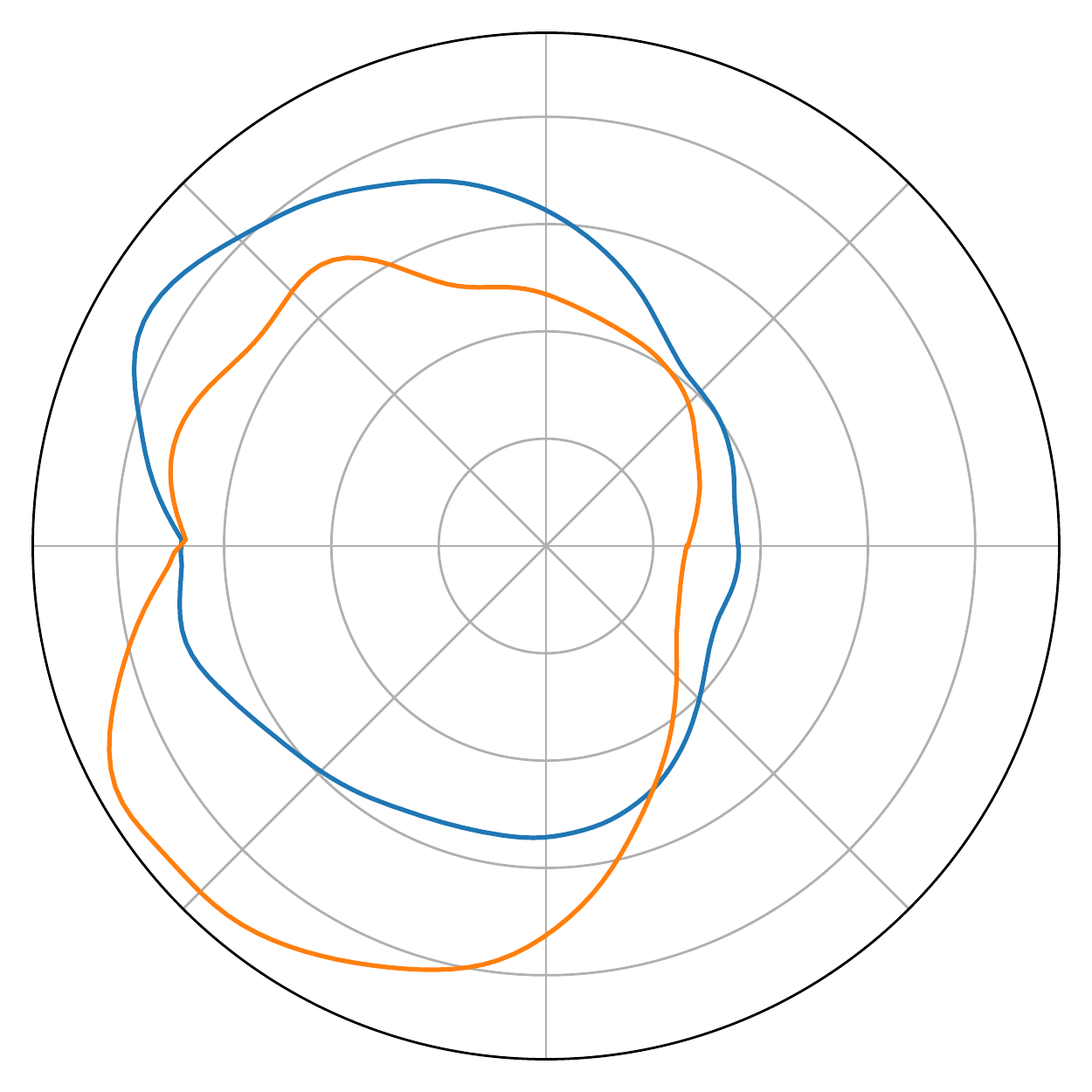}
		\includegraphics[height=6cm]{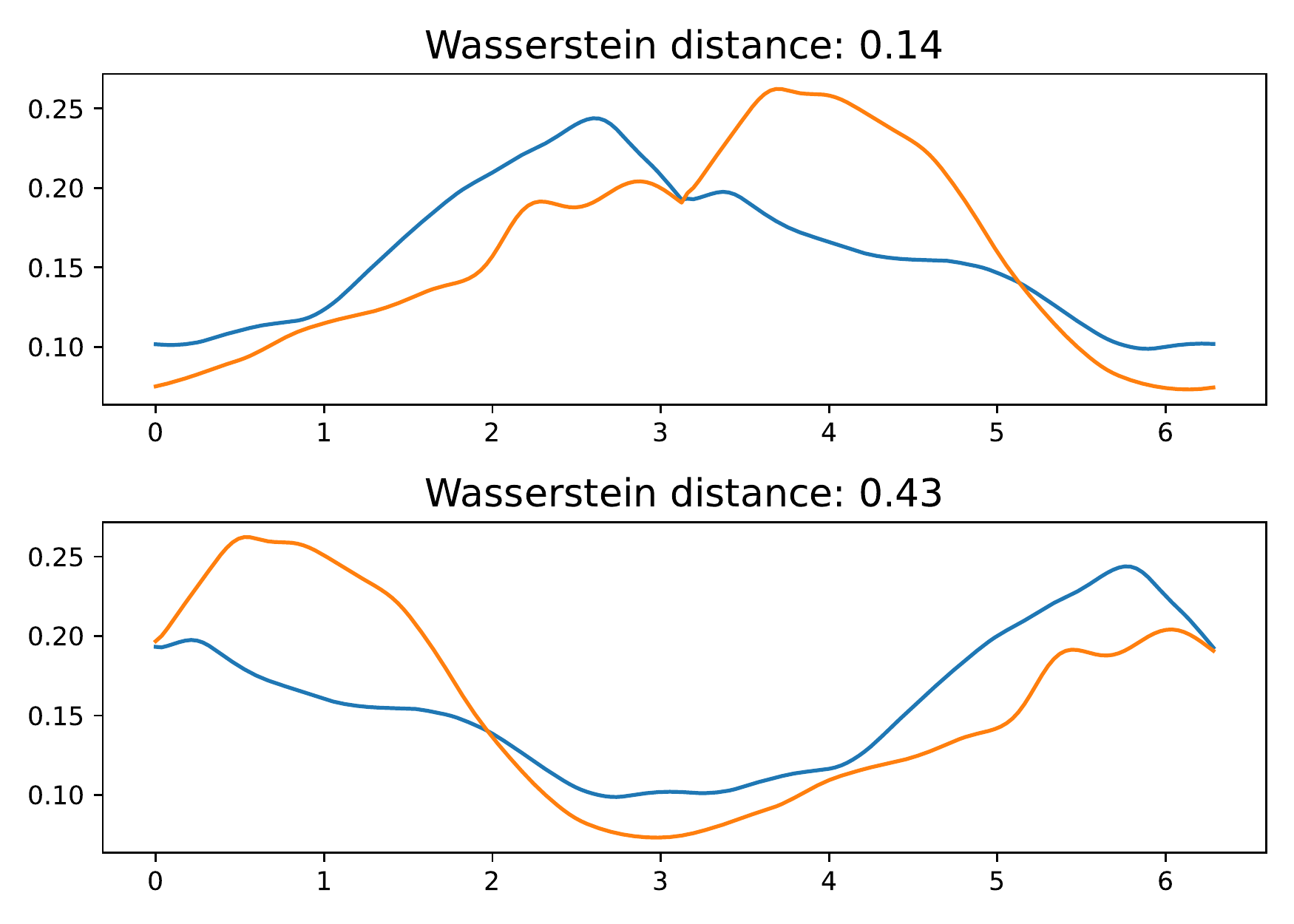}
		\caption{Two OCT samples on $\S_1$ (left) and when unrolled on $[0, 2\pi]$ (left) starting from 0 (top) or from $\pi$ (bottom) and the associated Wasserstein distances computed between the probability measures on $[0, 2\pi]$.}
		\label{fig:ex_data}
\end{figure}

\subsection{Related Works}

PCA for probability measures has been framed in different contexts, but, to the best of our knowledge, the focus has been either on analysing histograms (or discrete measures) or measures supported in $\mathbb R$.

Different definitions of PCA (and related algorithms) for distributions under the Wasserstein metric have been proposed in \cite{geodesic}, \cite{geod_vs_log} and \cite{projected}. 
In these works, the space of square-integrable probability measures on the real line, endowed with the $2$-Wasserstein metric (also called the Wasserstein space), is considered in close analogy to a ``Riemannian'' manifold and the characterisation of the tangent space at an absolutely continuous probability measure \citep{ambrosio2008gradient} is exploited to perform statistical analysis.

When the statistical units are not embedded in a linear space, classical tools from multivariate statistics need to be generalised to take into account the nonlinearity of the space. Think, for instance, on how the Frech\'aet mean generalises the notion of sample mean.
For data supported on manifolds, the statistical tools can be subdivided into \emph{extrisinc} or \emph{intrinsic} \citep{extrinsic, intrinsic, pennec_manifolds, huck_geod, patra, geod_regression_flet, geod_regression_chakra}. 
The extrinsic approach consists of finding a linear space (usually a tangent space at a suitable centring point) that approximates the manifold (or the region of the manifold where data are located), and performing standard (Euclidean) PCA on the projection of data onto the linear space.
In the intrinsic case, instead, the geodesic structure of the manifold is exploited to define a PCA based on the distance between datapoints and (geodesically) convex subsets of the manifold, whereby one considers convex subsets as the natural generalisation of linear subspaces. Note that extrinsic techniques introduce an approximation that might significantly impact the results if the manifold is not well approximated, while intrinsic techniques are usually computationally intensive and not suitable to analyse large datasets.

Focusing on the case of data in the 2-Wasserstein space of measures supported on $\R$, we can label the the \emph{geodesic}-PCA in \cite{geodesic} as an intrinsic method, while the \emph{log} PCA in \cite{geod_vs_log} and the \emph{projected} one in \cite{projected} are extrinsic tools.
These approaches are based on the explicit knowledge of optimal transport maps from an absolutely continuous measure to any other measure, which is a peculiarity of this particular setting. Moreover, \cite{geodesic, geod_vs_log, projected} exploit well known isometric isomorphisms between the 2-Wasserstein space and closed convex cones in suitably defined $L_2$ spaces. Thus, the ``manifold'' nature of the space of probability measure is taken into account by considering the ``cone constraints''.
The \emph{log}-PCA in \cite{geod_vs_log} can be, in principle, applied to distributions over more complex domains. 
However, as discussed in \cite{projected}, the \emph{log}-PCA results in poor interpretability of the components and does not allow to work on the scores, which is usually a standard requirement for PCA.

\subsection{Our Contribution and Outline}

In extending the previously proposed approaches for Wasserstein PCA to measures on $\S_1$ we face several nontrivial issues.
These have to do with the non-Euclidean nature of $\S_1$, which cannot be ignored. Indeed, consider the following example: fix a point $\theta \in \S_1$ and ``unroll'' the circle starting from $\theta$, which results in a bijection between $\S_1$ and $[0, 2\pi]$. It might be tempting to treat distributions on $\S_1$ as distributions on an interval of the real line. 
However, the Wasserstein metric is then dependant on the chosen $\theta$, as shown for instance in Figure \ref{fig:ex_data}.

Optimal transport for measures supported on Riemannian manifolds is an active area of research \citep{mccann_polar, gigli2011inverse,kim_barycenter}.
In particular, \cite{mccann_polar} provides a characterisation of optimal transport maps while \cite{gigli2011inverse} proposes different definition of tangent spaces based on the transport maps and plans.
Due to the generality of their framework, the resulting expressions are not amenable for computations.

The first main contribution of this paper is to provide a detailed investigation of optimal transport for measures supported on $\S_1$. In particular, we derive an expression for optimal transport maps in (almost) closed form and propose an alternative definition of tangent space at any absolutely continuous probability measure.
Contrary to the general definition in \cite{gigli2011inverse}, it is possible to characterise explicitly the image of the ``logarithmic map'' (i.e., the map from the Wasserstein to the tangent spaces). Moreover, we establish an homeomorphism between the Wasserstein space and the image of the logarithmic map.

Then, building on these important results, we propose a framework for PCA for measures on $\S_1$.
Our approach consists in choosing a suitable tangent space at a point $\bar \mu$, and analyse the transformed data obtained by mapping the observations to the tangent via the logarithmic map.
The tangent space is a Hilbert space, so that standard PCA could be carried out on the transformed data. However, the image of the logarithmic map is a convex cone inside the tangent space. We argue that such a constraint should be considered when performing PCA to obtain interpretable results. 
Indeed, as discussed in \cite{geod_vs_log,projected}, failing to do so results in poor interpretability of the directions, and the impossibility to work in the scores. Essentially, both issues are due to the fact that the principal directions might not be orthogonal (or even geodesics) when seen as curves in the Wasserestein space.
Following \cite{geodesic}, we propose a nested PCA, that requires solving a variational problem over the space of probability measures to find the principal directions.
Introducing a suitable B-spline approximation, we show how such an optimisation problem can be translated into a finite-dimensional constrained optimisation problem, whose solution can be approximated numerically using standard software for constrained optimisation. 

Finally, we discuss an algorithm to approximate the Wasserstein barycentre and propose to use the output of such an algorithm as the centring point $\bar \mu$ for the PCA.
Our algorithm follows the one in \cite{zemel2019frechet}, which requires explicit knowledge of optimal transport maps. We derive a sufficient \emph{a posteriori} condition to assess its convergence to the barycentre, and validate it on several simulations, leaving a theoretical analysis for future works.

The paper is structured as follows. In \Cref{sec:opt_on_manifolds} we cover the necessary background material on optimal transport. \Cref{sec:opt_s1} contains the main results related to optimal transport for measures on $\S_1$ and \Cref{sec:PCA} discusses our PCA framework and the numerical approximation of the Wasserestin barycentre. Numerical illustrations are presented in \Cref{sec:numerical_ill}, where we discuss a simulation study for the PCA on location-scale families of distributions, highlighting the differences between the case of measures on $\mathbb \R$ and $\mathbb \S_1$.
In \Cref{sec:eye} we present our analysis of the OCT measurements. Finally, we conclude the paper with a discussion on open problems and future work in \Cref{sec:conclusion}.
Proofs, further background material, and complementary results are deferred to the appendix.

\section{Background on Optimal Transport and on Manifold-valued Data Analysis}
\label{sec:opt_on_manifolds}

In this section, we provide a brief account of optimal transport and the Wasserstein distance for measures on compact manifolds. See, e.g., \cite{ambrosio2008gradient} for a detailed treatment. Technical details are deferred to Appendix \ref{sec:app_preliminaries}.

\subsection{Riemannian Manifolds.}
Informally, one can think of an $n$-dimensional smooth manifold $M$ as a set which locally behaves like a Euclidean space: it can be covered with a collection of open sets $(U_i)_{i \geq 1}$ for which there exist homeomorphisms $\varphi: U_i \rightarrow \varphi(U_i) \subset \mathbb R^n$, called coordinate charts, which satisfy some compatibility conditions.
We may refer to $(U_i, \varphi(U_i))$ as a \emph{local parametrisation} of the manifold.
A Riemannian manifold $(M, g)$ of dimension $n$ is a smooth manifold $M$ endowed with (a smooth family of) inner products $g = (g_x)_{x \in M}$ on the tangent space $T_x M$ at each point $x \in M$.
Its tangent bundle $TM$ is defined as 
\begin{equation}\label{eq:tang_bundle}
	TM := \coprod_{x \in M} T_x M =  \bigcup_{x \in M} \{x\} \times T_x M.
\end{equation}
Each $T_x M$ is a vector space of dimension $n$. The tangent bundle is itself a smooth manifold of dimension $2n$ with a standard smooth structure. See \cite{lee_intro_manifolds} for an introduction to Riemannian manifolds.

The \emph{exponential} map at $z \in M$ denoted by $\exp_z : TM \rightarrow M$ allows us to map a tangent vector $v \in T_x M$ onto the manifold itself. Informally, $\exp_z(v)$ is the arrival point of the geodesic starting at $z$ with direction $v$ travelled for a unit of time.
The \emph{logarithmic} map $\log_z:  M \rightarrow TM$, where it is defined, satisfies $\exp_z \circ \log_z(x) = x$.
The inner product $g$ induces the volume measure $\omega$, which is locally (i.e., on a chart $(U, \varphi)$) given by
\begin{equation}\label{eq:vol_meas}
	\mathcal{L}_M(A)=\int_{\varphi(A)} \mid \text{det}(g(\varphi^{-1}(x)))\mid^{1/2} d \mathcal{L}(x)
\end{equation}
for any measurable $A \subset U$ and with $\mathcal{L}$ being the Lebesgue measure. See  \Cref{sec:app_preliminaries} for measure theoretical details.

\subsection{Wasserstein space.}
To define the Wasserstein metric, denote by $\mathcal{P}(M)$ the space of probability measures on $M$ and let $c: M \times M \rightarrow \R_+$ be a cost function. The $p$-Wasserstein distance between two probability measures on $M$, say $\mu$ and $\nu$,  is 
\begin{equation}\label{eq:w_dist}
	W_p(\mu, \nu)^p = \min_{\gamma \in \Gamma(\mu, \nu)} \int_{M \times M} c(x, y)^p d\gamma(x,y), \qquad \mu, \nu \in \mathcal{P}(M)
\end{equation}
where $\Gamma(\mu, \nu)$ is the set of all probability measures on $M \times M$ with marginals $\mu$ and $\nu$.
The existence of (at least one) optimal plan  $\gamma^o$ attaining the minimum in \eqref{eq:w_dist} is ensured if $c$ is lower semicontinuous \citep{ambrosio2008gradient}.
Definition \eqref{eq:w_dist} is due to Kantorovich and can be seen as the weak formulation of Monge's optimal transportation problem, i.e.
\[
	W_p(\mu, \nu)^p = \inf_{T: T\push\mu=\nu} \int_{M} c(x, T(x))^p d\mu(x)
\]
where $\push$ denotes the pushforward operator: $T\push \mu(A) = \mu(T^{-1}(A))$ for all measurable $A$.
It can be proven that when an optimal map exists, then this induces an optimal transport plan $\gamma^o = (\text{Id}_M, T) \push \mu$ and the two formulations are equivalent. However, there are several situations in which Monge's problem has no solution.

In the following, we will always consider the Riemannian distance $\dist(\cdot, \cdot)$ as cost function and set $p=2$.
We restrict our focus on measures in the 2-Wasserstein space, that is the subset of probability measures
\[
	\Wcal_2(M)=\Big\{\mu \in \mathcal{P}(M) :\, \int_M \dist(x,x_0)^2 d\mu(x)<\infty \text{ for every }x_0\in M\Big\}.
\] 
This ensures that Wasserstein distance is always finite.

\subsection{Geometry of the Wasserstein space.}
The Wasserstein space $(\Wcal_2, W_2)$ can be endowed with a weak Riemannian structure induced 
by the tangent spaces of $\Wcal_2$ at any absolutely continuous measure with respect to the volume measure \eqref{eq:vol_meas}.
As in the case of measures supported in $\R^n$, the tangent spaces are subset of $L^2$ spaces of vector-valued functions defined on the ground space (in this case, $M$). Their definition needs some further background.

Consider a vector field $v: M \rightarrow TM$ such that for every $z \in M$, $v_z := v(z) \in T_z M$. 
To be more precise, denote by $\pi$ the canonical projection map $\pi: TM \rightarrow M$, i.e. $\pi(z, v) = z \in M$, then $v$ must be such that
\[
	\pi\circ v = \text{Id}_M
\]
where $\text{Id}_M$ is the identity map on $M$. Let $S(M)$ be the collection of all such vector fields.
Then, for a measure $\mu \in \mathcal{P}(M)$ we can define $L^2_\mu$ as 
\begin{equation}\label{eq:l2_mu}
	L^2_\mu(M) = \Big\lbrace v \in S(M): \int g(v_z,v_z)^2 d\mu(z) <\infty  \Big\rbrace.
\end{equation}
See \Cref{sec:app_preliminaries} for further details.
For $v \in S(M)$ we can define the map $\exp(v): M \rightarrow M$ such that  $\exp(v)(z) := \exp_z(v_z)$ for $z \in M$.
With this notation, we can state a fundamental theorem in optimal transportation due to \cite{mccann_polar}.
\begin{theorem}[Characterisation of optimal transport plans]
	Let $\mu, \nu \in \mathcal W_2(M)$. If $\mu$ is absolutely continuous with respect to the volume measure \eqref{eq:vol_meas}, there exists a unique optimal transport plan $\gamma^o$ which has the form $\gamma^o = (\text{Id}_M, T) \push \mu$, where $T: M \rightarrow M$.
	Moreover, there exists a $\dist^2$-concave function $\phi$ such that $T = \exp(- \nabla \phi)$. 
\end{theorem}
The $\dist^2$-concavity condition is rather technical and not needed in the following, for this reason we report it only in \Cref{sec:app_preliminaries} of appendix, see \cite{gigli2011inverse} for further details. To make explicit the dependence of the transport map on the source and target measures, we will use notation $T_{\mu}^\nu$ to refer to the optimal transport map (OTM) from $\mu$ to $\nu$.

The existence and uniqueness of optimal transport maps suggest the following definition of tangent spaces \citep[Corollary 6.4 of ][]{gigli2011inverse}
\begin{equation}
\text{Tan}_\mu(\Wcal_2(M))=\overline{\lbrace 
v \in L^2_\mu(M) \mid \exists \varepsilon>0 :
 (\text{Id}_{M}, \exp (tv))\push \mu \text{ is optimal for } t\leq \varepsilon
\rbrace}^{L^2_\mu}
\end{equation}

As in the case of Riemannian manifolds, we can define the exponential and logarithmic maps that allow to move from the tangent space $\text{Tan}_\mu(\Wcal_2(M))$ to the Wasserstein space and vice versa.
\begin{equation}
	\begin{aligned}
	\exp_\mu&: L^2_\mu(M) \rightarrow\Wcal_2(M), \qquad \exp_\mu(v)=\exp(v)\push \mu \\
	\log_\mu&:\Wcal_2(M)\rightarrow L^2_\mu(M), \qquad \log_\mu(\nu) = v \text{ s.t. }  \exp(v) = T_\mu^\nu
	\end{aligned}
\end{equation}
This structure is usually referred to as the \emph{weak Riemannian structure} of $\Wcal_2(M)$.

\section{Optimal Transport on the circle}\label{sec:opt_s1}

In this section, we specialise the general theory outlined in Section \ref{sec:opt_on_manifolds} to the case of measures supported on the unit-radius circle.

\subsection{Geometry of $\mathbb{S}_1$}

For our purposes, it is convenient to define the unit-radius circle as
$\S_1 := \{z \in \mathbb{C}: \, \mid z \mid = 1 \}$, where $\mid \cdot \mid$ denotes the module of a complex number.
We first present the smooth (group) structure of $\S_1$ and then describe its Riemannian structure.

To endow $\S_1$ with a group structure, we start by considering the map $\exp_c:\mathbb{R}\rightarrow \S_1$ defined as $\exp_c(x)=e^{ i x}$, and the map $\log_c: \S_1 \rightarrow \R$ defined as $\log_c(z) = x\in [0,2\pi)$ such that $z=e^{i x}$. Note that $log_c$ is right inverse of $\exp_c$, i.e., $\exp_c\circ \log_c = \text{Id}_{\S_1} $. 
The exponential map $\exp_c$ is usually referred to as \emph{universal covering} of $\S_1$ \citep{munkres}.
Clearly, we take the multiplication between complex numbers as the group operation: $\cdot: \S_1 \times \S_1 \rightarrow \S_1$ given by $z \cdot w = \exp_c(\log_c(z) + \log_c(w))$. Informally speaking $\log_c(z)$ is the ``angle'' associated with the polar representation of $z$ and $\cdot$ is the sum of the angles.
It can be trivially seen that $(\S_1, \cdot)$ is a group and $\exp_c: (\R, +) \rightarrow (\S_1, \cdot)$ is a group morphism.

Through $\exp_c$ and $\log_c$ we can define the smooth structure of $\S_1$ by considering at each $z \in \S_1$ the map $\exp_z(x) := \exp_c(x + \log_c(z))$, that is the shifted version of the exponential map, and $\log_z(w) = y$ such that $y \in [-\pi/2, \pi/2)$ and $\exp_z(\log_z(w)) = w$.
Letting $V_z := \S_1 \setminus \{-z\}$, we have that for each $z \in \S_1$ the couple $(V_z, \log_z)$ is a coordinate chart. 
With this differential structure $\S_1$ is a Lie Group and its tangent bundle is $T\S_1 = \{(x,v)\mid x\in \S_1 \text{ and } v\in T_x\S_1\}\simeq \S_1\times \R$. We call $1$ the point $(1,0)$ which gives the neutral element in $\S_1$.

We consider the Riemannian metric $g$ is induced by the embedding $\S_1\hookrightarrow \mathbb{C}\simeq \mathbb{R}^2$, that is $g_z(x, y) = xy$ for $x, y \in T_z\S_1 \simeq \R$. 
This induces the arc-length distance $\dist(z, w) = \mid \log_c(z) - \log_c(w) \mid$. 
Note that $det(g)\equiv 1$, so that  $\mathcal{L}_{\S_1} = \exp_c \push \mathcal{L}$ or, equivalenty, $\log_c\push\mathcal{L}_{\S_1} = \mathcal{L}$.
Thus for any $f: \S_1 \rightarrow \R$
\begin{equation}
\int_{\S_1} f(z) d\mathcal{L}_{\S_1}(z) = \int_{-\pi/2}^{\pi/2}  f(\exp_c(x)) d\mathcal{L}(x).
\end{equation}
See Appendix \ref{sec:app_preliminaries} for further details.

\subsection{Optimal transport maps}
\label{sec:OTM_S1}

With the notation introduced in the previous section we now focus on the optimal transportation problem on $M = \S_1$ endowed with its Riemannian distance $d_R$. 

The fundamental observation is that a measure $\mu$ on $\S_1$ can be equivalently represented by a \emph{periodic} measure on $\mathbb{R}$ defined as $\mutilde(A) := \mu(\exp_c(A))$ for any measurable $A$, which entails $\mutilde(A) = \mutilde(A + p)$ for any $p \in  2\pi\mathbb{Z}$, where $A + p$ amounts to shifting all the points in $A$ by the amount $p$.
Then we define the ``periodic cumulative distribution function'' associated with $\mutilde$ as $F_\mutilde(x) = \mutilde([0, x))$ for $x \in [0, 2\pi]$ and extend it over $\R$ via the rule 
$F_\mutilde(x +2\pi) = F_\mutilde(x) + 1$. 
For $\theta \in \R$, let $F^\theta_\mutilde(x) = F_\mutilde(x) + \theta$ denote a vertical shift of the cumulative distribution function. 
Note that the measure induced by $F^\theta_\mutilde$ is independent from $\theta$ and is always $\mutilde$. This easily follows from, for instance, $\mutilde([a,b])=F^\theta_\mutilde(b)-F^\theta_\mutilde(a)=F_\mutilde(b)-F_\mutilde(a)$.

Denote with $\quant_\mutilde$ the associated quantile function, i.e., the (generalised) inverse of $F_\mutilde$. 
We have that $( F^\theta_\mutilde )^-(x)=\quant_\mutilde(x-\theta)$.
Thus, $\theta$ acts as a rotation of the quantiles around the circle, by a factor of $z_\theta^{-1}=\exp_c(-\theta)$.
Hence, the $0$-th quantile $(F^\theta_\mutilde)^{-}(0)$ is not $0$ but $z_\theta^{-1}$.
Equivalently, $F^\theta_\mutilde(y) = \mutilde([z_\theta^{-1}, y))$.

Exploiting results contained in \cite{delon2010fast}, the following theorem provides an explicit characterisation for the optimal transport maps between two measures on $\S_1$.
\begin{theorem}\label{teo:opt_map}
	Define $\theta^*$ as the solution of the following minimisation problem:
	\begin{equation}\label{eq:theta_cost}
		\theta^* = \argmin_{\theta \in \R} \int_0^1 \left(\quant_\mutilde(u)  - (F^\theta_\nutilde)^{-}(u)\right)^2 \dd u
	\end{equation}
	Then the optimal transport map between $\mu$ and $\nu$ is
\begin{equation}\label{eq:otm}
	T^\nu_\mu := \exp_c \circ \left((F^{\theta^*}_\nutilde)^{-} \circ F_\mutilde \right) \circ \log_c.
\end{equation}
\end{theorem}
Note that \eqref{eq:otm} is closely related to the expression of optimal transport maps for measures on $\R$. In that case, setting $\exp_c = \log_c = \mbox{Id}$ and $\theta^* = 0$ we recover the classical formulation of OTMs for measures on the real line.
In the following, we will write $\widetilde T_{\mutilde}^\nutilde := (F^{\theta^*}_\nutilde)^{-} \circ F_\mutilde$ to denote the map between $\mutilde$ and $\nutilde$ associated with the optimal $\theta^*$ in \eqref{eq:theta_cost}.
Although $\widetilde T_{\mutilde}^\nutilde$ is not ``optimal'' (since the cost associated to the transport of periodic measures is either zero or unbounded), we will refer to it as the optimal transport map between $\mutilde$ and $\nutilde$ in light with its connection with $T^\nu_\mu$.

Let us give some intuition behind the optimal transport map $T_\mu^\nu$.
Observe that precomposing $(F^{\theta^*}_\nutilde)^{-}$ with $\left( F_\mutilde \right)_{\mid[0,2\pi]}$, obtaining $\widetilde T_{\mutilde}^{\nutilde}$, means transporting quantiles identified by $F_\mutilde^-$ onto the corresponding shifted quantiles of $(F^{\theta^*}_\nutilde)^{-}_{\mid [0,1]}$, in an anti-clockwise order (due to the definition of $\exp_c$). 
Note that $T_{\mutilde}^{\nutilde}
((F_\mutilde)^{-}(0))=
T_{\mutilde}^{\nutilde}(0)= F^-_\nutilde(-\theta^*) =: x_{-\theta^*}$ and 
\begin{equation}\label{eq:OTM_shift}
T_{\nutilde}^{\mutilde}
((F_\mutilde)^{-}(1))\leq T_{\nutilde}^{\mutilde}
(2\pi)=  (F^{\theta^*}_\nutilde)^{-}
(1)=\quant_\nutilde(1-\theta^*)=2\pi+\quant_\nutilde(-\theta^*)=2\pi+x_{-\theta^*},
\end{equation}
which means that the optimal transport maps sends $[0, 2\pi)$ into $[x_{-\theta^*}, 2\pi + x_{-\theta^*})$. 
As a consequence we can think at this situation as \virgolette{unrolling} the circle in two different points, namely $z_{\theta^*}^{-1} = \exp_c(-\theta^*)$ for $\nu$ and $1 = \exp_c(0)$ for $\mu$, and then matching the measures induced on $\R$. 
For instance, suppose $\mu$ and $\nu$ have densities $f_\mu$ and $f_\nu$ with respect to the Lebesgue measure on $\S_1$, $\mathcal{L}_{\S_1}$, then $(F^\theta_\nutilde)^{-}_{\mid [0,1]}$ is the quantile function associated with the density $f_\nu(\exp_{c}(x))$ supported on $[x_{-\theta},2\pi+x_{-\theta}]$. Clearly no action is taken on $\mu$ and thus we transport $f_\mu(\exp_{c}(x))$ supported on 
 $[0,2\pi]$ onto $f_\nu(\exp_c(x))$ supported on $[x_{-\theta},2\pi+x_{-\theta}]$.      
The parameter $\theta^*$ then selects the optimal point from which to start unrolling the circle for $\nu$.

Optimal transport maps are fundamental for the statical methods we develop in the later sections: the optimal transport maps $T_i$ from a reference distribution to the $i$-th datapoint will play the role of \virgolette{tangent vectors}, allowing us to approximate the Wasserstein space, with a space of functions.
Thus, it is essential to characterise the optimal transport maps on $\S_1$, understanding their properties, and inspecting them assuming the perspective of the associated maps $\widetilde T$ between periodic measures on $\R$.

The following theorem proves a fundamental property of OTMs.

\begin{theorem}\label{teo:opt_theta}
Given $\mu$ a.c. measure and $\nu\in\Wcal_2(\S_1)$, $\widetilde T:=(F^{\theta^*}_\nutilde)^{-} \circ F_\mutilde $ is an optimal transport map if and only if:
\begin{equation}\label{eq:int_condition}
\int_0^{2\pi} \widetilde T(u)-u \, \dd u =0 .
\end{equation} 
\end{theorem}

Comments on \Cref{teo:opt_theta} will follow throughout the manuscript as it impacts many of the upcoming definitions and results. Here we just point out that \Cref{eq:int_condition} is independent of the measure $\mu$ and is a purely analytical/geometric condition on $\widetilde T$.

\subsection{Weak Riemannian structure}\label{sec:riem}

As already mentioned, our aim is to exploit the weak-Riemannian structure of $\Wcal_2(\S_1)$ to obtain a more tractable representation of a data set of probability measure, which enables the use of statistical tools. Thus, we now specialise the definition of $\text{Tan}_\mu(\Wcal_2(M))$ and the associated exponential and logarithmic maps when $M \equiv \S_1$, translating the original vector-field definition in terms of more tractable functions. 
Furthermore, we establish properties of the logarithmic map that will be fundamental to develop a coherent statistical framework for analysing probability measures in $\Wcal_2(\S_1)$.

For our purposes, it is convenient to define $L^2_\mu(\S_1)$ as 
\begin{align*}
	L^2_\mu(\S_1) :&= \Big\{ v : \S_1 \rightarrow \R \text{ such that } \int_{\S_1} v^2(x) d\mu(x) < +\infty \Big\} \\
	&= \Big\{ v : [0, 2\pi) \rightarrow \R \text{ such that } \int_0^{2\pi} v^2(x) d\mutilde(x) < +\infty \Big\}
\end{align*}
where the second equality follows, with a slight abuse of notation, by considering $v \mapsto v \circ \log_c$.
Observe that we recover the space in \eqref{eq:l2_mu} by identifying $v(x)$ as an element of $T_x \S_1$.
Then, if $\mu$ is an absolutely continuous measure, we have 
\begin{equation}\label{eq:tanw_s1}
	\text{Tan}_\mu(\Wcal_2(\S_1)) = \overline{\{v : L^2_\mu(\S_1)  \mid \exists \varepsilon>0 :
 (\text{Id}_{\S_1}, \exp (tv))\push \mu \text{ is optimal for } t\leq \varepsilon \}}^{L^2_\mu}
\end{equation}
where we can interpret $v$ as a function defined on $\S_1$ or $[0, 2\pi)$ according to our needs.
Now we want to rewrite this definition to make it more easily tractable.

First, note that the optimality condition in \eqref{eq:tanw_s1} is equivalent to saying that there exist $\nu$ such that $\exp(tv)$ is an optimal transport map between $\mu$ and $\nu$.
Then, by \Cref{teo:opt_map} and the fact that $\exp_z(v_z) = \exp_c(\log_c(z) + v_z)$, the vector field $v$ in \eqref{eq:tanw_s1} can be written as $t v(\log_c(x)) = \Ttilde(x) - x$, where $\Ttilde$ is as in \Cref{teo:opt_map}, so that the OTM is $\exp_c(x + (\Ttilde(x) - x)) \equiv \exp_c(\Ttilde(x))$.
Hence, we can restate the definition of tangent space in terms of the maps $\Ttilde$ as:
\begin{equation}\label{eq:tanw_s1_v2}
	\text{Tan}_\mu(\Wcal_2(\S_1)) = \overline{\{\Ttilde : L^2_\mutilde([0, 2\pi])  \mid \exists \varepsilon>0 :
  \exp_c( \text{Id}  + t (\Ttilde - \text{Id})) \text{ is OTM for } t\leq \varepsilon \}}^{L^2_\mutilde}
\end{equation}
The definition of exponential and logarithmic map comes quite naturally:
\begin{equation}\label{eq:log_exp}
	\begin{aligned}
	\exp_\mu&: L^2_\mu(\S_1) \rightarrow\Wcal_2(\S_1), \qquad \exp_\mu(\Ttilde)= \left(\exp_c \circ \Ttilde \circ \log_c  \right)\push \mu \\
	\log_\mu&:\Wcal_2(\S_1)\rightarrow L^2_\mu(\S_1), \qquad \log_\mu(\nu) = \Ttilde \text{ s.t. }   \Ttilde(x) = \quant_\nutilde(F_\mutilde(x) - \theta^*)
	\end{aligned}
\end{equation}
where $\theta^*$ in the definition of the $\log_\mu$ map is as in \Cref{teo:opt_map}. Observe that $\exp_c \circ \Ttilde \circ \log_c$ is an OTM between $\mu$ and $\nu$.
Furthermore, from Theorem \ref{teo:opt_theta} we note that the vector field $v: [0, 2\pi) \rightarrow \R$ induced by an optimal transport map $\Ttilde$ (i.e. $v(u) = \Ttilde(u) - u$) satisfying \eqref{eq:int_condition} has zero mean when integrated along $\S_1$ with respect to $\mathcal L_{\S_1}$.
In particular, note that this condition does not depend on $\mu$ and gives a purely geometric characterisation of optimal transport maps. This is in accordance to other typically used optimality conditions such as cyclical monotonicity of the support of the transport plan and Brenier's characterisation of OTMs for measures on $\R^n$ \citep{ambrosio2008gradient}.

We now provide some further characterisations of the optimal transport maps in light of the pieces of notation we have just introduced. 
These will be useful to investigate the map $\log_\mu$ and implementation of numerical algorithms.
\begin{theorem}\label{teo:opt_map2}
	Given $\mu$ a.c. measure, $\widetilde T: \R\rightarrow \R$ induces an optimal transport map between $\mu$ and $\nu := \exp_c \circ \widetilde{T} \circ \log_c \push \mu$ if and only if 
	\begin{itemize}
		\item $\widetilde{T}$ is monotonically nondecreasing with $\widetilde{T}(x + p) = \widetilde{T}(x) + p$ for all $p \in 2\pi\mathbb{Z}$ 
		\item $\widetilde{T}$ satisfies \eqref{eq:int_condition}
		\item $|\widetilde{T}(x) - x| < \pi$  $\mu$-a.e.
	\end{itemize}
\end{theorem}
From the previous result, it is immediate to prove the following.
\begin{corollary}
	Let $\mu$ be an a.c. measure on $\S_1$. Then the image of $\log_\mu$ defined in \eqref{eq:log_exp} is a convex set.
\end{corollary}

Moreover, the following proposition establishes the continuity of both $\exp_\mu$ and $\log_\mu$.

\begin{theorem}\label{teo:continuity}
	Let $\mu$ be an a.c. measure on $\S_1$. Then: 
	\begin{enumerate}
		\item for any $\nu_1, \nu_2 \in \Wcal(\S_1)$
	\[
		W^2_2(\nu_1,\nu_2)\leq \int_{\S_1}d^2_R(T^{\nu_1}_\mu,T^{\nu_2}_\mu)d\mu 
\leq \parallel \log_\mu(\nu_1)-log_\mu(\nu_2) \parallel^2_{L^2_\mu}.
	\]
	In particular, the $\exp_\mu$ map is continuous;
	
	\item if $W_2(\nu,\nu_n)\rightarrow 0$ in $\Wcal_2(\S_1)$ then
	\[
		\parallel \log_\mu (\nu_n)-\log_\mu (\nu)\parallel_{L^2_\mu} \rightarrow 0,
	\]
	\end{enumerate}
	that is, the $\log_\mu$ map is sequentially continuous. 
	As a consequence, since in metric spaces sequential continuity is equivalent to continuity, $\Wcal_2(\S_1)$ and $\log_\mu(\Wcal_2(\S_1))$ are homeomorphic via $\log_\mu$ and $\exp_\mu$.
\end{theorem}

We present also another proof of \Cref{teo:continuity}, item $2.$. To be more precise, it is a proof for a weaker result, but which we believe can be generalised to other compact Riemannian manifolds, on the contrary of the proof of \Cref{teo:continuity}, item $2.$.

\begin{proposition}\label{prop:weaker_cont}
    Let $\sigma$ be an a.c. measure and $\{\mu_t\}_t$ be a sequence of a.c. measures such that $\mu_t \rightarrow \mu_0$ (in the Wasserstein metric) as $t \rightarrow 0$. Further assume that the support of $\sigma$ and $\mu_t$ is (geodesically) convex and their density is bounded from above and strictly greater than zero.
Then	$\|\widetilde{T}_{\sigma}^{\mu_t} - \widetilde{T}_{\sigma}^{\mu_0}\| \rightarrow 0.$
\end{proposition}

We highlight that \Cref{teo:continuity} ensures that there is a high level of coherence between the measures in 
$\Wcal_2(\S_1)$ and their representation via $\log_\mu(\Wcal_2(\S_1))$. It is not an isometric representation as in the case $\Wcal_2(\R)$ (see \cite{projected}), but the continuity of the exponential and logarithmic maps implies a high level of interpretability.

\section{PCA for Measures on $\mathbb{S}_1$}
\label{sec:PCA}

In this section, we demonstrate how the results obtained in Section \ref{sec:opt_s1} can be leveraged to develop a principal component analysis framework for measures on $\S_1$ in an extrinsic fashion, by considering  $\mu_1, \ldots, \mu_n \in \Wcal_2(\S_1)$ in analogy to points of a Riemannian manifold, cf. Section \ref{sec:riem}.
This parallelism was first exploited to perform inference on the Wasserstein space in \cite{geodesic, geod_vs_log, projected} to develop a PCA for probability measures on the real line, and in \cite{muller} and \cite{zhang2020wasserstein} who propose linear regression and autoregressive models for measures on $\R$ respectively.

As already mentioned in the introduction, in the case of measures on the real line, the weak Riemannian structure of the Wasserstein space allows the definition of both intrinsic and extrinsic techniques \citep{geodesic, geod_vs_log, muller, zhang2020wasserstein, projected}. 
In particular, since $\Wcal_2(\R)$ can be seen as a convex cone inside a suitably defined $L_2$ space (by identifying each measure with the associated optimal transport map),  intrinsic methods simply need to take into account the ``cone constraints'' \citep{projected}.
As noted above,  such a cone representation does not hold in the case of $\Wcal(\S_1)$.
Therefore, developing intrinsic methods would require working with curves of probability measures.
We believe that the results established in Section \ref{sec:opt_s1} could be the first building block of such intrinsic methods.
However, supported by the continuity result in item (3.) of Theorem \ref{teo:continuity}, we propose a  \emph{log} PCA, which is computed after mapping all datapoints onto a suitable tangent space. In fact, the continuity results suggest that the approximation we make when mapping data to the tangent space is not too coarse, or, at least, should always produce interpretable results. 
The numerical illustrations presented in Section \ref{sec:numerical_ill} seem to validate this claim. 

\subsection{Log Convex PCA on $\Wcal_2(\S_1)$}

\begin{figure}[t]
	\centering
	\includegraphics[width=0.6\linewidth]{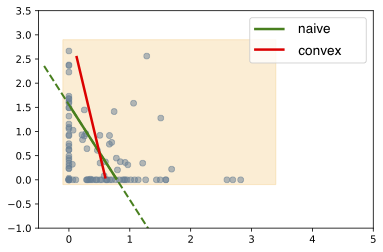}
	\caption{First principal direction found by the naive $L_2$ and the convex PCA when the space $H = \R^2$ and $X$ is the yellow rectangle. The blue dots denote observations.}
	\label{fig:naive_vs_conv}
\end{figure}

As shown in Corollary 6.6 of \cite{gigli2011inverse}, the tangent space at absolutely continuous measures is Hilbert so that we could apply standard PCA techniques to $\log_{\mubar}(\mu_1), \ldots, \log_{\mubar}(\mu_n)$, for some fixed measure $\mubar$. We call this approach  ``naive'' log-PCA.
However, as argued in \cite{projected}, disregarding the fact that the image of the $\log_{\mubar}$ map is not the whole $\text{Tan}_{\mubar}(\Wcal_2(\S_1))$ tangent space, but only a convex subset, might produce misleading results.
In particular, when two elements of the tangent space lie outside the image of $\log_{\mubar}$, returning to the Wasserstein space and then back to the tangent via $\log_{\mubar} \circ \exp_{\mubar}$ can produce undesired behaviours in terms of distances and angles. 
More in general, a principal direction is interpretable and captures meaningful variability only as long as it lies inside 
the convex subset.
This fact undermines, for instance, the interpretability of scores and principal directions when they lie outside $\log_{\mubar}(\Wcal_2(\S_1))$: directions may not the orthogonal and variance inside $\Wcal_2(\S_1)$ may not be decomposed appropriately. 

To avoid the problems with the ``naive'' log-PCA, we propose the following definition of log convex PCA, which amounts to performing a convex PCA \citep{geodesic} in the tangent space, thus taking into account the constraints enforced by the image of the $\log$ map.
Let us introduce some notation first.
Let $X := \log_{\mubar}(\Wcal_2(\S_1))$, $H := \text{Tan}_{\mubar}(\Wcal_2(\S_1))$.
For a closed convex set $C \subset X$ and a point $x \in X$ let $d(x, C) = \argmin_{y \in C} \|x - y\|_{L^2_{\mubar}}$  
Let $Sp$ denote the span of a set of vectors and $\mathcal{C}_{x_0}(U) := (x_0 + Sp(U)) \cap X$ for $x_0 \in X$ and $U \subset H$.

As in \cite{projected}, we also make the following technical assumption: given a collection of probability measures $\mubar, \mu_0, \mu_1, \ldots, \mu_n \in \Wcal(\S_1)$ we assume that $\log_{\mubar}(\mu_0)$ lies in the relative interior of the convex hull of $\{\log_{\mubar}(\mu_i)\}$. The most common choice for  $\mu_0$ is to be chosen as the \virgolette{mean} of $\{\log_{\mubar}(\mu_i)\}$, which, being inside an Hilbert space, could violate our assumption in some pathological cases. However, in applications we always resort to a finite-dimensional approximation of $L^2_{\mubar}$, in which the assumption is always satisfied. For more details see Appendix A in \cite{projected}.    

\begin{definition}\label{def:PCA}
	Consider a collection of probability measures $\mubar, \mu_0, \mu_1, \ldots, \mu_n \in \Wcal(\S_1)$.
	Let $\Ttilde_i = \log_{\mubar}(\mu_i) = \Ttilde_{\mubar}^{\mu_i}$, $i=0, \ldots, n$.
	A $(k, \mubar, \mu_0)$ log convex principal component for $\mu_1, \ldots, \mu_n$ is the subset $C_k := \mathcal{C}_{\Ttilde_0}(\{w^*_1, \ldots, w^*_k\})$ such that
	\begin{enumerate}
		\item for $k=1$, 
		\[
			w^*_1 = \argmin_{w \in H, \|w\|=1} \sum_{i=1}^n d\left(\Ttilde_i, \mathcal{C}_{\Ttilde_0}(\{w\})\right)
		\]
		\item for $k > 1$,
		\[
			w^*_k = \argmin_{w \in H, \|w\| = 1, w \perp Sp(\{w^*_1, \ldots, w^*_{k-1}\}} \sum_{i=1}^n d\left(\Ttilde_i, \mathcal{C}_{\Ttilde_0}(\{w\})\right)
		\]
	\end{enumerate}
\end{definition}  

\Cref{fig:naive_vs_conv} exemplifies the difference between the naive $L_2$ and the convex one in a simpler example when $H = \R^2$ and $X$ is a convex subset.
When data are close to the border of $X$, the $L_2$ metric between data and the principal components capture variability that lies outside of the convex set. See also \cite{projected} for some indexes that quantify the loss of information of the $L_2$ PCA opposed to the convex one.

\subsection{Computation of the Log Convex PCA via B-Spline approximation}

The definition of convex PCA translates into a constrained optimisation problem to find the directions $\{w^*_1, \ldots, w^*_k\}$.
In \cite{geod_vs_log}, the authors discretise the transport maps and solve the optimisation problem via a forward-backward algorithm. As discussed in \cite{projected}, a more efficient approach consists in approximating the transport maps via quadratic B-splines and solving a constrained optimisation problem via an interior-point method.
Here, we follow the second approach.

Let $\{\psi_1, \ldots, \psi_J\}$ a B-spline basis on equispaced knots in $[0, 2\pi]$. We let $\Ttilde_i(x) \approx \sum_{j=1}^J a_{ij} \psi_j(x)$. Note that if the spline is quadratic then 
(i) the function $\sum_{j=1}^J a_{j} \psi_j(x)$ is monotonically nondecreasing if an only if the coefficients $a_1, \ldots, a_J$ are \citep[see, e.g., Proposition 4 in][]{projected}.
Hence, from now on, we consider the $\psi_j$'s to be quadratic spline basis functions on $[0,2\pi]$.
The spline basis expansion also allows for faster computations of $L_2$ inner products: let $E$ be a $J \times J$ matrix with entries $E_{i, j} = \int_0^{2\pi} \psi_i(x) \psi_j(x) \dd x$ and $\bm a_i = (a_{i, 1}, \ldots, a_{i, J})$, we have $\langle \Ttilde_i, \Ttilde_j \rangle = \langle \bm a_i, \bm a_j \rangle_E := \bm a_i^T E \bm a_j$. We denote by $\| \cdot \|_E$ the associated norm.

Similarly to Proposition 6 in \cite{projected}, we obtain that the $k$-th direction $\bm w_k$ and the associated scores $\lambda^k_{1:n} = \lambda_1, \ldots, \lambda_n$ (of the observations the $k$-th direction) of the log-convex PCA can be computed by solving a constrained optimisation problem.
The objective function is:
\begin{equation}\label{eq:log_pca_obj}
	\lambda^k_{1:n}, \bm w_k = \argmin_{\lambda_{1:n}, \bm w} \sum_{i=1}^n  \| \bm a_i - \bm a_0 - \sum_{j=1}^k \lambda_i^k \bm w_k \|
\end{equation}
where $\lambda_i \in \mathbb{R}$ is the of score for the $i$-th datum along the $k$-th direction.
Moreover, the usual orthogonality and unit-norm constraints must be satisfied:
\[
	\|\bm w\|_E = 1, \quad \langle \bm w_h, \bm w \rangle_E = 0, \quad h=1, \ldots, k-1.
\]
In addition to those, we must also require that $\sum w_j \psi_j$ belongs to $H := \text{Tan}_{\mubar}(\Wcal_2(\S_1))$. 
The monotonicity constraint is equivalent to
\[
	\lambda_i w_{j} + a_{0,j} - \lambda_{i} w_{j-1} - a_{0, j-1} \geq 0 , \quad j=2 \ldots J
\] 
that is the monotonicity of the spline coefficients (since the splines are quadratic. See, e.g., Proposition 4 in \cite{projected}). 
Moreover, the ``periodicity'' constraint is satisfied by design.
To impose \eqref{eq:int_condition}, let $M_j = \int \psi_j(u) \mathrm{d} u$, then \eqref{eq:int_condition} is equivalent to
\[
	\sum w_j M_j = 2 \pi^2.
\] 
Finally, thanks to \eqref{eq:int_condition} it is sufficient to control the value of the function $w$ at the initial point, i.e. $w_0 \in (-\pi/2, \pi/2)$.

We implement the resulting constrained optimisation problem using the \texttt{Python} package \texttt{pyomo} and approximate the solution using an interior point method using the Ipopt solver.

\subsection{Wasserstein Barycentre}

We are left to discuss the choice of the base point $\mu_0$ of the PCA as well as the measure $\bar \mu$ at which the tangent space is considered.
A standard choice when performing PCA in non-Euclidean spaces, it to set both $\mu_0$ and $\bar \mu$ equal the barycentre, that is the Fr\'echet mean. In our case, the barycentre minimises the following Fr\'echet functional:
\begin{equation}\label{eq:frechet_func}
	F(\nu; \mu_1, \ldots, \mu_n) = \frac{1}{2n} \sum_{i=1}^n W_2^2(\nu, \mu_i).
\end{equation}
While, in principle, the log-PCA can be carried out by working in the tangent at any absolutely continuous measure, embedding the PCA in the tangent at the barycentre is to be preferred since, intuitively, this should result in the distances between datapoints in the tangent space (at the barycentre) to be more similar to the distances in the Wasserstein space. The quality of the approximation provided by tangent spaces decays as distances from the tangent point increase, and thus choosing as a tangent point the barycentre of the data set is a good choice for trying to minimise the average error produced by the approximations. As a consequence, the projections of the principal components can be interpreted as deviations from the \virgolette{average} of the data set.
Note that centring the PCA at the barycentre poses no conceptual problem in our case as the Wasserstein barycentre is unique if at least one of the measures $\mu_j$ is absolutely continuous. See Theorem 3.1 in \cite{kim_barycenter}. Similar results for measures supported on $\mathbb{R}^d$ have been developed in  \cite{agueh_barycenter}.

\SetNlSty{textbf}{[}{]}
\begin{algorithm}[t]
	\textbf{input}{ Measures $\mu_1, \ldots, \mu_n$, starting point $\nu$, threshold $\varepsilon$.}
	\DontPrintSemicolon
	
	\Repeat{$W_2(\nu, \nu^\prime) < \varepsilon$ }{
		Compute the optimal transport maps $\widetilde T_{\nu}^{\mu_i}$ as in \Cref{teo:opt_map}. \\
		Set 
		\[	\widetilde{\nu}^\prime := \left( \frac{1}{n} \sum_{i=1}^n \widetilde T_{\bar \mu}^{\mu_i} \right) \push \widetilde{\nu} \]
	}
	
	Output $\bar \mu = \exp_c \circ (\widetilde{\nu}^\prime)$. \\
	\textbf{end}
	\caption{\label{algo:barycentre}Procrustes Barycentre}
\end{algorithm}

Numerical algorithms for computing the solution of \eqref{eq:frechet_func} have been developed in \cite{carlier2015numerical,srivastava15} for the case of atomic measures, whereby the optimisation can be reduced to a linear program.
\cite{zemel2019frechet} instead propose a procustes  algorithm based on gradient descent which works for general measures on $\mathbb{R}^d$ (of which one must be absolutely continuous). In a nutshell, the gradient descent algorithm in \cite{zemel2019frechet} starts from an initial guess of the barycentre and updates it by pushing forward the current guess $\nu_r$ via the average of the transport maps between $\nu_r$ and all the measures.
This procedure is guaranteed to converge to the barycentre under some technical conditions on the measures $\mu_i$'s. In particular, it converges
in one iteration if the measures are \emph{compatible} \citep[see Section 2.3.2 in][]{panaretos}. As a drawback, this approach requires solving $n$ optimal transportation problems at each iteration, which might be challenging outside the case of measures supported on $\R$ or location-scatter families, for which explicit solutions exist \citep{alvarez2018wide}.
Taking a different approach, \cite{cuturi2014fast} propose an approximate solution to the Fr\'echet mean by introducing in \eqref{eq:frechet_func} an ``entropic regularisation'' term, which makes optimisation easier.

Here, we propose to use  the gradient descent algorithm developed in \cite{zemel2019frechet}. Indeed, our \Cref{teo:opt_map} allows for (almost) explicit solutions to the optimal transportation problem. Moreover, as shown in \cite{delon2010fast}, the optimisation problem in \eqref{eq:theta_cost} is convex in $\theta$ so that finding $\theta^*$ is simple.
We report the pseudocode in Algorithm~\ref{algo:barycentre}.

We want to remark that we have not been able (yet) to prove neither the convergence of the algorithm to the barycentre in the general case nor if such procustes algorithm amounts to a gradient descent also in our framework.
From the technical point of view, the proofs in \cite{zemel2019frechet} do not hold in our case, since they are based on sub-differentiability and super-differentiability results of the Wasserstein distance as provided in Theorems 10.2.2 and 10.2.6 in \cite{ambrosio2008gradient} which are stated for measures on separable Hilbert spaces.
Nonetheless, the following result establishes a sufficient condition for the convergence of Algorithm~\ref{algo:barycentre}.
\begin{proposition}\label{prop:conv_bary}
	Let $\mu^*$ be an absolutely continuous measure in $\Wcal(\S_1)$, and $\mu_1, \ldots, \mu_n$ be measures in $\Wcal(\S_1)$. If, for any $i, j = 1, \ldots, n$
	\[
		\|\log_{\mu^*}(\mu_i) - \log_{\mu^*}(\mu_j)  \|_{L^2_{\mu^*}} = W_2(\mu_i, \mu_j),
	\]
	then letting $\bar T := n^{-1} \sum_{i=1}^n T_{\mu^*}^{\mu_i}$ be the barycentre of the $\log_{\mu^*}(\mu_i)$'s, we have that $\bar T \push \mu^*$ is the Wasserstein barycentre of $\mu_1, \ldots, \mu_n$.
\end{proposition}
The condition in \Cref{prop:conv_bary} has the practical advantage that it can be easily checked after Algorithm~\ref{algo:barycentre} terminates. Indeed, if $\|\log_{\bar \mu}(\mu_i) - \log_{\bar \mu}(\mu_j)  \|_{L^2_{\mu^*}} = W_2(\mu_i, \mu_j)$, where $\bar \mu$ is the output of Algorithm \ref{algo:barycentre}, we are sure that $\bar \mu$ is the barycentre.
Intuitively, if the Wasserstein distances are similar to the distances in the tangent space, this means that, along the geodesics connecting the datapoints, the curvature is small. Hence, the problem of finding the Wasserstein barycentre reduces to averaging the quantiles. Therefore, the output of Algorithm \ref{algo:barycentre} should be accurate.
In the following section we provide empirical evidence of its convergence, by checking the  condition in \Cref{prop:conv_bary} and comparing the output of Algorithm \ref{algo:barycentre} to the one of the Sinkhorn algorithm proposed in \cite{cuturi2014fast}.

\begin{remark}
	Although stated for measures on $\S_1$, \Cref{prop:conv_bary} is true for measures on general connected compact finite dimensional Riemannian manifolds whose exponential map is non-expansive. This is the case, for instance, of manifolds with positive curvature. In \Cref{app:conv_vary} we prove the result in this more general setting.

\end{remark}

\section{Numerical Illustrations}\label{sec:numerical_ill}

In this section we present the numerical simulations dealing with the Wasserstein barycentre and the PCA defined in \Cref{sec:PCA}.

\subsection{Simulations for the Barycentre}\label{sec:simu_bary}

\begin{figure}
	\centering
	\includegraphics[width=\linewidth]{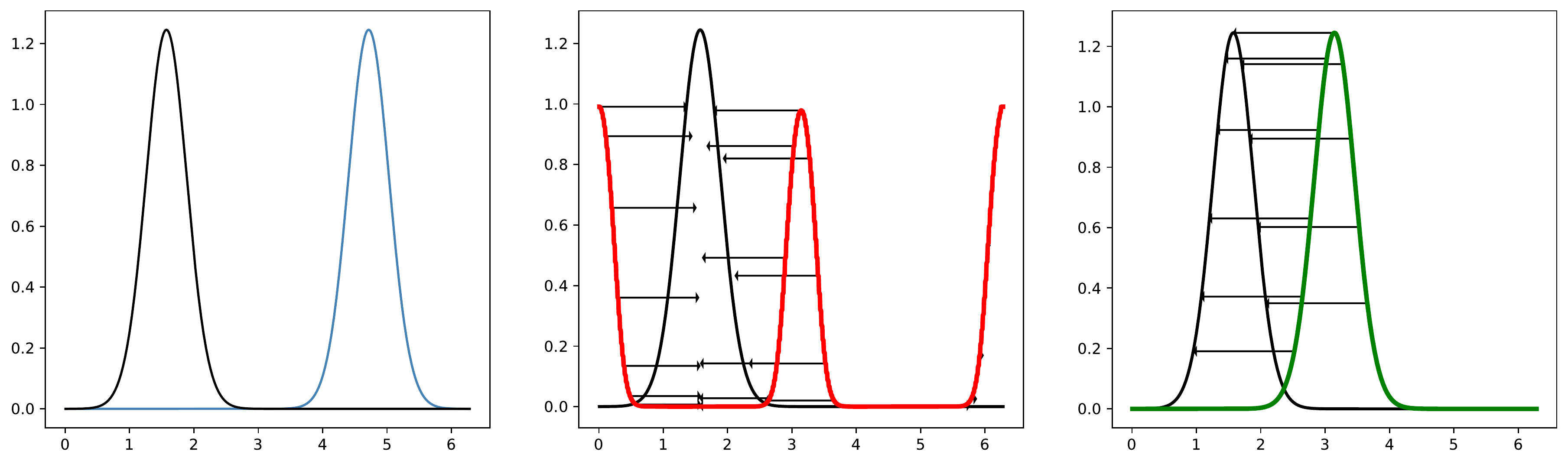}
	\caption{From left to right: two measures on $\mathbb{S}_1$ (unrolled on $[0, 1]$), the barycentre on $\mathbb S_1$ (red) and its transport to the leftmost measure, the barycentre on $\mathbb{R}$ and its transport to the leftmost measure}
	\label{fig:example_bary}
\end{figure}

Let us give an illustrative example of the peculiarities that may arise when considering distributions on $\mathbb{S}_1$.
Consider the two measures on the leftmost panel in Figure \ref{fig:example_bary}. When the transport cost is the Euclidean one, the resulting barycentre is the one displayed in the rightmost panel: it has unimodal density with the same scale of the two measures and is centred exactly in the middle of them. 
When the cost instead is computed on $\mathbb{S}_1$, the barycentre becomes bimodal as shown in the middle panel of Figure \ref{fig:example_bary}. In this specific example, the cost (on $\mathbb{S}_1$) of transporting the ``correct'' barycentre on the two measures is 30\% lower than the cost of transporting the ``Euclidean'' one.

We now give some examples of barycentres. 
In what follows, we use $\bar \mu$ to represent the measure on $\mathbb S_1$ returned from Algorithm \ref{algo:barycentre} and $\widetilde{\bar \mu}$ the associated periodic measure on $\R$. 
In some cases, it is intuitive what should be the barycentre and we show that our algorithm correctly converges to it. 
In other ones, intuition fails but we still might get an idea of the goodness of the approximation 
of the barycentre by comparing the Wasserstein 
distances $W_2(\mu_i, \mu_j)$ with the distances in the tangent space as in \Cref{prop:conv_bary}.
Moreover, we also compare the output of Algorithm \ref{algo:barycentre} with the so-called Sinkhorn barycentre \citep{cuturi2014fast, janati2020debiased} as implemented in the \texttt{Python} package \texttt{ott-jax} \citep{jax_ott}. 
To compute the Sinkhorn barycentre, we approximate each measure with an atomic measure with $1,000$ equispaced support points on $[0, 2\pi)$, equipped with the geodesic distance on $\S_1$, giving to each point $x_i$ a weight proportional to $\mu(\dd x_i)$. 
Informally, we should expect the Wasserstein and Sinkhorn barycentres to be similar, but the Sinkhorn barycentre should be smoother due to the regularisation term involved in the Sinkhorn divergence.

\begin{figure}[t]
	\centering
	\includegraphics[width=\linewidth]{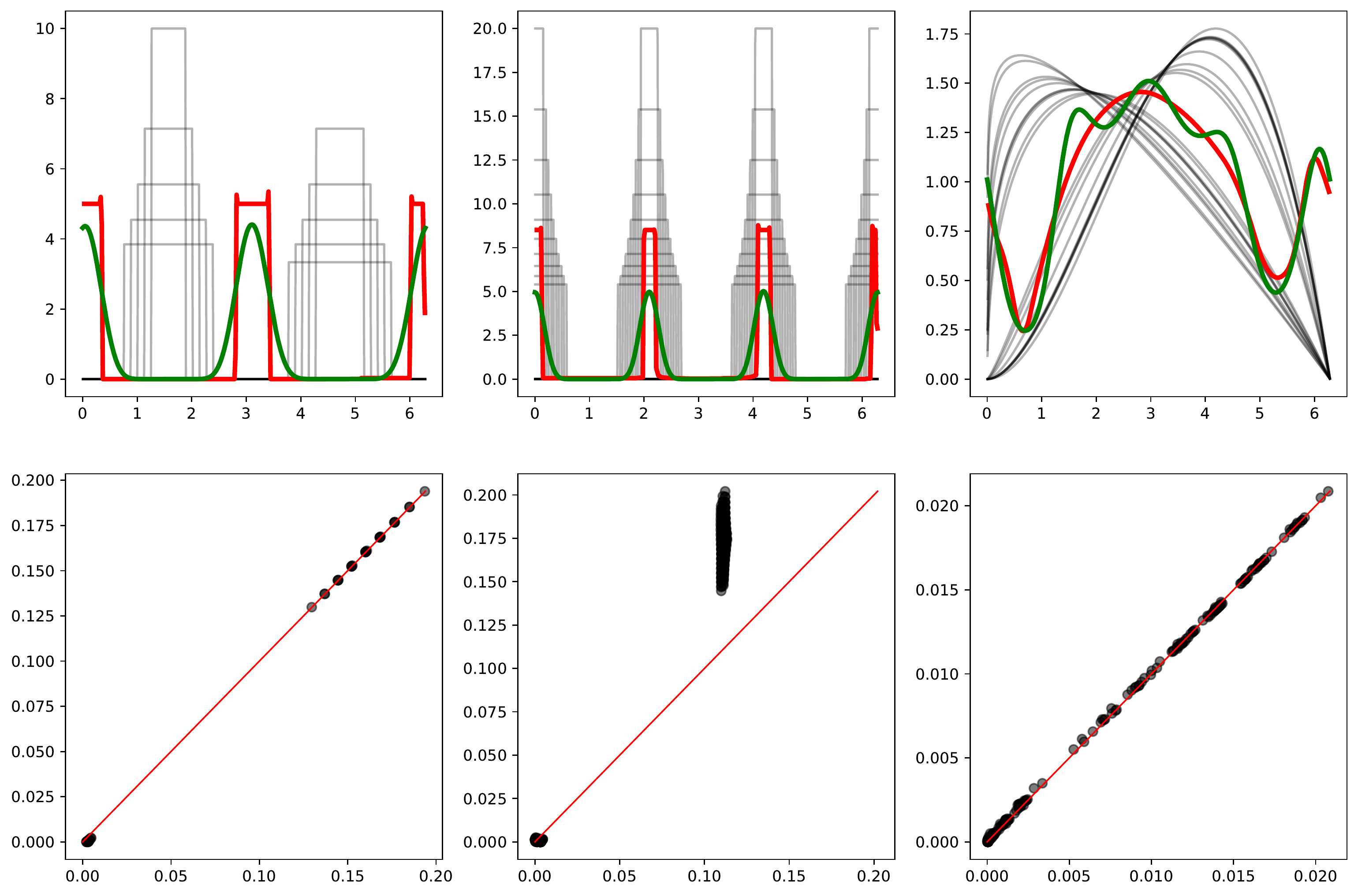}
	\caption{Top row: denisties of the $\mutilde_j$'s on $[0, 2\pi]$, and of the Wasserstein and Sinkhorn barycentres (red and green line respectively). Bottom row: Wasserstein distance vs $d_{\log}$ for every possible couple of measures.}
	\label{fig:bary_simu}
\end{figure}

We consider three simulated datasets as follows.
Let $\mathcal{U}(c, w)$ denote the uniform measure centred in $c$ and with width $w$, i.e. the uniform measure over $(c - w/2, c + w/2)$.
In the first example, the measures are 
\begin{align*}
	\mutilde_i &= \mathcal{U}\left(0.25, 0.1 + 0.05 i\right), \quad i=1, \ldots, 5 \\
	\mutilde_i &= \mathcal{U}\left(0.75, 0.1 + 0.05 (i-5)\right), \quad i=5, \ldots, 10
\end{align*}
and extended periodically over the whole $\mathbb{R}$. In the second one instead
\begin{align*}
	\mutilde_i &= \mathcal{U}\left(0, 0.05 + 0.015 i\right), \quad i=1, \ldots, 10 \\
	\mutilde_i &= \mathcal{U}\left(1/3, 0.05 + 0.015 (i-10)\right), \quad i=11, \ldots, 20 \\
	\mutilde_i &= \mathcal{U}\left(2/3, 0.05 + 0.015 (i-20)\right), \quad i=21, \ldots, 30
\end{align*}
In the third case instead, we generate the $\mutilde_i$'s by first considering Beta distributions on $(0,2\pi)$ with parameters $(a_i, 2)$ and then taking their periodic extension.
Specifically, $a_i \sim \mathcal{U}(1.3, 0.2)$ for $i=1, \ldots, 10$ and $a_i \sim \mathcal{U}(2.6, 0.4)$ for $i=11, \ldots, 20$.
Figure \ref{fig:bary_simu} reports the Wasserstein barycentres as found by Algorithm \ref{algo:barycentre} and the Sinkhorn ones for three different simulated datasets. 
We can see that the Wasserestein ans Sinkhorn barycentres agree and that the Sinkhorn ones are generally smoother. Moreover, in the first and third example the log and Wasserstein distances are indistinguishable which suggests the convergence of Algorithm \ref{algo:barycentre},
while in the second example there are some discrepancies.
The third simulation allows us to gather some insights into the geometry of $\Wcal_2(\S_1)$. Indeed, note how, despite all the measures $\mutilde_i$ being unimodal, the barycentre is bimodal.
This clearly arises from the manifold structure of $\S_1$ and specifically because of mass going through $0$ along the geodesics connected some measures.

\subsection{Simulations for the PCA}
\label{sec:simu}

In this section we analyse some simulated datasets which we use to showcase and interpret some behaviours of the PCA defined in previous sections. Another simulation with additional details and comparisons can be found in \Cref{app:further_simu}.
To interpret the principal directions found by the PCA, we produce the plots of the densities of $\exp_{\bar \mu}(\log_{\bar \mu}(\mu_0) +  \lambda w^*_k)$, where $w^*_k$ is the $k$-th principal direction and $\lambda$ varies in some range specified case-by-case.
Unless otherwise stated, $\bar \mu$ and $\mu_0$ are both equal to the Wasserstein barycentre approximated using Algorithm \ref{algo:barycentre}.
In particular, note that the score $\lambda$ represents the distance from the base point travelled along the geodesic whose direction is specified by the $k$-th principal direction. It is then possible to compare different values of $\lambda$ across the simulations to interpret the distance from the barycentre after which some behaviours start to occur (for instance, it might happen that at a certain distance from the barycentre, the measures switch from unimodal to bimodal).

First, we consider a sample from the von Mises distribution with location $\pi$ and scale $\alpha$, whose density function on $[0, 2\pi]$ is 
\begin{equation}\label{eq:vm}
	f(x; \alpha) = \frac{\exp\left(3 \cos\left( \frac{x - \pi}{\alpha}\right)\right)}{2 \pi I_0(3)},
\end{equation}
where $I_0$ is the modified Bessel function of order zero.
We simulate two datasets of $n=100$ measures from \eqref{eq:vm}, by considering $\alpha \sim \mathcal U(0.8, 1.5)$ and $\alpha \sim \mathcal U(2, 3.5)$ respectively. 
Data and the first principal direction are shown in \Cref{fig:vonmises_pc}. 
In the first case, the measures are sufficiently concentrated so that, in the neighbourhood of the barycentre associated to the grid of values for $\lambda$, the periodicity of $\S_1$ is effectively irrelevant, and the first principal direction reflects the change in scale of the distribution. 
On the other hand, in the second case, we have a good amount of mass around $0$ for all distributions in the data set, and the variance of such distributions ranges over a bigger interval compared to the first data set. As a consequence, moving along the first principal direction (with the same scale as in the previous example), we keep pushing the mass on \virgolette{the sides} at faster rates, so that it concentrates even more around $0$ and we go from a unimodal to a bimodal density.  

Although not shown here, when the same measures are considered as points in $\Wcal_2(\R)$, in both cases the first principal direction is associated with a change in the scale of the measures, while the location is kept fixed.

\begin{figure}[t]
	\centering
	\begin{subfigure}{0.33\linewidth}
		\includegraphics[width=\linewidth]{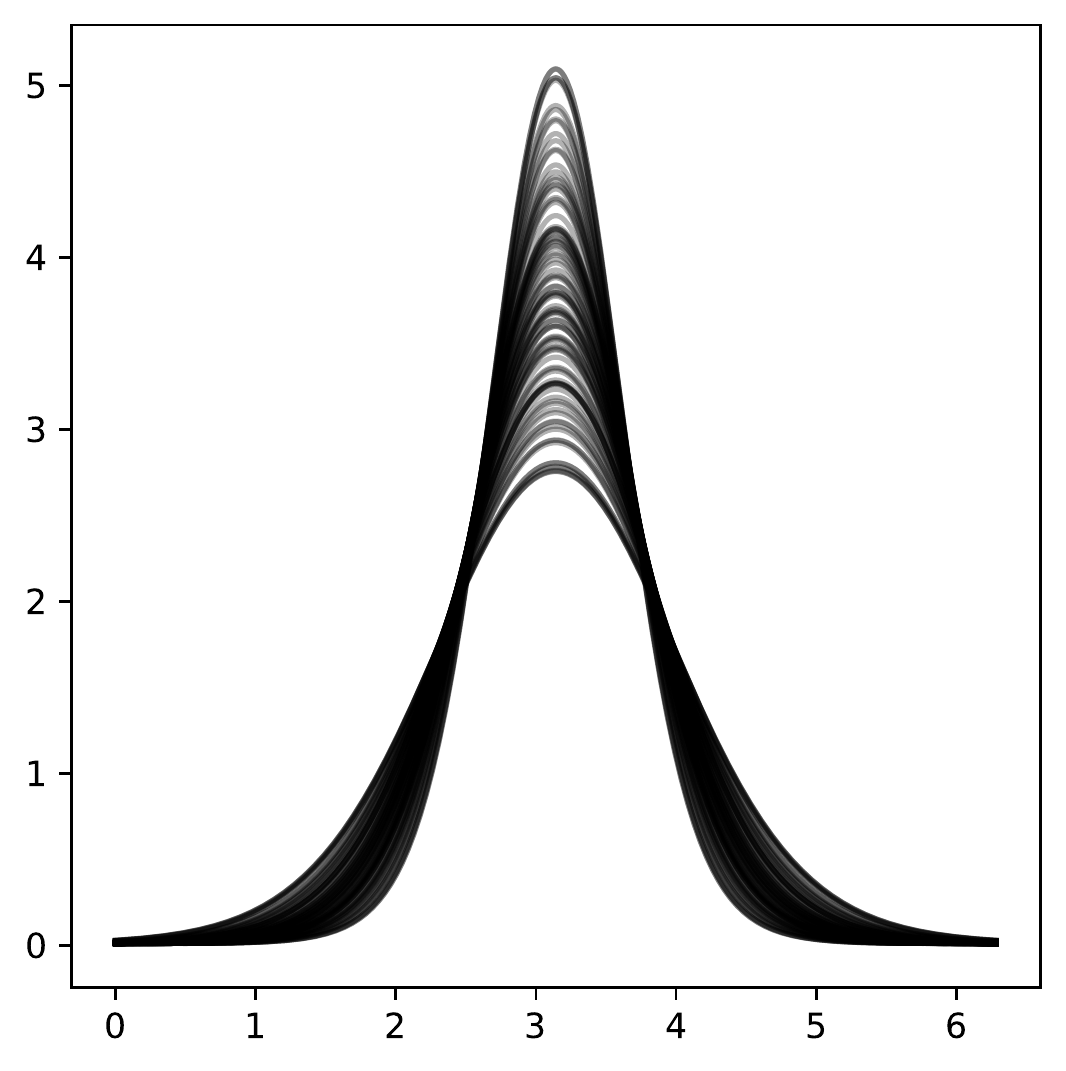}
	\end{subfigure}%
	\begin{subfigure}{0.66\linewidth}
		\includegraphics[width=\linewidth]{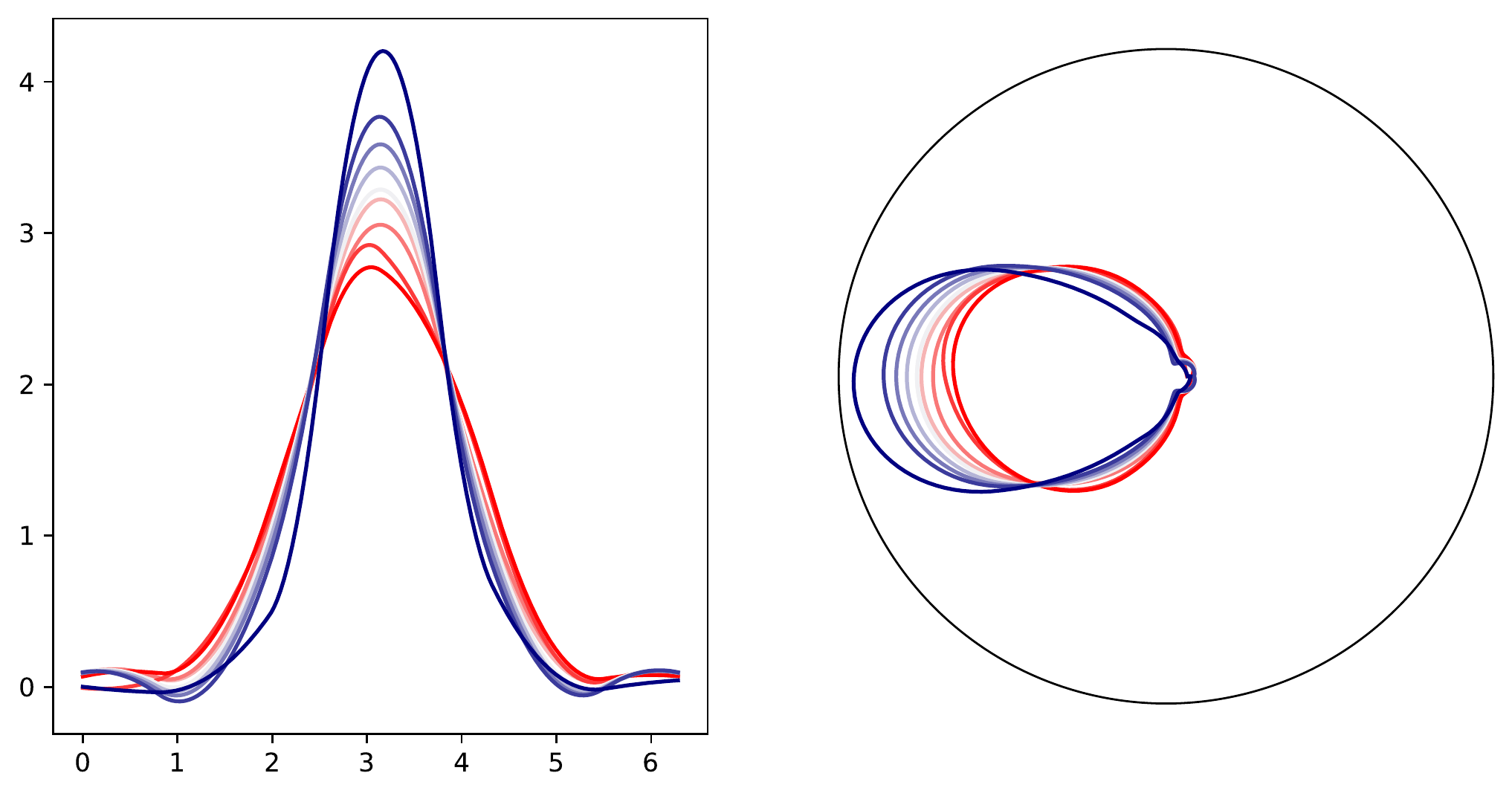}
	\end{subfigure}
	\begin{subfigure}{0.33\linewidth}
		\includegraphics[width=\linewidth]{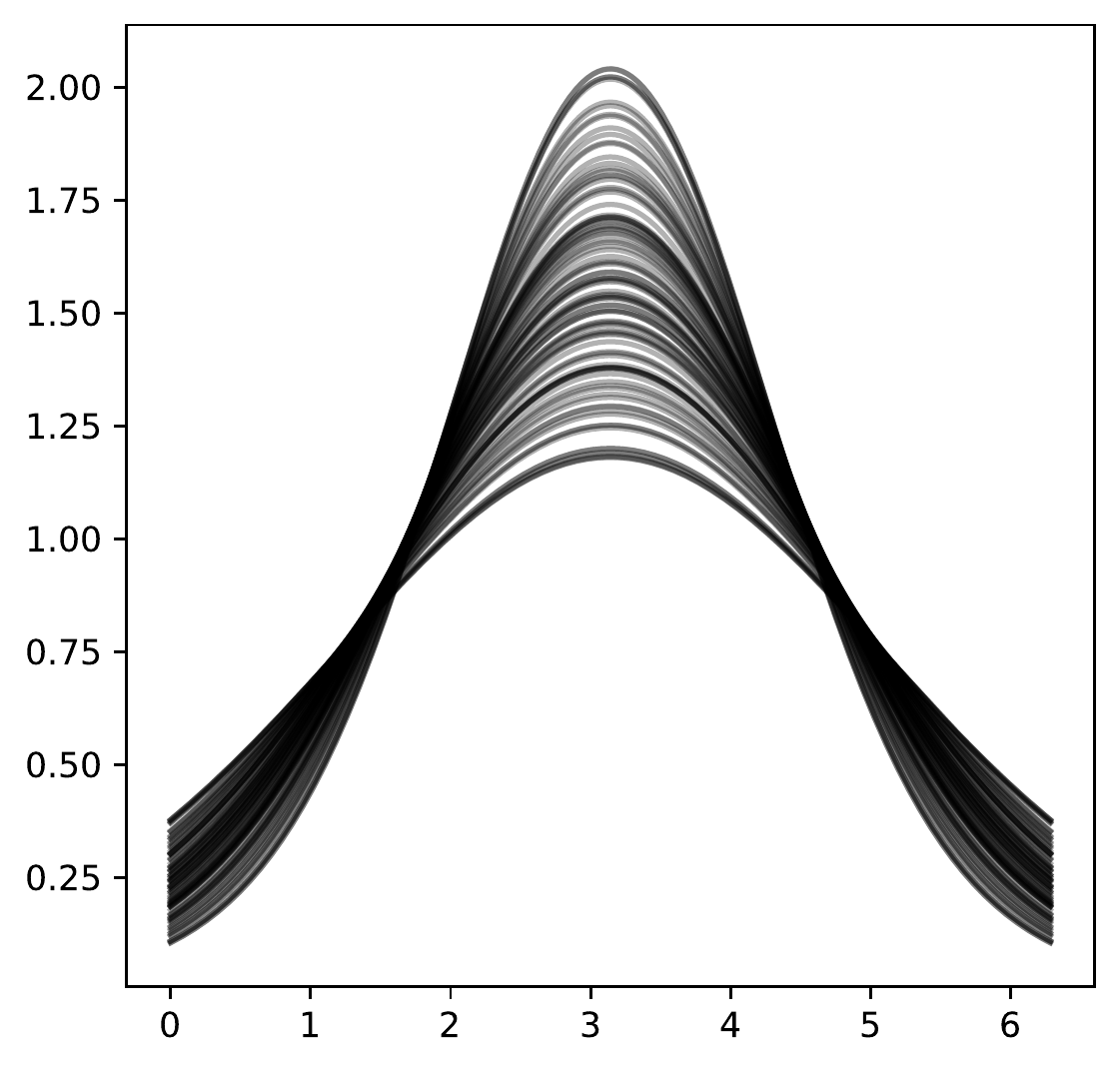}
	\end{subfigure}%
	\begin{subfigure}{0.66\linewidth}
		\includegraphics[width=\linewidth]{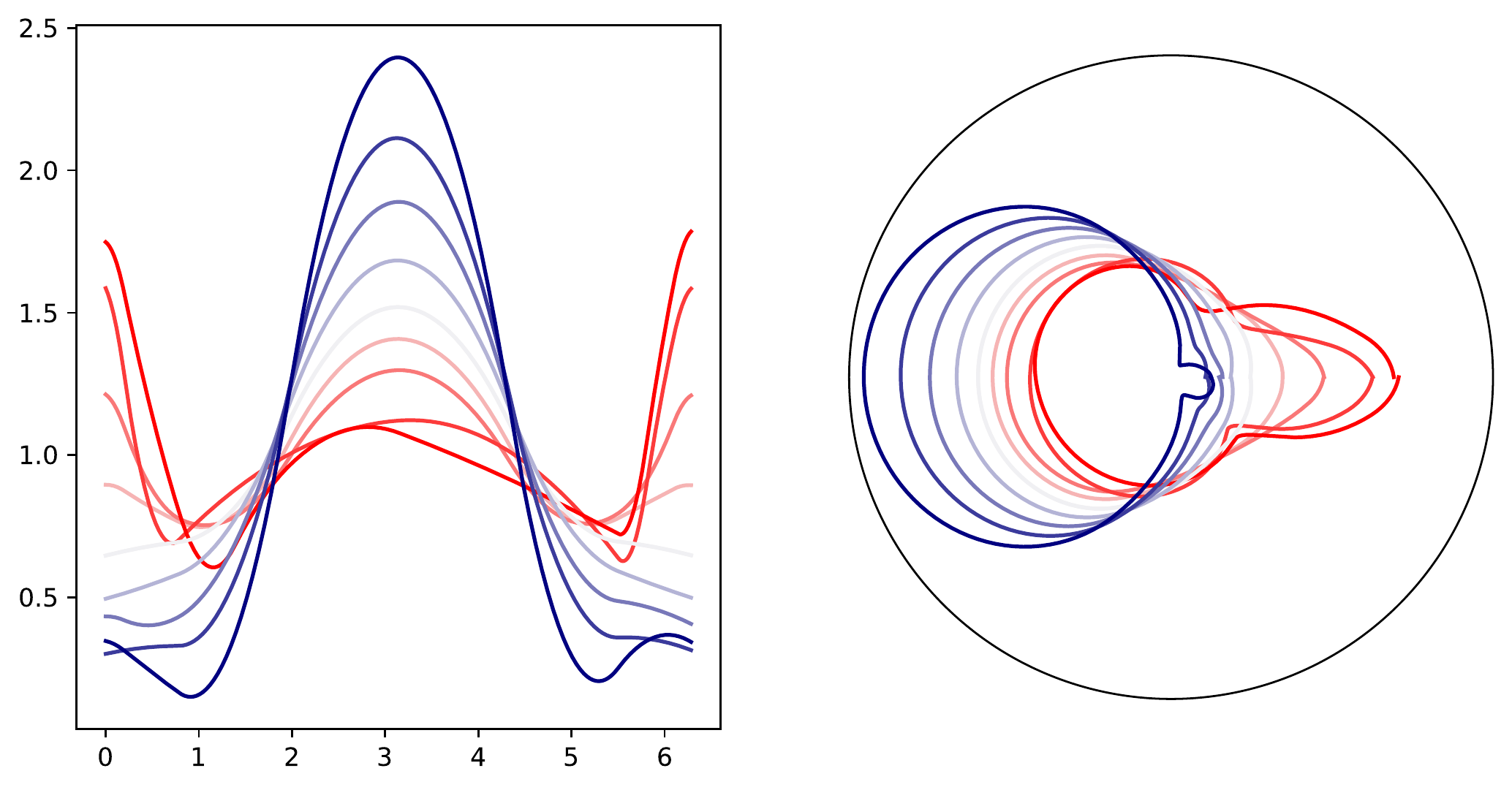}
	\end{subfigure}
	\caption{Data and first principal direction for the Von Mises simulation. The second and third column represent densities along the first principal direction as $\lambda$ varies between $-0.1$ (darkest blue) to $0.1$ (darkest red), plotted as distributions on $[0, 2 \pi)$ and on $\S_1$ respectively.}
	\label{fig:vonmises_pc}
\end{figure}

Next, we consider the same dataset as in the third simulation of \Cref{sec:simu_bary}. \Cref{fig:betas_pc} reports the first two principal directions. The first one corresponds mostly to a shift on the location but simultaneously it also captures the decrease of the density around the second mode that is located in $0$ (see the barycentre in \Cref{fig:bary_simu}).
Starting from the barycentre (white),  if we go towards the red densities we see that the mode in zero gradually is absorbed the main mode; while if we go towards the blue ones the mode in $0$ crosses the circle and it merge on the main mode, but on the right side of the plot. According to the geodesic structure of $\Wcal(\S_1)$.
The second direction, instead, is more clearly focused on separating distribution with significant amount of mass close to $0$ (blue), from the measures which, instead, have all their mass away from $0$ (red).

\begin{figure}[t]
	\centering
	\begin{subfigure}{0.5\linewidth}
		\includegraphics[width=\linewidth]{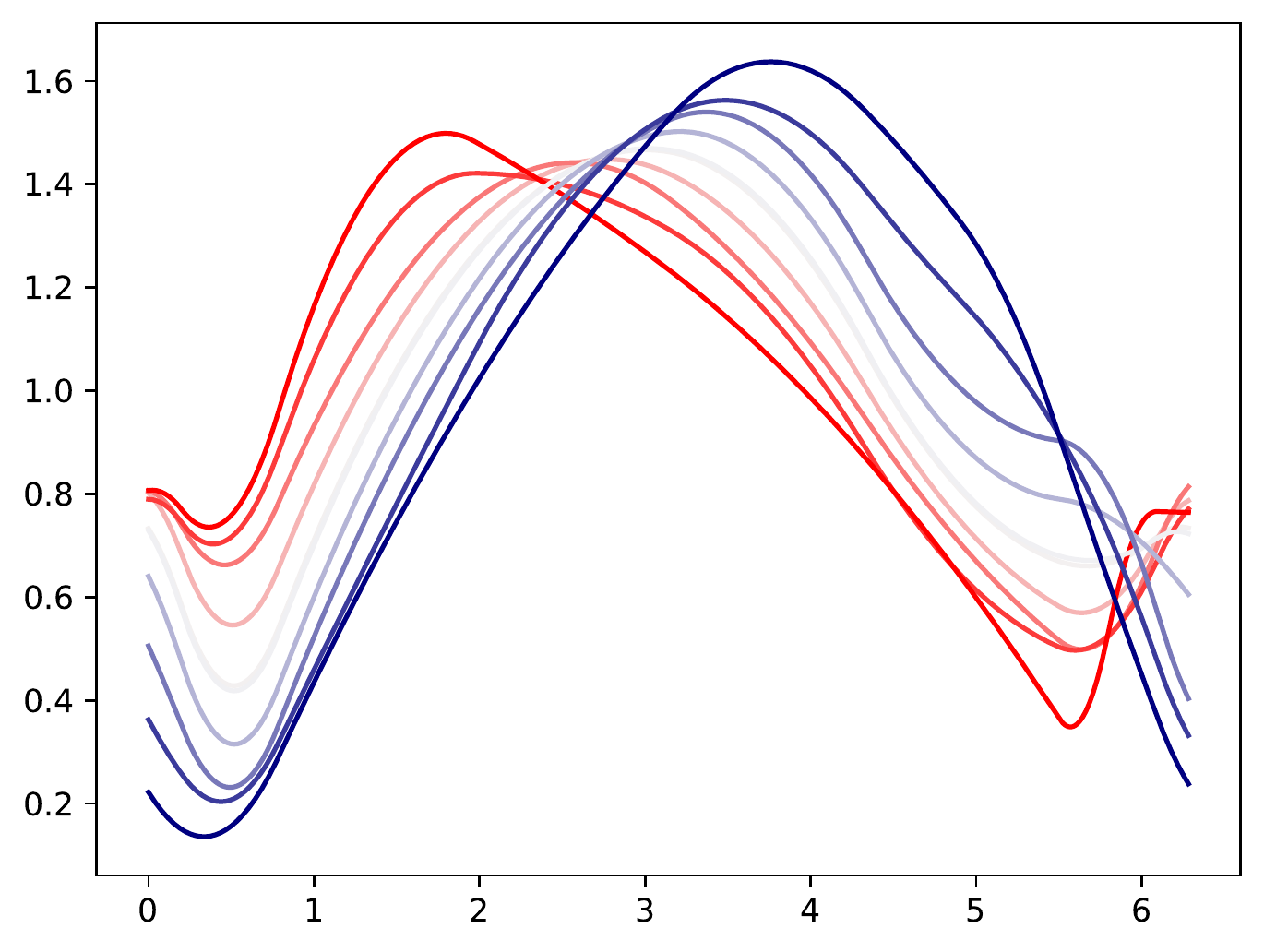}
	\end{subfigure}%
	\begin{subfigure}{0.5\linewidth}
		\includegraphics[width=\linewidth]{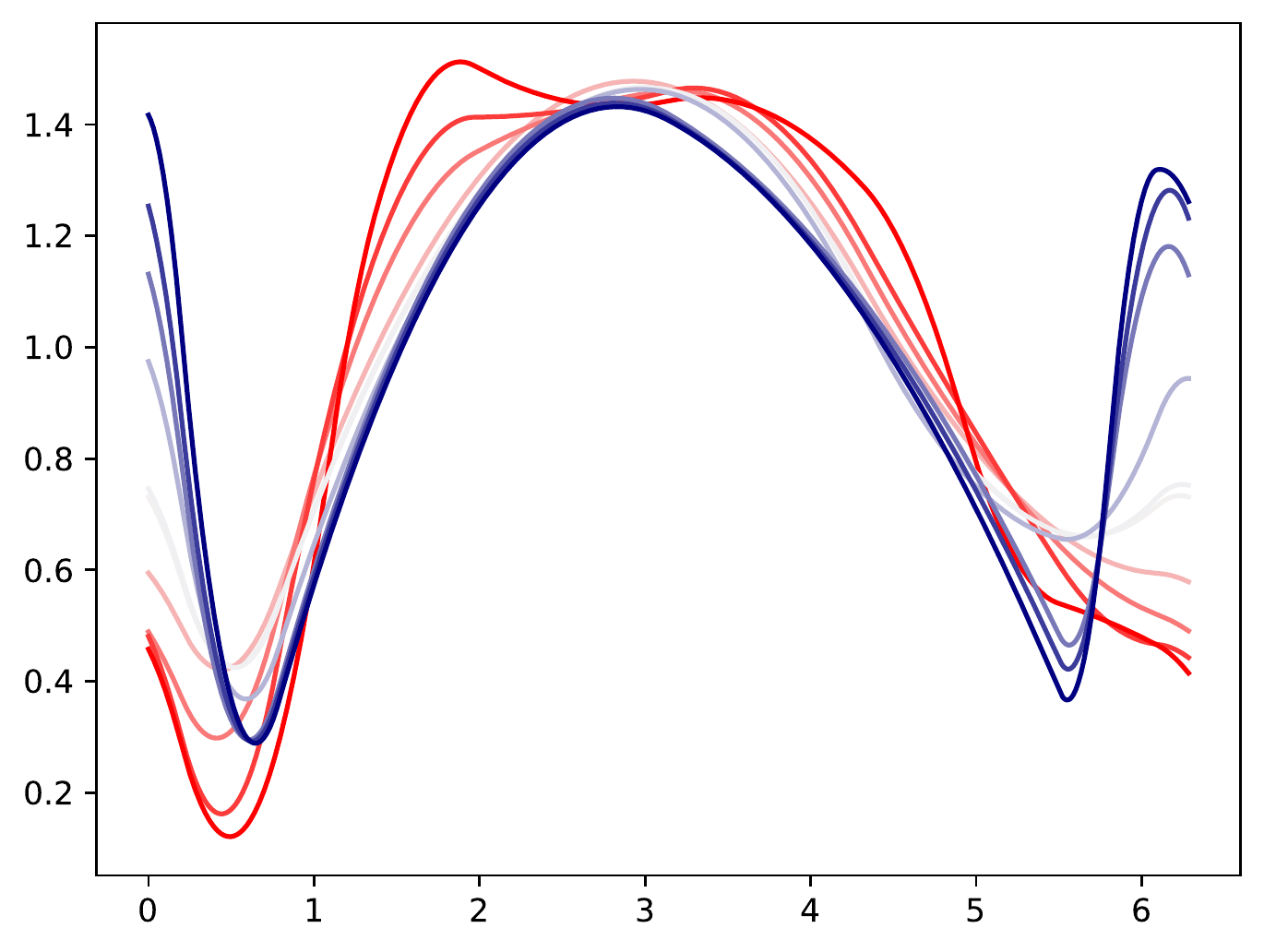}
	\end{subfigure}
	\caption{Densities along the first two principal directions for the Beta distribution as in  \Cref{sec:simu_bary}, as $\lambda$ varies between $-0.05$ (darkest blue) to $0.05$ (darkest red), plotted as distributions on $[0, 2 \pi)$.}
	\label{fig:betas_pc}
\end{figure}

In summary,  these simulations help us understand the geometry of $\Wcal_2(\S_1)$ and, in particular, the differences with $\Wcal_2(\R)$.
Indeed, it is well-known that, for measures on $\R$, the Wasserstein geodesics of location-scale families 
are obtained by lifting the Euclidean geodesics in the location-scale plane to the Wasserstein space. 
Hence, the Wasserstein PCA will disentangle the effect of the location and the effect of the scale.
Instead, as shown by our simulations, when measures are supported on $\S_1$ it is not possible to completely separate the effects of location and scale.
Moreover, even if the datapoints are unimodal, it is often the case that the barycentre is multimodal. 
Multimodality is inherited by the measures along the principal directions, which might make the interpretation cumbersome.
In \Cref{app:further_simu} we report an additional simulation for the PCA, where we discuss the choice of the point $\bar \mu$ (at which the tangent is attached) and its impact on the interpretability of the directions. In particular, we consider a dataset of truncated gaussians, for which the barycentre has three modes. Instead, if $\bar \mu$ is chosen to be equal to one of the datapoints, then moving along the principal directions results in unimodal densities for which interpretation is easy. 
Of course, this poses a conceptual issue as the principal directions are not the ``main directions of variability'' per se, but the main directions of variability starting from one particular $\bar \mu$.

\section{Case Study: Eye Dataset}\label{sec:eye}

\begin{figure}[ht]
\centering
	\begin{subfigure}{0.65\linewidth}
		\includegraphics[width=\linewidth]{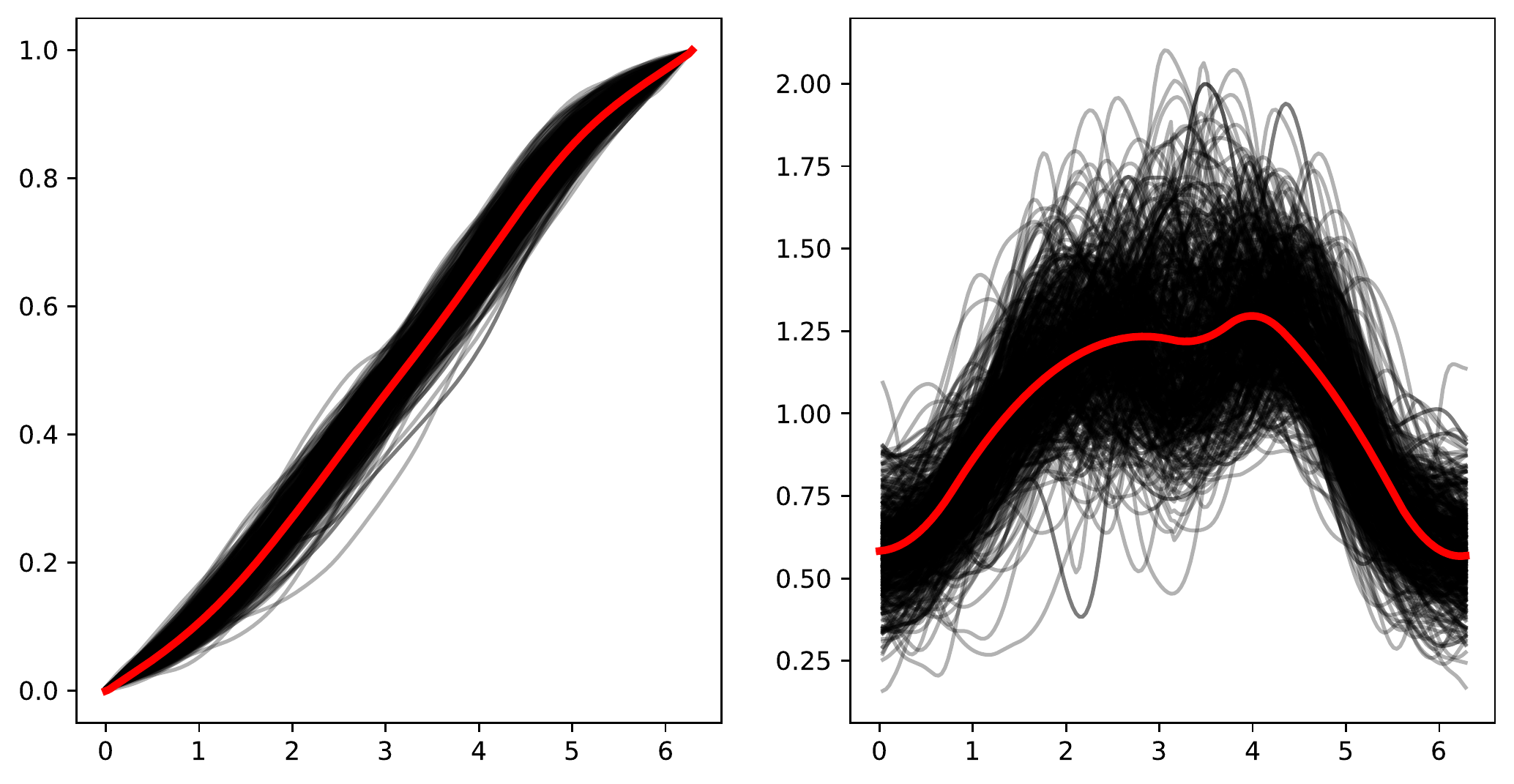}
	\end{subfigure}
	\begin{subfigure}{0.3\linewidth}
		\includegraphics[width=\linewidth]{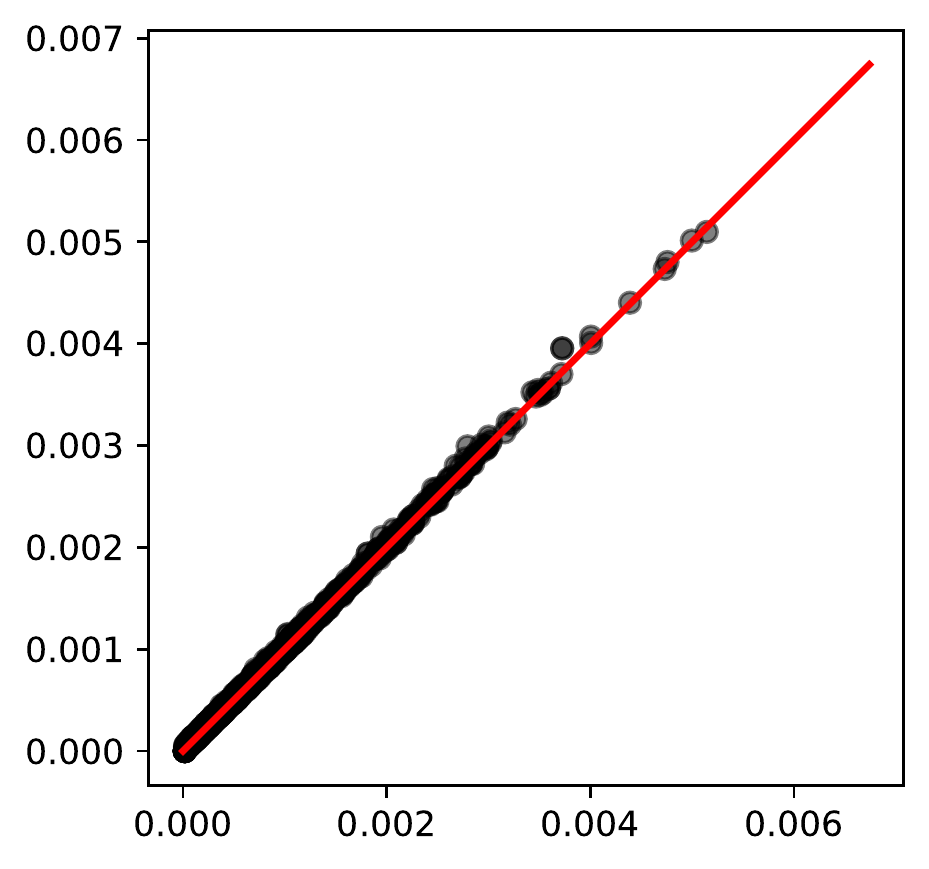}
	\end{subfigure}
	\caption{From left to right: (a subsample of) cdfs of the eye's dataset measures (red line denotes the barycentre), pdfs of the eye's dataset measures (red line denotes the barycentre), Wasserstein distance against $d_{\log}$ in the tangent space at the barycentre.}
	\label{fig:eye_data}
\end{figure}

\begin{figure}[ht]
	\centering
	\begin{subfigure}{0.5\linewidth}
		\includegraphics[width=\linewidth]{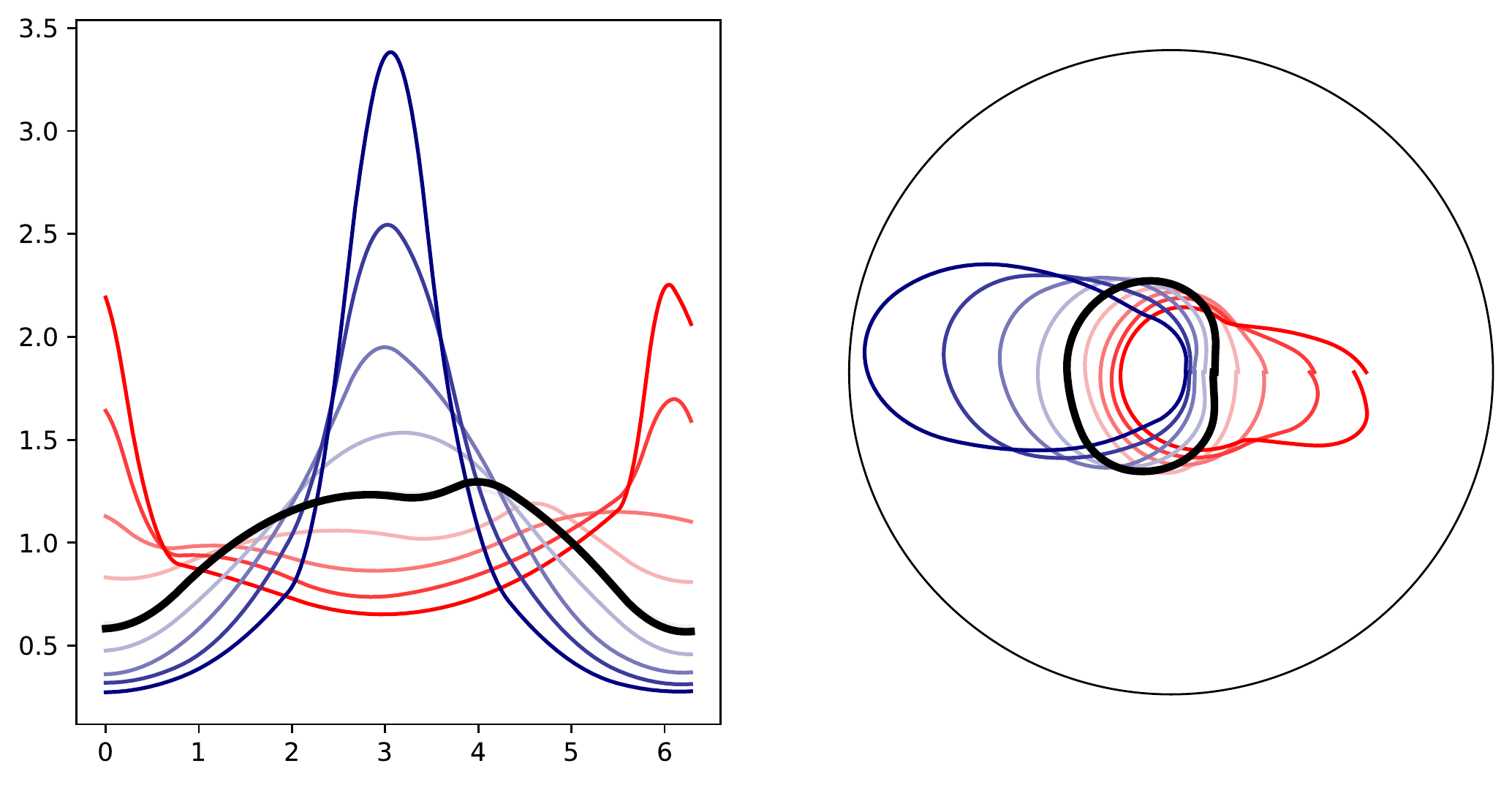}
	\end{subfigure}%
	\begin{subfigure}{0.5\linewidth}
		\includegraphics[width=\linewidth]{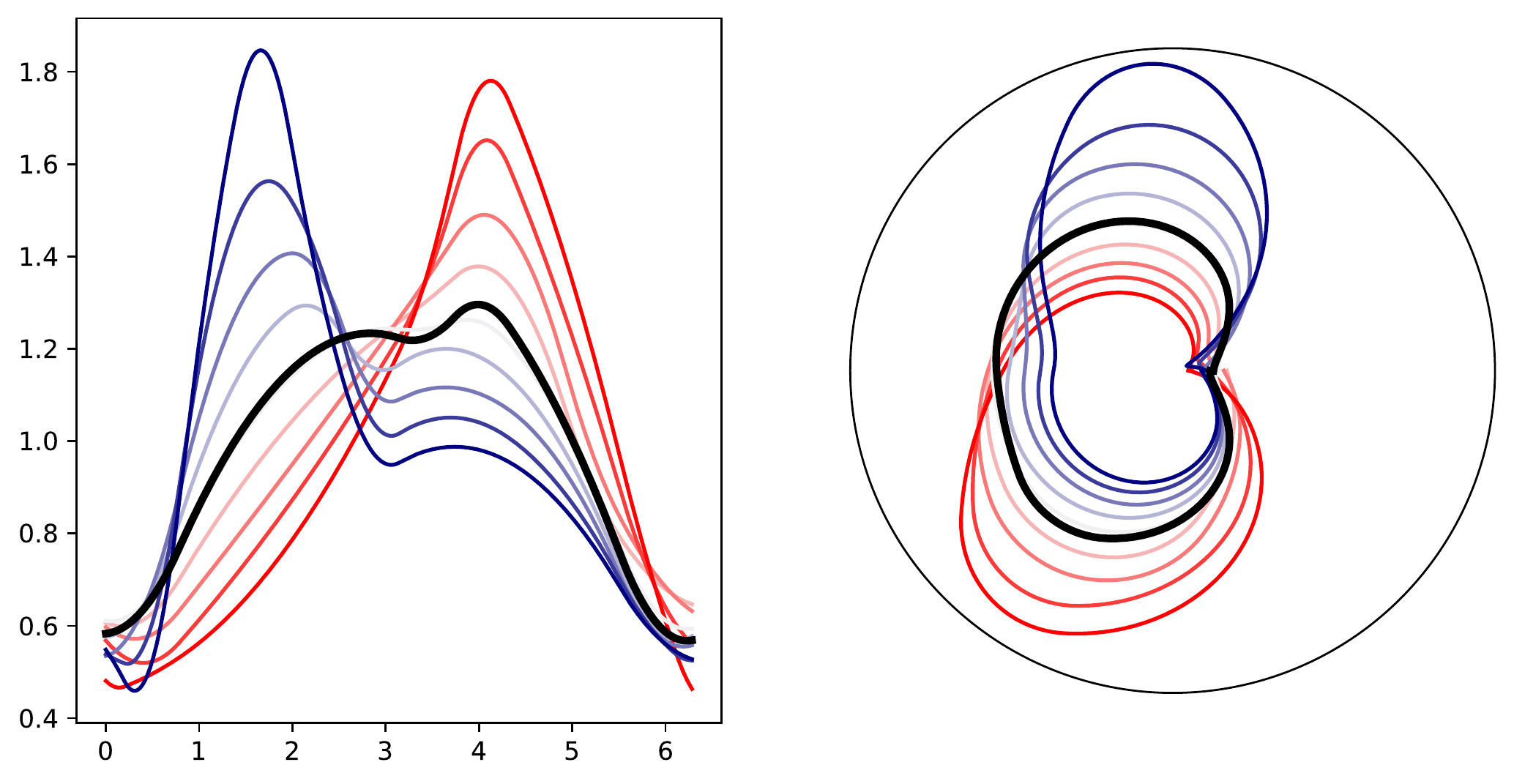}
	\end{subfigure}
	\caption{First (left plots) and second (right plots) principal directions: we report the pdfs on $[0, 1]$ (first and third panels) and in a polar plot (second and fourth panels). The black line denotes the barycentre.}
	\label{fig:eye_pd}
\end{figure}

We present here the results of applying PCA to the dataset of OCT measurements of NRR in \cite{ali2021circular}, available in their supplementary materials, which contains the OCT measurements of $3973$ patients, stratified according to their age groups. 
In particular, we assess the adequacy of Wasserstein PCA by interpreting the principal direction and performing clustering on the scores, showing how these clusters meaningfully capture shape patterns in data.
Data are displayed in \Cref{fig:ex_data} together with the Wasserstein barycentre found via Algorithm \ref{algo:barycentre}.
In the rightmost plot, we show how the Wasserstein and $L_2$ distances in the tangent space at the barycentre agree for almost all the couples of datapoints, thereby validating the use of the red measure in \Cref{fig:ex_data} as centering point for our PCA.

The first two principal directions -- which, by construction, are the two directions capturing most variability --  are reported in \Cref{fig:eye_pd}. 
We can clearly see that these decouple the shape variability along the horizontal and vertical axes. In particular, this implies that most of the variability in the data set is made by variations (in the distribution of the) of thickness of the optical nerve, along the horizontal axis.
To assess the adequacy of Wasserstein PCA for this dataset, we compute the \emph{average normalised reconstruction error} as a function of the number of directions $k$ used for the PCA.
\[
    \mbox{ANRE}_k :=\frac{1}{n} \sum_{i=1}^n\frac{ W^2_2(\mu_i,\mu^k_i)}{W^2_2(\bar{\mu},\mu_i)},
\]
where $\mu^k_i$ is the projection on the first $k$ principal components of the measure $\mu_i$. The ANRE index measures the approximation error, normalising by the deviation of the datapoints from the centre of the PCA, in close analogy with the decomposition of variance in the case of PCA in Euclidean spaces.
\Cref{fig:eye_err_clus} (left plot) reports the ANRE index as a function of $k$, as well as the (normalised) eigenvalues of the $L_2$ PCA in the tangent space.
Both measures show how the first $k=5$ directions are enough to capture the variability of the dataset. Moreover, the $L_2$ variance decreases faster than ANRE. This is expected since $L_2$ PCA ignores that data are constrained on the image of $\log_{\bar \mu}$, and ``captures variability'' also outside this set. Lastly, we believe that the ANRE in stabilises to a positive (small) number due to numerical errors. 
In \Cref{sec:app_eye2}, we report the scatter plot of the scores along the first two directions, stratified by age groups. From the plot, it is clear that, on the first two components, there is no evident effect of age alone on the shape of the optical nerve.

We cluster the datapoints via a hierarchical clustering algorithm with ward linkage working on the scores along the first $k=5$ principal directions.
In \Cref{sec:app_eye2} we show the dendrogram, while the two main clusters found are shown in  \Cref{fig:eye_err_clus}.
\Cref{fig:eye_7clus} reports a refined clustering obtained by cutting the dendrogram to get 7 clusters.
We have reported in red the barycentres of the clusters, which may be of some help in interpreting the clusters, even though our clustering pipeline is not barycentre-driven like a K-means algorithm. When looking at the two clusters in \Cref{fig:eye_err_clus}(b), it is clear that they identify two different shapes of the optical nerve with the left one being characterised by a clear bump in the left side.
The refined clusters in \Cref{fig:eye_7clus} in the appendix show interesting patterns as well, see the appendix for further details.

We close this section by highlighting that, as mentioned in the introduction, a very important byproduct of PCA is that classical tools from multivariate statistics can be applied to our dataset after projecting data on the principal components. 
We leave it to future works to complement our unsupervised analysis with an investigation involving the covariates contained in the original dataset.

\begin{figure}[t]
    \centering
	\begin{subfigure}{0.5\linewidth}
		\includegraphics[width=\linewidth]{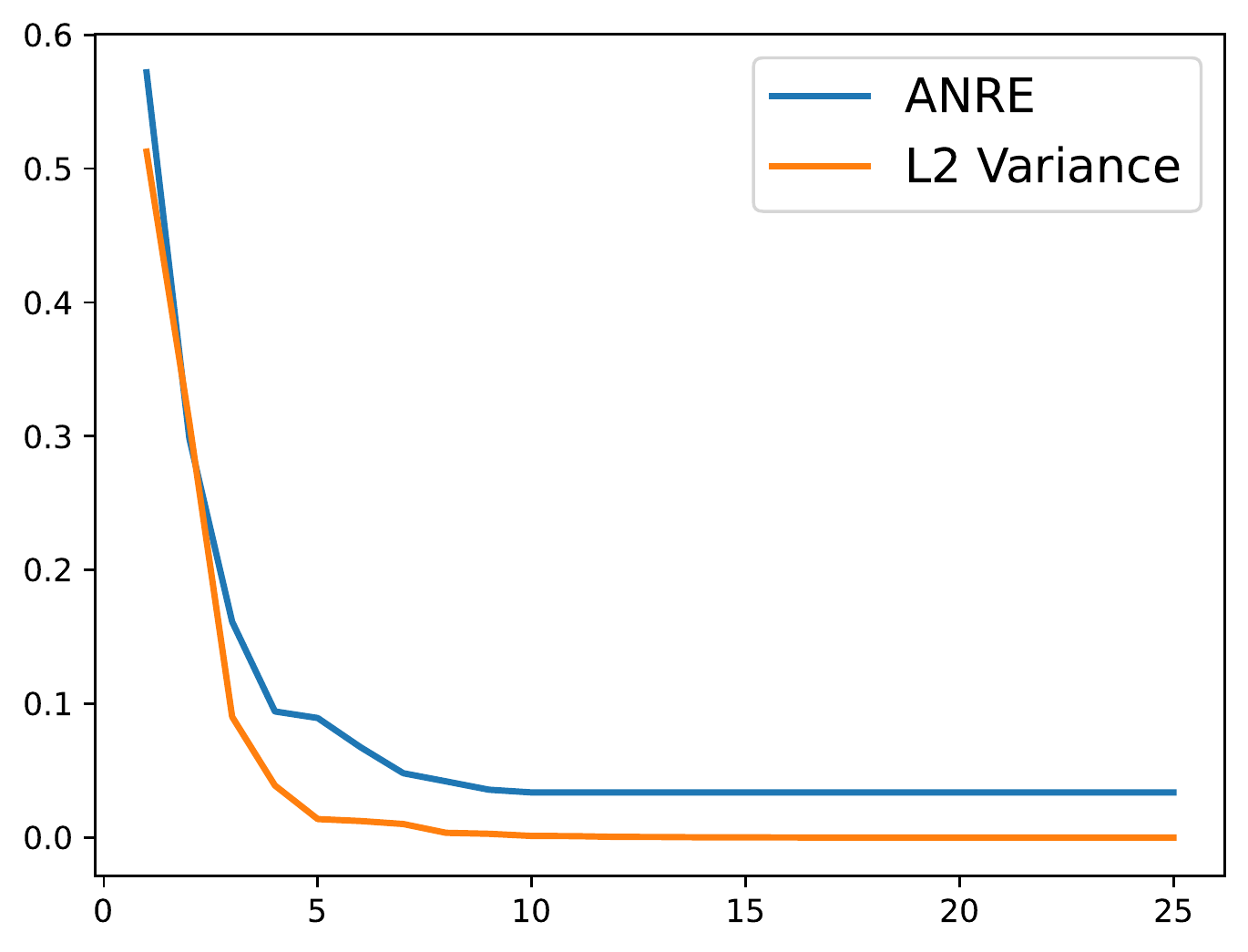}
        \caption{}
	\end{subfigure}%
    \begin{subfigure}{0.5\linewidth}
		\includegraphics[width=\linewidth]{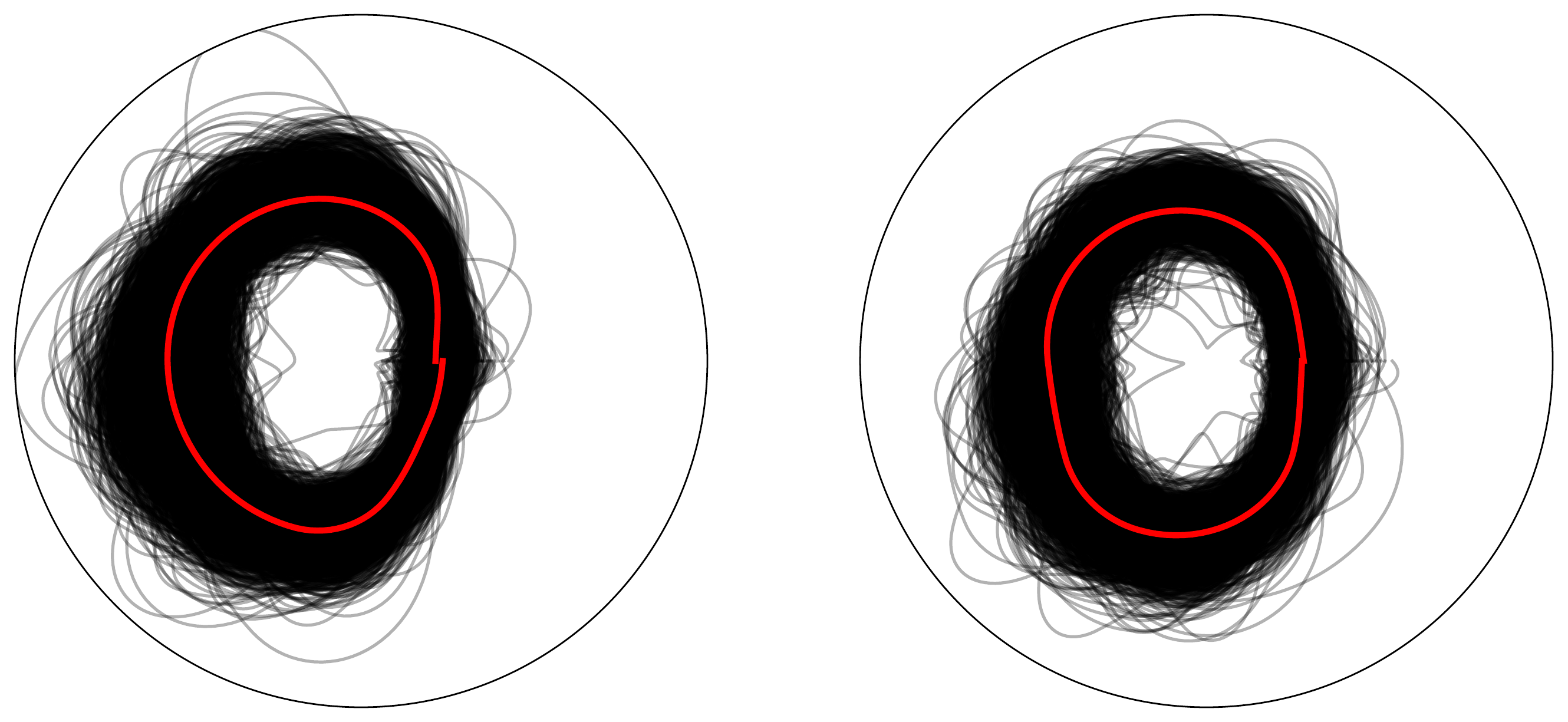}
        \caption{}
	\end{subfigure}
	\caption{ANRE index as a function of the dimension and fraction of ($L_2$) variance explained by each component (left plot) and data subdivided in two clusters (right), with the corresponding Wasserstein barycentre (red lines).}
	\label{fig:eye_err_clus}
\end{figure}

\section{Discussion}
\label{sec:conclusion}

In this paper, we tackled the problem of analysing distributional data supported on the circle.
Following recent trends in statistics and machine learning, we set out to use the Wasserstein distance to compare probability distributions.
To this end, we studied the optimal transportation problem on $\S_1$ and established several new theoretical results, which could also be of independent interest.
In particular, we provide an explicit characterisation of the optimal transport maps.
This result is rather surprising given that optimal transport on Riemannian manifolds is not well established and that the only case where such explicit formulas exist is for measures on the real line.
We further explored the weak Riemannian structure of the Wasserstein space and established strong continuity results for the exponential and logarithmic maps, as well as an explicit characterisation of the image of the logarithmic map.

Building on our theoretical findings, we propose a counterpart of the convex PCA in \cite{geodesic} for measures on $\S_1$. Following the approach in \cite{projected}, we propose a numerical method to compute the principal directions by means of a B-spline expansion, which leads to an easily implementable numerical algorithm.

Our definition of PCA requires a ``central point'', which is usually set equal to the barycentre. We used the algorithm in \cite{zemel2019frechet} to approximate the Wasserestein barycentre. However, we have not been able to prove the convergence of this algorithm in our setting.
Despite numerical simulations do seem to validate the use of Algorithm \ref{algo:barycentre}, the theoretical analysis is still an open problem.

Our investigation paves the way to several interesting extensions.
First, it is natural to consider the problem of Wasserestin regression. 
Thanks to the expression for the optimal transport maps, the geodesic regression in \cite{geod_regression_flet} can be defined in an analogous way for measure-valued dependent random variables.
Similarly, our definition of tangent space is amenable to the definition of a log regression for measures on $\S_1$.
For measures on $\R$, \cite{projected} proposed to map both dependent and independent variables onto the same tangent space, given that the Wasserstein space is isomorphic to any tangent.
Here, it would be more suitable to consider two tangent planes: one for the independent and one for the dependent variables, centred at the respective barycentres, similarly to \cite{muller}.

More broadly, we believe that the interplay between optimal transport and distributional data analysis can nourish further developments of both fields. Specialising the treatment of the optimal transportation theory to specific cases of statistical interest, such as the sphere, could lead on one hand to a better understanding of how the properties of tangent spaces relate to the base manifold, and on the other hand to data analysis frameworks which can extract insights for instance from earth-related distributions and other relevant data which are nowadays collected.

\FloatBarrier

\bibliography{MYBib}

\begin{thebibliography}{33}
\providecommand{\natexlab}[1]{#1}
\providecommand{\url}[1]{\texttt{#1}}
\expandafter\ifx\csname urlstyle\endcsname\relax
  \providecommand{\doi}[1]{doi: #1}\else
  \providecommand{\doi}{doi: \begingroup \urlstyle{rm}\Url}\fi

\bibitem[Agueh and Carlier(2011)]{agueh_barycenter}
M.~Agueh and G.~Carlier.
\newblock Barycenters in the {W}asserstein space.
\newblock \emph{SIAM J. Math. Anal.}, 43\penalty0 (2):\penalty0 904--924, 2011.
\newblock ISSN 0036-1410.
\newblock \doi{10.1137/100805741}.
\newblock URL \url{https://doi.org/10.1137/100805741}.

\bibitem[Ali et~al.(2021)Ali, Wainwright, Petersen, Jonnadula, Desai, Rao,
  Srinivas, Jammalamadaka, Senthil, Pyne, et~al.]{ali2021circular}
M.~Ali, B.~Wainwright, A.~Petersen, G.~B. Jonnadula, M.~Desai, H.~L. Rao,
  M.~Srinivas, S.~R. Jammalamadaka, S.~Senthil, S.~Pyne, et~al.
\newblock Circular functional analysis of oct data for precise identification
  of structural phenotypes in the eye.
\newblock \emph{Sci. Rep.}, 11\penalty0 (1):\penalty0 1--13, 2021.

\bibitem[Alvarez-Esteban et~al.(2018)Alvarez-Esteban, del Barrio,
  Cuesta-Albertos, and Matran]{alvarez2018wide}
P.~C. Alvarez-Esteban, E.~del Barrio, J.~A. Cuesta-Albertos, and C.~Matran.
\newblock Wide consensus aggregation in the wasserstein space. application to
  location-scatter families.
\newblock \emph{Bernoulli}, 24\penalty0 (4A):\penalty0 3147--3179, 2018.

\bibitem[Ambrosio et~al.(2008)Ambrosio, Gigli, and
  Savar{\'e}]{ambrosio2008gradient}
L.~Ambrosio, N.~Gigli, and G.~Savar{\'e}.
\newblock \emph{Gradient flows: in metric spaces and in the space of
  probability measures}.
\newblock Springer Science \& Business Media, 2008.

\bibitem[Ambrosio et~al.(2019)Ambrosio, Glaudo, and
  Trevisan]{ambrosio2019optimal}
L.~Ambrosio, F.~Glaudo, and D.~Trevisan.
\newblock On the optimal map in the 2-dimensional random matching problem.
\newblock \emph{arXiv preprint arXiv:1903.12153}, 2019.

\bibitem[Banerjee et~al.(2015)Banerjee, Chakraborty, Ofori, Vaillancourt, and
  Vemuri]{geod_regression_chakra}
M.~Banerjee, R.~Chakraborty, E.~Ofori, D.~Vaillancourt, and B.~C. Vemuri.
\newblock Nonlinear regression on riemannian manifolds and its applications to
  neuro-image analysis.
\newblock In \emph{Int. Conf. on Medical Image Computing and Computer-Assisted
  Intervention}, pages 719--727. Springer, 2015.

\bibitem[Bhattacharya et~al.(2012)Bhattacharya, Ellingson, Liu, Patrangenaru,
  and Crane]{extrinsic}
R.~N. Bhattacharya, L.~Ellingson, X.~Liu, V.~Patrangenaru, and M.~Crane.
\newblock Extrinsic analysis on manifolds is computationally faster than
  intrinsic analysis with applications to quality control by machine vision.
\newblock \emph{Appl. Stoch. Models Bus. Ind.}, 28\penalty0 (3):\penalty0
  222--235, 2012.

\bibitem[Bigot et~al.(2017)Bigot, Gouet, Klein, and L{\'o}pez]{geodesic}
J.~Bigot, R.~Gouet, T.~Klein, and A.~L{\'o}pez.
\newblock Geodesic {PCA} in the {W}asserstein space by convex {PCA}.
\newblock In \emph{Ann. inst. Henri Poincare (B) Probab. Stat.}, volume~53,
  pages 1--26. Institut Henri Poincar{\'e}, 2017.

\bibitem[Carlier et~al.(2015)Carlier, Oberman, and Oudet]{carlier2015numerical}
G.~Carlier, A.~Oberman, and E.~Oudet.
\newblock Numerical methods for matching for teams and wasserstein barycenters.
\newblock \emph{ESAIM Math Model Numer Anal}, 49\penalty0 (6):\penalty0
  1621--1642, 2015.

\bibitem[Cazelles et~al.(2018)Cazelles, Seguy, Bigot, Cuturi, and
  Papadakis]{geod_vs_log}
E.~Cazelles, V.~Seguy, J.~Bigot, M.~Cuturi, and N.~Papadakis.
\newblock Geodesic {PCA} versus log-{PCA} of histograms in the {W}asserstein
  space.
\newblock \emph{SIAM J. Sci. Comput.}, 40\penalty0 (2):\penalty0 B429--B456,
  2018.

\bibitem[Chen et~al.(2021)Chen, Lin, and M{\"u}ller]{muller}
Y.~Chen, Z.~Lin, and H.-G. M{\"u}ller.
\newblock Wasserstein regression.
\newblock \emph{J. Am. Stat. Assoc.}, 0\penalty0 (ja):\penalty0 1--40, 2021.

\bibitem[Cuturi and Doucet(2014)]{cuturi2014fast}
M.~Cuturi and A.~Doucet.
\newblock Fast computation of wasserstein barycenters.
\newblock In \emph{Int. Conf. Mach. Learn.}, pages 685--693. PMLR, 2014.

\bibitem[Cuturi et~al.(2022)Cuturi, Meng-Papaxanthos, Tian, Bunne, Davis, and
  Teboul]{jax_ott}
M.~Cuturi, L.~Meng-Papaxanthos, Y.~Tian, C.~Bunne, G.~Davis, and O.~Teboul.
\newblock Optimal transport tools (ott): A jax toolbox for all things
  wasserstein.
\newblock \emph{arXiv preprint arXiv:2201.12324}, 2022.

\bibitem[Delon et~al.(2010)Delon, Salomon, and Sobolevski]{delon2010fast}
J.~Delon, J.~Salomon, and A.~Sobolevski.
\newblock Fast transport optimization for monge costs on the circle.
\newblock \emph{SIAM J. Appl. Math.}, 70\penalty0 (7):\penalty0 2239--2258,
  2010.

\bibitem[Fletcher(2013)]{geod_regression_flet}
P.~Fletcher.
\newblock Geodesic {R}egression and the {T}heory of {L}east {S}quares on
  {R}iemannian {M}anifolds.
\newblock \emph{Int. J. Comput. Vis.}, 105, 11 2013.

\bibitem[Gigli(2011)]{gigli2011inverse}
N.~Gigli.
\newblock On the inverse implication of brenier-mccann theorems and the
  structure of {$(P_2(M),W_2)$}.
\newblock \emph{Meth. Appl. of Anal.}, 18\penalty0 (2):\penalty0 127--158,
  2011.

\bibitem[Hron et~al.(2014)Hron, Menafoglio, Templ, Hrůzová, and
  Filzmoser]{menafoglio}
K.~Hron, A.~Menafoglio, M.~Templ, K.~Hrůzová, and P.~Filzmoser.
\newblock Simplicial principal component analysis for density functions in
  {B}ayes spaces.
\newblock \emph{Comput. Stat. Anal. Data}, 94:\penalty0 330--350, 2014.

\bibitem[Huckemann et~al.(2010)Huckemann, Hotzand, and Munk]{huck_geod}
S.~Huckemann, T.~Hotzand, and A.~Munk.
\newblock Intrinsic shape analysis: {G}eodesic {PCA} for {R}iemannian manifolds
  modulo isometric lie group actions.
\newblock \emph{Stat. Sin.}, 20:\penalty0 1--58, 2010.

\bibitem[Janati et~al.(2020)Janati, Cuturi, and Gramfort]{janati2020debiased}
H.~Janati, M.~Cuturi, and A.~Gramfort.
\newblock Debiased sinkhorn barycenters.
\newblock In \emph{Int. Conf. Mach. Learn.}, pages 4692--4701. PMLR, 2020.

\bibitem[Kim and Pass(2017)]{kim_barycenter}
Y.-H. Kim and B.~Pass.
\newblock Wasserstein barycenters over {R}iemannian manifolds.
\newblock \emph{Adv. Math.}, 307:\penalty0 640--683, 2017.
\newblock ISSN 0001-8708.
\newblock \doi{10.1016/j.aim.2016.11.026}.
\newblock URL \url{https://doi.org/10.1016/j.aim.2016.11.026}.

\bibitem[Lee(2013)]{lee_intro_manifolds}
J.~M. Lee.
\newblock \emph{Introduction to smooth manifolds}, volume 218 of \emph{Graduate
  Texts in Mathematics}.
\newblock Springer, New York, second edition, 2013.
\newblock ISBN 978-1-4419-9981-8.

\bibitem[McCann(2001)]{mccann_polar}
R.~J. McCann.
\newblock Polar factorization of maps on {R}iemannian manifolds.
\newblock \emph{Geom. Funct. Anal.}, 11\penalty0 (3):\penalty0 589--608, 2001.

\bibitem[Munkres(2000)]{munkres}
J.~R. Munkres.
\newblock \emph{Topology}.
\newblock Prentice Hall, Inc., Upper Saddle River, NJ, 2000.
\newblock ISBN 0-13-181629-2.
\newblock Second edition of [ MR0464128].

\bibitem[Panaretos and Zemel(2020)]{panaretos}
V.~M. Panaretos and Y.~Zemel.
\newblock \emph{An Invitation to Statistics in Wasserstein Space}.
\newblock Springer Nature, 2020.

\bibitem[Patrangenaru and Ellingson(2015)]{patra}
V.~Patrangenaru and L.~Ellingson.
\newblock \emph{Nonparametric Statistics on Manifolds and Their Application to
  Object Data Analysis}.
\newblock CRC Press, 2015.

\bibitem[Pegoraro and Beraha(2022)]{projected}
M.~Pegoraro and M.~Beraha.
\newblock Projected statistical methods for distributional data on the real
  line with the wasserstein metric.
\newblock \emph{J. Mach. Learn. Res.}, 23\penalty0 (37):\penalty0 1--59, 2022.

\bibitem[Pennec(2006)]{intrinsic}
X.~Pennec.
\newblock Intrinsic {S}tatistics on {R}iemannian {M}anifolds: {B}asic {T}ools
  for {G}eometric {M}easurements.
\newblock \emph{J. Math. Imaging Vis.}, 25:\penalty0 127--154, 07 2006.

\bibitem[Pennec(2008)]{pennec_manifolds}
X.~Pennec.
\newblock Statistical {C}omputing on {M}anifolds: From {R}iemannian geometry to
  {C}omputational {A}natomy.
\newblock In \emph{LIX Fall Colloquium on Emerging Trends in Visual Computing},
  pages 347--386. Springer, 2008.

\bibitem[Sangalli et~al.(2009)Sangalli, Secchi, Vantini, and
  Veneziani]{aneurisk_jasa}
L.~M. Sangalli, P.~Secchi, S.~Vantini, and A.~Veneziani.
\newblock A case study in exploratory functional data analysis: Geometrical
  features of the internal carotid artery.
\newblock \emph{J. Am. Stat. Assoc.}, 104\penalty0 (485):\penalty0 37--48,
  2009.

\bibitem[Srivastava et~al.(2015)Srivastava, Cevher, Dinh, and
  Dunson]{srivastava15}
S.~Srivastava, V.~Cevher, Q.~Dinh, and D.~Dunson.
\newblock {WASP: Scalable Bayes via barycenters of subset posteriors}.
\newblock In G.~Lebanon and S.~V.~N. Vishwanathan, editors, \emph{Proc. of the
  Eighteenth Int. Conf. on Art. Intel. and Stat.}, volume~38 of \emph{Proc.
  Mach. Learn. Res.}, pages 912--920, San Diego, California, USA, 09--12 May
  2015. PMLR.

\bibitem[Villani(2003)]{VillaniBook}
C.~Villani.
\newblock \emph{Topics in {Optimal} {Transportation}}.
\newblock 2003.
\newblock URL \url{https://bookstore.ams.org/gsm-58}.

\bibitem[Zemel and Panaretos(2019)]{zemel2019frechet}
Y.~Zemel and V.~M. Panaretos.
\newblock Fr{\'e}chet means and procrustes analysis in wasserstein space.
\newblock \emph{Bernoulli}, 25\penalty0 (2):\penalty0 932--976, 2019.

\bibitem[Zhang et~al.(2020)Zhang, Kokoszka, and Petersen]{zhang2020wasserstein}
C.~Zhang, P.~Kokoszka, and A.~Petersen.
\newblock Wasserstein autoregressive models for density time series.
\newblock \emph{arXiv preprint arXiv:2006.12640}, 2020.

\end{thebibliography}

\appendix

\section{Technical Preliminaries}\label{sec:app_preliminaries}

\subsection{Measure Theoretic Preliminaries}

Let $(M,g)$ be a Riemannian manifold of dimension $n$, with $TM$ being its tangent bundle and $TM^*$ its cotangent bundle. 
We know by definition that $g$ is a section $g:M\rightarrow (TM\otimes TM)^*$ and the volume form $\omega:M\rightarrow \wedge^n (TM)^*$ is defined locally by $\omega = \mid \text{det}(g)\mid^{1/2} d x_1\wedge \ldots \wedge d x_n$. 

Let $\mathcal{L}$ be the Lebesgue measure on $\R^n$, 
we consider the $\sigma$-algebra generate by all sets $A$ such that $\varphi(A\cup U)$ is in the Lebesgue $\sigma$-algebra of $\R^n$ for some chart $(U,\varphi)$.
Then we indicate with $\mathcal{L}_M$ the Riemann-Lebesgue volume measure, i.e. the measure on $M$ such that for every chart $(U,\varphi)$ and $A\subset U$ contained in the $\sigma$-algebra just define: 
\begin{equation}
\mathcal{L}_M(A)=
\int_{\varphi(A)} \mid \text{det}(g(\varphi^{-1}))\mid^{1/2} d\mathcal{L} 
\end{equation}
Note that, in general, $\varphi\push \mathcal{L}_M \neq \mathcal{L}$.

Consider $h:M\rightarrow \R$ such that $\text{supp}(h)\subset U$, with $(U,\varphi)$ being a chart, we can integrate $h$ as follows:
\begin{equation}
 \int_M h d\mu = \int_U h d\mu = \int_U h f_\mu d\mathcal{L}_M = \int_{z(U)} \mid \text{det}(g(\varphi^{-1}))\mid^{1/2} h(\varphi^{-1}) f_\mu(z^{-1}) d\mathcal{L}.
\end{equation}
The general case is defined in a natural way through a partition of unity. 

Now we can consider a measure $\mu$ on $M$, with density function $f_\mu$ wrt $\mathcal{L}_M$, that is:
\begin{equation}
\label{eq:measure_on_M}
\mu(A)=\int_A f_\mu d(\mathcal{L}_m) = \int_{\varphi(A)} \mid \text{det}(g(\varphi^{-1}))\mid^{1/2}f_\mu(\varphi^{-1}) d\mathcal{L}. 
\end{equation} 

Lastly, if $\mu$ doesn't have a density function wrt $\mathcal{L}_M$, to integrate some function against $\mu$ we pick a weak converging sequence $\mu_n \rightharpoonup \mu$ such that, for every $n$, $\mu_n$ has a density function and extend the definition taking the limit of the integrals.

\subsection{McCann's Result}

Let us recall the definition of $c$-concavity. Let $c: M \times M \rightarrow \R \cup {+\infty}$. For a function $\psi: M \rightarrow \R \cup \{-\infty\}$ define its $c$-transform $\psi^{c_+}: M \rightarrow \R \cup \{-\infty\}$ as 
\[
	\psi^{c_+}(x) = \inf_{y \in M} c(x, y) - \psi(y).
\]
Note that this generalises the Legendre transform, which is recovered when $M = \R^d$ and $c(x, y) = \langle x, y \rangle$.
\begin{definition}
	A function $\phi: M \rightarrow \R \cup \{-\infty\}$ is $c$-concave if its not identically $-\infty$ and there exists $\psi: M \rightarrow \R \cup \{-\infty\}$ such that 
	\[
		\phi = \psi^{c_+}
	\]
\end{definition}

Given $\mu\in\Wcal_2(M)$ and $U\subset M$ open we define $S(U)=\{v:U\rightarrow TM \mid \pi\circ v = \text{Id}_U\}$ be the sheaf of local sections of the tangent bundle of $M$, that is the vector space of tangent vector fields on $U$. 
Whenever $U$ is a local trivialisation of the tangent bundle, we may  use the notation $v_z:=v(z)\in T_z M$ for $v\in S(U)$.
Now we can define the following sheaf of functions:
\begin{equation}
	L^2_\mu(U) = \lbrace v \in S(U) \mid \int \| v_z \|^2 d\mu(z) <\infty  \rbrace,
\end{equation}
where $\| v_z \|^2$ stands for $g(v_z,v_z)$.

For any $v \in L^2_\mu(U)$ we can consider the map $\exp(v)$ defined as $\exp(v)(z) := \exp_z (v_z)$, $z \in U$.
\cite{mccann_polar} proved that if $\mu$ is absolutely continuous with respect to the volume measure on $M$, the unique optimal plan $\gamma^o$ between $\mu$ and $\nu$ is induced by a map, i.e. we have $T: M \rightarrow M$, inducing $(\text{Id},T):M\rightarrow M\times M$, such that  $\gamma^o = (\text{Id}, T)\push \nu$. 
Moreover, the map $T$ has the form $T = \exp(- \nabla \phi)$ where $\phi$ is a $d^2$-concave function \citep{gigli2011inverse}.

\subsection{More details on $\S_1$}
\label{sec:details_S_1}

We call $\exp_c:(\mathbb{R},+)\rightarrow (\S_1,\cdot)$ the map defined as $\exp_c(x)=e^{ i x}$ and view $\S_1\subset \mathbb{C}$ with the multiplication operation $e^{ i x}\cdot e^{ i y}= e^{ i (x+y)}$. Thus $\exp_c$ is a group morphism: $\exp_c(x+y)=\exp_c(x)\cdot \exp_c(y)$. 
Similarly $log_c:\S_1\rightarrow \R$ given by $\log_c(z)=x\in [0,2\pi)$ such that $z=e^{ i x}$. 
Clearly
$log_c$ is right inverse of $\exp_c$ i.e. $\exp_c\circ \log_c = \text{Id}_{\S_1} $. 

In a similar fashion, we use the projection on the quotient $\pi_c:\R \rightarrow \R/2\pi\mathbb{Z}$ which is a map of groups. This is an alternative, though equivalent representation of $\S_1$ in the following sense: we have that $\exp_c(x)=\exp_c(\pi_c(x))$ and so $\log_c= (\exp_c\circ \pi_c)^{-1} $ are group isomorphisms. We employ the metric induced by this alternative representation of the circle in the proofs - see \Cref{eq:dist_non_exp}.

Lastly for any $z\in \S_1$ we define shifted versions of the maps $\exp_c$ and $\log_c$, centred in $z$: $\exp_z(x)=\exp_c(x+\log_c(z))$ and $\log_z(z')=x'$ with $x'\in [-\pi/2,\pi/2)$ such that $\exp_z(x')=z'$. 
If we call $W:=(-\pi/2,\pi/2)$ and $V_z:=\S_1-\{-z\}$ then for every $z\in\S_1$, the couple $(V_z, \log_z)$ is a local chart which gives a homeomorphism on $W$.
With this differential structure $\S_1$ is a Lie Group and its tangent bundle is $T\S_1 = \{(x,v)\mid x\in \S_1 \text{ and } v\in T_x\S_1\}\simeq \S_1\times \R$. We call $1$ the point $(1,0)$ which gives the neutral element in $\S_1$.

We can make the following observations: the Riemannian metric $g:\S_1\rightarrow (T\S_1\otimes T\S_1)^*$ is
 induced by the embedding $\S_1\hookrightarrow \mathbb{C}\simeq \mathbb{R}^2$.
In local coordinates $(V_{1}, log_1)$ centred in $1\in\S_1$ the embedding is
$x\in (-\pi/2,\pi/2) \mapsto e^{i x} = (\cos(x),\sin(x))$. 
The map between tangent spaces is therefore 
$d/dx \mapsto -\sin(x)d/dx_1 + \cos(x)d/dx_2$.
Let $g_{\mathbb{E}}$ be the euclidean metric in $\R^2$.
In canonical coordinates of $\R^2$, $x_1,x_2$, this metric is clearly given by the quadratic form $\text{Id}_{2\times 2}$. Thus $g(d/dx)=g_E(-\sin(x),\cos(x))=\sin^2(x)+\cos^2(x)=1$. In other words in local coordinates $(V_{1},\log_1)$, $g(v,w)=v\cdot w$ for $v,w\in T_1\S_1\simeq \R$. 
With this metric, $\exp_z:T_z\S_1\rightarrow \S_1$ and $\log_z:\S_1\rightarrow T_z\S_1$ are respectively the Riemannian exponential and logarithm.

\subsection{Vector Fields on $\S_1$}

Consider now $ \mu\in \mathcal{P}(\S_1)$ and a real valued function $f:\S_1\rightarrow \R$ which we would like to integrate against $\mu$.
We saw that $det(g)\equiv 1$ and the volume form $\omega$ locally is $\omega = dx$. This immediately implies that 
$\mathcal{L}_{\S_1} = \exp_c \push \mathcal{L}$ or, equivalenty, $\log_c\push\mathcal{L}_{\S_1} = \mathcal{L}$.
Thus:
\begin{equation}
\int_{\S_1} f(z) d\mathcal{L}_{\S_1}(z) = \int_{[-\pi/2,\pi/2)}  f(\exp_c(x)) d\mathcal{L}(x)
\end{equation}

Along this line, the sheaf of tangent vector fields $L^2_\mu$ can be easily transported on $[0,2\pi)$ with the change of variables:
\begin{equation}
\int_{\S_1} \mid v_z \mid^2 d\mu(z) = 
\int_{[0,2\pi)} \mid v(\exp_c(x))\mid^2 d\log_c\push \mu(x)
\end{equation}

In fact from \Cref{eq:measure_on_M} plus the observation that $\mathcal{L}_{\S_1} = \exp_c \push \mathcal{L}$
we have that a measure $\mu$ on $\S_1$ is equivalently represented by the measure $\log_c\push \mu$ extended to a periodic measure $\mutilde$ on $\R$ as follows: for any measurable $A$ and $p \in \mathbb{Z}$, 
$\mu(\exp_c(A)) =: \mutilde(A) = \mutilde(A + p)$

\section{Proofs}

Let us introduce some notation, motivated by \Cref{sec:details_S_1}.
 
\begin{equation}\label{eq:dist_non_exp}
	d_\mathbb{Z}(x, y)^2 := \inf_{p \in 2\pi \mathbb{Z}} (x - y - p)^2\leq (x-y)^2
\end{equation}
Note that $d_R(z,z')=d_\mathbb{Z}(\log_c(z),\log_c(z'))$

\bigskip

\subsection{Proof of Theorem \ref{teo:opt_map}}
\begin{proof}
	The proof follows from the notion of locally optimal plans in \cite{delon2010fast}.
	Let $\gamma_\theta$ be the transport plan that takes an element of mass from position $\quant_\nutilde(u)$ to position $(F^\theta_\mutilde)^{-}(u)$. Then $\gamma_\theta$ is locally optimal and the associated cost is
	\[
		C_{[\mu,\nu]}(\theta) = \int_0^1 \left(\quant_\mutilde(u)  - (F^\theta_\nutilde)^{-}(u)\right)^2 \dd u
	\]
	The (global) optimal plan is associated to $\theta^* = \argmin C(\theta)=W^2_2(\mu,\nu)$. 
	To recover the optimal transport map we operate the change of variables $x = F^-_\mutilde(u)$, which yields:
	
\begin{align}
		&W^2_2(\mu,\nu) = \int_{0}^{2\pi} \left(T^\nutilde_\mutilde(x)  - x\right)^2 \dd\mutilde(x)\geq \\
		&\int_{0}^{2\pi} d^2_\mathbb{Z}(T^\nutilde_\mutilde(x),x) \dd\mutilde(x)=\\
		&\int_{\S_1} d^2_R(\exp_c(T^\nutilde_\mutilde(\log_c(z))) ,z)^2 \dd\mu(z)\geq 
		W^2_2(\mu,\nu)
\end{align}
where the first equality follows by defining $T_{\mutilde}^\nutilde := (F^{\theta^*}_\nutilde)^{-} \circ F_\mutilde$, while the last equality is obtained with $z=\exp_c(x)$ and the properties of $d_\mathbb{Z}$.
\end{proof}

\subsection{Proof of \Cref{teo:opt_theta}}

\begin{proof}
	First we observe that:
	\[
	(F^{\theta^*}_\nutilde)^{-} \left( F_\mutilde(u+2\pi p) \right)=(F^{\theta^*}_\nutilde)^{-} \left( F_\mutilde(u)+p \right)=(F_\nutilde)^{-} \left( F_\mutilde(u)+p-\theta^* \right)=(F^{\theta^*}_\nutilde)^{-} \left( F_\mutilde(u) \right)+2\pi p
	\]
	
	which means that $\widetilde T(u+2\pi p)= \widetilde T(u)+2\pi p$ for every $p\in \mathbb{Z}$.

	By \Cref{teo:opt_map} we know that $\widetilde T$ is an optimal transport map if and only if $\theta=\theta^*$ as in \Cref{eq:theta_cost}.
	Define:
	\[
	C_{[\mu,\nu]}(\theta)= \int_0^1 \left(\quant_\mutilde(u)  - (F^\theta_\nutilde)^{-}(u)\right)^2 \dd u
	\]
	\cite{delon2010fast} prove that the map $\theta \mapsto C_{[\mu,\nu]}(\theta)$ is strictly convex if $\mu$ is a.c.. Thus $\theta^*$ is the unique stationary point of the function. For this reason we compute the derivative of $C_{[\mu,\nu]}$ in $\theta$, knowing that $\theta^*$ is the only value such that $(C_{[\mu,\nu]})' =0$.
	Thanks to Leibniz rule we can write:
	\begin{align*}
		\frac{\dd}{\dd \theta} C_{[\mu,\nu]}(\theta) &= 
	\int_0^1 \frac{\dd}{\dd \theta} \left(\quant_\mutilde(u)  - (F_\nutilde)^{-}(u-\theta)\right)^2 \dd u \\
	& = \int_0^1 \frac{\dd}{\dd \theta} \left(\quant_\mutilde(u+\theta)  - (F_\nutilde)^{-}(u)\right)^2 \dd u \\
	& = \int_0^1 2\left(\quant_\mutilde(u+\theta)  - (F_\nutilde)^{-}(u)\right)\frac{1}{f_\mutilde(\quant_\mutilde(u+\theta))}  \dd u
	\end{align*}
	with the change of variables $v=\quant_\mutilde(u+\theta)$, which entails $F_\mutilde(v)-\theta=u$ and $\dd u=\dd\mutilde (v)=f_\mutilde(v)\dd v$ we obtain
	\[
	= \int_{v_0}^{2\pi+v_0} 2\left(v  - \widetilde T(v)\right)\frac{f_\mutilde(v)}{f_\mutilde(v)}  \dd v = -2\int_{v_0}^{2\pi} \left(\widetilde T(v)  - v\right)\dd v -2 \int_{2\pi}^{2\pi+v_0} \left(\widetilde T(v)  - v\right)\dd v 
	\]
	with $v_0=\quant_\mutilde(\theta)$.

	Via the change of variables $u=v-2\pi$ we obtain:
	\[
	\frac{\dd}{\dd \theta} C_{[\mu,\nu]}(\theta)=-2\int_{0}^{2\pi} \left(\widetilde T(u)-2\pi  - (u-2\pi)\right)\dd v 
=	-2\int_{0}^{2\pi} \left(\widetilde T(u)  - u\right)\dd u .
	\]
	Optimality follows if and only if such quantity is equal to zero and thus:
	\[
	0 = \int_{0}^{2\pi} \left(\widetilde T(u)  - u\right)\dd u. 
	\]
	\end{proof}

\subsection{Proof of Theorem \ref{teo:opt_map2}}

\begin{proof}
	Observe that $T$ is an optimal transport map, then defining 
	\[
		\widetilde{T}(x) :=  \log_c \circ T \circ \exp_{c}, \, x \in (0, 2\pi), \qquad \widetilde{T}(x + p) = \widetilde{T}(x) + p, \, p \in 2\pi\mathbb{Z}
	\]
	clearly satisfies the monotonicity and ``periodicity'' requirements. Moreover, \eqref{eq:int_condition} is satisfied by Theorem \ref{teo:opt_theta}.
	To prove that $|\widetilde{T}(x) - x| < \pi/2$ note that this is equivalent to $\inf_{p\in 2\pi\mathbb{Z}} \mid \widetilde{T}(x)-x-p\mid = \mid \widetilde{T}(x) - x\mid $.
	Since $T:\S_1\rightarrow \S_1$ is an optimal transport map from $\mu$ to $\nu$ we have:
	\begin{align*}
		W_2(\mu,\nu)& =\int_{\S_1}d_R(T(z),z)^2 d\mu \\
		&= \int_{[0,2\pi]}d_{R}(T(\exp_c(x)),\exp_c(x))^2d(\log_c\push \mu)(x) \\ 
		&= \int_{[0,2\pi]}d_{\mathbb{Z}}(\widetilde{T}(x),x)^2d\mutilde(x)
	\end{align*}
	where the first equality is obtained via the definition of optimal transport map, the second through the change of variables $z = \exp_{c\mid [0,2\pi]}(x)$, and in the last one we use the definition of $\widetilde{T}$, $\mutilde$ and the properties of $d_\mathbb{Z}$.
	As already noted, we have $\inf_{p\in\mathbb{Z}} \mid \widetilde{T}(x)-x-p\mid \leq \mid \widetilde{T}(x) - x\mid $. If the strict inequality holds for some $A\subset [0,2\pi]$ with $\mutilde(A)>0$ then also the integrals on $[0,2\pi]$ must be different, and the thesis follows.
	
	To prove the reverse statement, it suffices to prove that $\widetilde{T}(v)$ can be written as $\quant_{\nutilde}(F_\mutilde(v) + \theta)$, which is equivalent to saying that 
	\[
		\quant_\nutilde(v) = \widetilde{T}\left((F_\mutilde(\cdot) + \theta)^{-1}(v)\right) = \widetilde{T}(\quant_\mutilde(v - \theta)).
	\]
	Define $G_{\nutilde} := \widetilde{T} \circ \quant_\mutilde$, then of course $\widetilde{T} = G_\nutilde \circ F_\mutilde$. We show that $G_\nutilde(u) \equiv \quant_\nutilde(u + \theta)$.
	We have that, for $x \in [0, 2\pi)$ 
	\begin{align*}
		F_\nutilde(x) = \nutilde([0, x]) = \mutilde(\Ttilde^{-1}([0, x])) = \mutilde([\Ttilde^{-1}(0), \Ttilde^{-1}(x)]) = F_\mutilde(\Ttilde^{-1}(x)) - F_\mutilde(\Ttilde^{-1}(0))
	\end{align*}
	and observe that $F_\nutilde(x) \leq 1$ thanks to $|\Ttilde(x) - x| < \pi/2$. Hence, the pushforward of $\mutilde$ on $\S_1$ gives a valid probability measure.
	Taking the inverse of $F_\nutilde$ 
	\[
		\quant_\nutilde(u) = \left( F_\mutilde(\Ttilde^{-1}(\cdot)) - F_\mutilde(\Ttilde^{-1}(0))\right)^{-}(u) = \Ttilde \circ \quant_\mutilde(u + F_\mutilde(\Ttilde^{-1}(0)))
	\]
	and setting $- \theta = F_\mutilde(\Ttilde^{-1}(0))$ yields the result.
	\end{proof}

\subsection{Proof of Theorem \ref{teo:continuity}}

To prove item (ii), we will need the two following preliminary lemmas.

\begin{lemma}\label{lemma:tilde_conv}
	Suppose we have $W_2(\nu,\nu_n)\rightarrow 0$ in $\Wcal_2(\S_1)$ with $\nu,\nu_n$ being a.c. wrt $\mathcal{L}_{\S_1}$ (for every $n$). Then $W_2(\log_c\push\nu,\log_c\push\nu_n)\rightarrow 0$ in $\Wcal_2(\R)$.
\end{lemma}	
	
	\begin{proof}
	From Theorem 7.12 in \cite{VillaniBook}, convergence in the Wasserstein metric is equivalent to weak convergence plus the tightness condition: there exist $x_0$ such that
	\[
		\lim_{R \rightarrow +\infty} \limsup_{k \rightarrow +\infty} \int_{d(x, x_0) > R} d(x, x_0)^p d\log_c\push\nu_k(x)
	\]
	Observe that each measure $\log_c\push\nu_k$ is supported on $[0, 2\pi]$ so that the condition is always met.
	Hence, we just need to show that the sequence $\log_c\push\nu_k$ converges weakly.
	For measures on the real line, weak convergence is equivalent of pointwise convergence of the associated distribution functions at continuity points.
	That is, letting $F_k(x) := \log_c\push\nu_k([0, x)) $ and $F(x) := \log_c\push\nu([0, x))$, it must hold that
	\begin{equation}\label{eq:conv_cdf}
		F_k(x) \rightarrow F(x), \qquad \text{all $x$ such that $F(x)$ is continuous}
	\end{equation}
	Observe that $F_k(x) = \nu_k(\exp_c([0, x)))$ by definition. By Portmanteau's theorem, for any $x$ such that $\nu(\{\exp_c(x)\}) = 0$ we have that $\nu_k(\exp_c([0, x))) \rightarrow \nu_k(\exp_c([0, x)))$ which easily implies $\eqref{eq:conv_cdf}$

	\end{proof}

\begin{lemma}
	Suppose we have $W_2(\nu,\nu_n)\rightarrow 0$ with $\mu,\nu,\nu_n$ being a.c. wrt $\mathcal{L}_{\S_1}$ (for every $n$).
	Then $\parallel \log_\mu (\nu_n)-\log_\mu (\nu)\parallel_{L^2_\mu} \rightarrow 0$.
\end{lemma}
	
	\begin{proof}
	By \Cref{lemma:tilde_conv} we have $F_{\nutilde_n}(x)\rightarrow F_{\nutilde}(x)$ and the same for the quantile functions. 
	As a consequence $C_{[\nu_n,\mu]}(\theta)\rightarrow C_{[\nu,\mu]}(\theta)$.

	Thus consider $\theta_n = \arg\min C_{[\nu_n,\mu]}$. By the discussion in \Cref{sec:OTM_S1} and in particular \Cref{eq:OTM_shift}, we have that the minimisation domain of $C_{[\nu_n,\mu]}$ can be restricted to a sufficiently large compact interval $K$.  	
Since $\{\theta_n\}\subset K$ compact, we can consider a converging subsequence which we still call $\{\theta_n\}$ with an abuse of notation. Let $\theta_n\rightarrow \theta^*$. Recall that $C_{[\nu_n,\mu]}$ \citep{delon2010fast} is strictly convex. Thus, by standard arguments,
	we conclude that $\theta^*=\arg\min C_{[\nu,\mu]}$.

	Now consider:
	\begin{align}
	&\mid \left(F_\mutilde(\quant_{\nutilde_n}(u+\theta_n))\right) - \left(F_\mutilde(\quant_\nutilde(u+\theta^*))\right)\mid \leq \\
	&\mid \left(F_\mutilde(\quant_{\nutilde_n}(u+\theta_n))\right) - \left(F_\mutilde(\quant_\nutilde(u+\theta_n))\right)\mid+
	\mid \left(F_\mutilde(\quant_{\nutilde}(u+\theta_n))\right) - \left(F_\mutilde(\quant_\nutilde(u+\theta^*))\right)\mid 
	\end{align}
	
	Which implies the pointwise convergence $T^\mutilde_{\nutilde_n}(u)\rightarrow T^\mutilde_\nutilde(u)$: both addends in the last sum go to $0$. Since these maps are continuous and bounded on $[0,2\pi]$ we have uniform convergence and strong convergence. 
	The strong convergence in the image of $\log_\mu$ then follows.
	\end{proof}

\bigskip

We are now ready to prove Theorem \ref{teo:continuity}.

\begin{proof}
\begin{enumerate}
\item To check the continuity of $\exp_\mu$, consider $T_\mu^{\nu_1}\times T_\mu^{\nu_2}:\S_1\rightarrow \S_1 \times \S_1$ and induce the transport plan $\gamma=(T_\mu^{\nu_1},T_\mu^{\nu_2})\push \mu$.
	Then we have:
	\begin{align*}
		W_2^2(\nu_1, \nu_2) &\leq \int_{\S_1\times \S_1} 
		d_R(z,w)^2 d\gamma(dzdw)\\
		&= \int_{\S_1} 
		d_R(T_\mu^{\nu_1}(z),T_\mu^{\nu_2}(z))^2 d\mu(dz) \\
		&\leq \int_{[0,2\pi]} 
		d_\mathbb{Z}(T_\mutilde^{\nutilde_1}(x),T_\mutilde^{\nutilde_2}(x))^2 d\mutilde(dx) \\
		& \leq \parallel T_\mutilde^{\nutilde_1}-T_\mutilde^{\nutilde_2}\parallel^2_{L^2_\mutilde([0,2\pi])}=				\parallel \log_\mu(\nu_1)-\log_\mu(\nu_2) \parallel^2_{L^2_\mu},	
	\end{align*}
where the last identity is obtained thanks to $\log_c\push \mu =\mutilde$ on $[0,2\pi]$.

\item To check the continuity of $\log_\mu$ instead, by an approximation argument we obtain sequential continuity of $\log_\mu$ at any measure $\nu\in \Wcal_2(\S_1)$:
consider $\nu_n\rightarrow \nu$, with $\nu_n$ a.c. measures. Then $\{\log_\mu(\nu_n)\}$ is a Cauchy sequence in $L^2_\mu$, which is a complete metric space, and so it converges to a vector field $v$. Consider $\exp_\mu(v)$. By the continuity of $\exp_\mu$ we have $\nu_n\rightarrow \exp_\mu(v)$ which then entails $\exp_\mu(v)=\nu$.

Lastly, sequential continuity in metric spaces implies continuity.

\end{enumerate}
\end{proof}	

\subsection{Proof of \Cref{prop:weaker_cont}}
Theorem 3.2 in \cite{ambrosio2019optimal} ensures that, in the hypotheses of the proposition, $\int_{\S_1}d_R(T_\mu^{\nu_n}(x),T_\mu^\nu(x))^2 d\mu(x)\rightarrow 0$.

Now we prove the following lemma which ends the proof.

\begin{lemma}
	Suppose we have $\mu,\nu,\nu_n$ being a.c. wrt $\mathcal{L}_{\S_1}$ (for every $n$). And suppose that the following hold.
 	\begin{align*}
 	\int_{\S_1}d_R(T_\mu^{\nu_n}(x),T_\mu^\nu(x))^2 d\mu(x)\rightarrow 0.
 	\end{align*}

	Then $\parallel \log_\mu (\nu_n)-\log_\mu (\nu)\parallel_{L^2_\mu} \rightarrow 0$.
\end{lemma}

	\begin{proof}
	
	To simplify the notation, call $G:=\log_\mu(\nu)$ and $G_n := G:=\log_\mu(\nu_n)$. So that $G(x)= \log_x(T_\mu^\nu(x))$ and $G_n(x)= \log_x(T_\mu^{\nu_n}(x))$. On top of that define $f_n:\S_1\rightarrow \R$ as $f_n(x)=G(x)-G_n(x)$. Note that $f_n$ is bounded.

	We can write:
 	\begin{align*}
 	\int_{\S_1}d_R(T_\mu^{\nu_n}(x),T_\mu^\nu(x))^2 d\mu(x) = \int_{A^+_n}\mid f_n(x)\mid^2 d\mu(x) + \int_{A^-_n}\mid f_n(x)+2\pi \mid^2 d\mu(x).
 	\end{align*}
	
	Where $A_n^+ = f_n^{-1}([-\pi,\pi])$ and $A_n^- = f_n^{-1}([\pi,2\pi]) \cup f_n^{-1}([-2\pi,-\pi])$.	
	
	We want to prove that:	
 	\begin{align*}
 	\int_{\S_1}\mid f_n(x)\mid^2 d\mu(x) = \int_{A^+_n}\mid f_n(x)\mid^2 d\mu(x) + \int_{A^-_n}\mid f_n(x) \mid^2 d\mu(x)\rightarrow 0,
 	\end{align*}	
 	
	and thus we need to work on the integral $\int_{A^-_n}\mid f_n(x) \mid^2 d\mu(x)$ as we already know $\int_{A^+_n}\mid f_n(x)\mid^2 d\mu(x) \rightarrow 0$. 	
 	
 	We want to show that $\mu(A_n^-)\rightarrow 0$.

 	Reasoning by contradiction, suppose that there exist $\varepsilon>0$ such that for every $N>0$ there is $n>N$ satisfying $\mu(A_n^-)>\varepsilon$. For an ease of notation, instead of taking a subsequence $\{n_k\}$ such that 
 	$\mu(A_{n_k}^-)>\varepsilon$, we just suppose that it holds for every $n$.
 	
 	 Consider $c>0$, $c\in\R$.  The set $B_{n,c}^+:= f_n^{-1}([c,+\infty)) \cap A_n^-$ must satisfy $\mu(B_{n,c}^+)\rightarrow 0$. In fact:
 	\begin{align*}
 	\int_{A_n^-}\mid f_n(x)+2\pi\mid^2 d\mu(x) = \int_{B_{n,c}^+}\mid f_n(x)+2\pi\mid^2 d\mu(x) + \int_{B_{n,c}^-}\mid f_n(x)+2\pi \mid^2 d\mu(x),
 	\end{align*} 	
 	and $\int_{B_{n,c}^+}\mid f_n(x)+2\pi\mid^2 d\mu(x)\geq c\cdot \mu(B_{n,c}^+)$.
	
	If $x\in B_{n,c}^+$  then either 1) $G(x) \in [-\pi,-\pi+c]$ and $G_n(x) \in [\pi-c,\pi]$ or 2) $G_n(x) \in [-\pi,-\pi+c]$ and $G(x) \in [\pi-c,\pi]$. At least one between 1) or 2) must hold an infinite number of times. WLOG 1) holds a countable number of times. Again, instead of taking a subsequence $\{n_k\}$ such that 
 	1) always holds, we just suppose that it holds for every $n$.  	 	
 	
 	In other words, for any $c>0$, $c\in\R$ we have that: for any $n$, for every  $x\in B_{n,c}^+$,  $G(x) \in [-\pi,-\pi+c]$. And $B_{n,c}^+$ is always a set of positive measure. By the compactness of $\S_1$ we can find a set $B$ with positive measure such that $G(x)=-\pi$ for all $x\in B$. Thus $T_\mu^\nu$ is an OTM that sends each point of a set of full measure, into its antipodal point on the circle. This is absurd as it violates the cyclical monotonicity condition. Thus $\mu(A_n^-)\rightarrow 0$.
 	
 	The following concludes the proof.
 	\begin{align*}
 	\int_{\S_1}\mid f_n(x)\mid^2 d\mu(x) =& \int_{A^+_n}\mid f_n(x)\mid^2 d\mu(x) + \int_{A^-_n}\mid f_n(x) \mid^2 d\mu(x)\\
 	\leq& \int_{A^+_n}\mid f_n(x)\mid^2 d\mu(x) + 4\pi^2\mu(A^-_n)\rightarrow 0.
 	\end{align*} 	
	\end{proof}

\subsection{Proof of \Cref{prop:conv_bary}}\label{app:conv_vary}

Our proof relies on the multi-marginal formulation of the Wasserstein barycentre in \cite{agueh_barycenter,panaretos}.
In the following, consider $\mu_1, \ldots, \mu_n$ in $\Wcal(M)$ where $(M, d)$ is connected and compact manifold. 
Let $\Pi(\mu_1, \ldots, \mu_n)$ be the set of probability measures on $M^n := M \times \cdots \times M$ having marginals $\mu_1, \ldots, \mu_n$.
The multi-marginal problem is to minimise
\[
	G(\pi) = \frac{1}{2n^2} \int_{M^n} \sum_{i < j} d(x_i, x_j)^2 \dd \pi(x_1, \ldots, x), \quad \pi \in \Pi(\mu_1, \ldots, \mu_n).
\]

As shown in \cite{kim_barycenter} (Theorem 2.4),  minimising $G(\pi)$ is equivalent to minimising $F(\nu)$ in \eqref{eq:frechet_func}. Indeed, let $\bar x: M^n \rightarrow M$ such that
\[
	x_1, \ldots, x_n \mapsto \argmin_{z \in M} \sum_{i=1}^n d^2(x_i, z)
\]
then the optimal multicoupling $\pi^o = \argmin_{\pi \in \Pi} G(\pi)$ gives the minimiser of the Frech\'et functional via the rule
\begin{equation}\label{eq:bary_from_coupling}
	\bar \mu := \bar x \push \pi^o.
\end{equation}

\begin{lemma}\label{lemma:conv_bary}
	Let $(M, d)$ be a connected compact Riemannian manifold whose exponential map $\exp_M$ is non expansive. Denote by $\log_M$ the associated logarithmic map.
	Let $\mu^*$ be an absolutely continuous measure in $\Wcal_2(M)$ and $\mu_1, \dots, \mu_n \in \Wcal_2(M)$.
	Assume that, for any $i, j = 1, \ldots, n$,
	\[
		\|\log_{\mu^*}(\mu_i) - \log_{\mu^*}(\mu_j)\|^2_{L^2_{\mu^*}} = W^2_2(\mu_i, \mu_j),
	\]
	where $T_{\mu^*}^{\mu_i} = \exp_M \circ \log_{\mu^*}(\mu_i)$.
	Letting $\bar T = n^{-1} \sum_{i=1}^n\log_{\mu^*}(\mu_i)$, then the Wasserestein barycentre of $\mu_1, \ldots, \mu_n$ is $\bar T \push \mu^*$.
\end{lemma}
\begin{proof}
	Consider the multicoupling $\pi^* = (T_{\mu^*}^{\mu_1}, \ldots, T_{\mu^*}^{\mu_n}) \push \mu^*$. Of course, $\pi^* \in \Pi(\mu_1, \ldots, \mu_n)$ by construction, and $G(\pi^*) \geq G(\pi^o)$.
We observe that, by the non-expansiveness of the exponential map,
\begin{align*}
	\int_{M^n} d^2(x_i, x_j) \dd \pi^*(x_1, \ldots, x_n) &= \int_{M} d^2(T_{\mu^*}^{\mu_i}(x), T_{\mu^*}^{\mu_j}(x)) \dd \mu^*(x) \\
	& \leq \|\log_{\mu^*}(\mu_i) - \log_{\mu^*}(\mu_j)\|^2_{L^2_{\mu^*}} \\
	& = W^2_2(\mu_i, \mu_j),
\end{align*}
where the first inequality follows from the first point of \Cref{teo:continuity} and the last equality by hypothesis.
Hence, $\pi^*$ minimises $G(\pi)$ and the result follows from \eqref{eq:bary_from_coupling}.
\end{proof}

The proof of \Cref{prop:conv_bary} follows from \Cref{lemma:conv_bary} and the fact that the exponential map of $\S_1$ is indeed non-expansive, see \eqref{eq:dist_non_exp}.

\section{Additional Simulation}\label{app:further_simu}

We report here an additional simulation study for the PCA, with the goal of showing the effect of the base point $\bar \mu$.
We consider truncated Gaussian measures on $[0, 2\pi)$ (and extended periodically over $\R$), parametrised by the mean parameter $m$ and scale parameter (square root of the variance) $s$. 
We simulate $n=100$ datapoints by sampling $m \sim \mathcal{U}(0.3 \pi, 1.7 \pi)$ and $s \mid m \sim \mathcal U(\pi / 20, \pi/10) I_{(0, 2 \pi)}(m \pm 3 s)$.
\Cref{fig:gaussian_data} (top row) shows the data and the barycentres in $\Wcal_2(\R)$ and  $\Wcal_2(\S_1)$. In particular, the barycentre in  $\Wcal_2(\R)$ is unimodal and centred on the domain. Also the barycentre in  $\Wcal_2(\S_1)$ presents its tallest mode at the centre of the domain, but can be seen to be trimodal, with, in particular, a mode around the origin.
As already pointed out in \Cref{sec:numerical_ill}, this is due to the fact that the geodesics between some measures, force mass to travel across $0$ around the circle.

\begin{figure}[t]
	\centering
	\begin{subfigure}{0.5\linewidth}
		\centering
		\includegraphics[width=\linewidth]{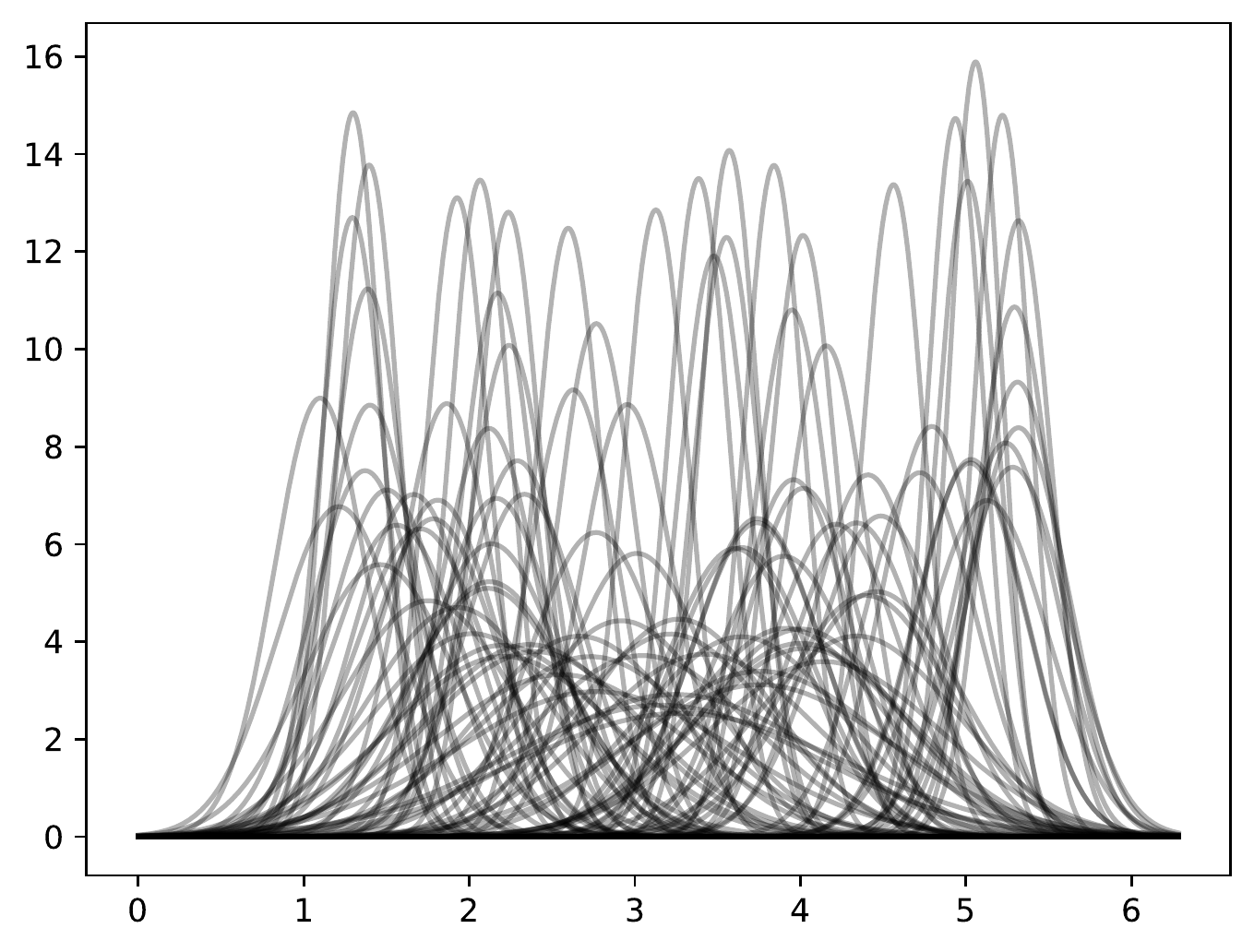}
	\end{subfigure}%
	\begin{subfigure}{0.5\linewidth}
		\centering
		\includegraphics[width=\linewidth]{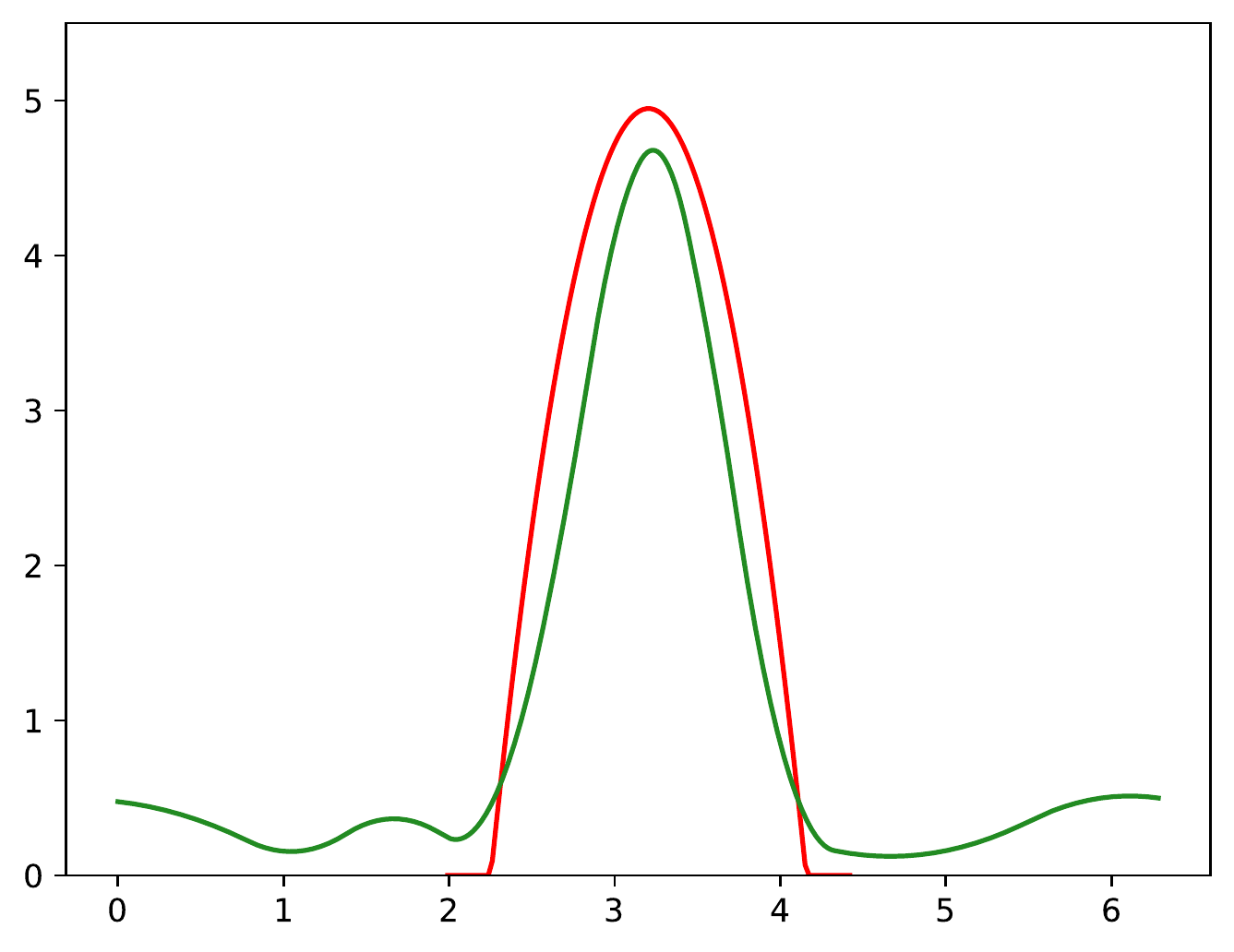}
	\end{subfigure}
	\begin{subfigure}{0.5\linewidth}
		\centering
		\includegraphics[width=\linewidth]{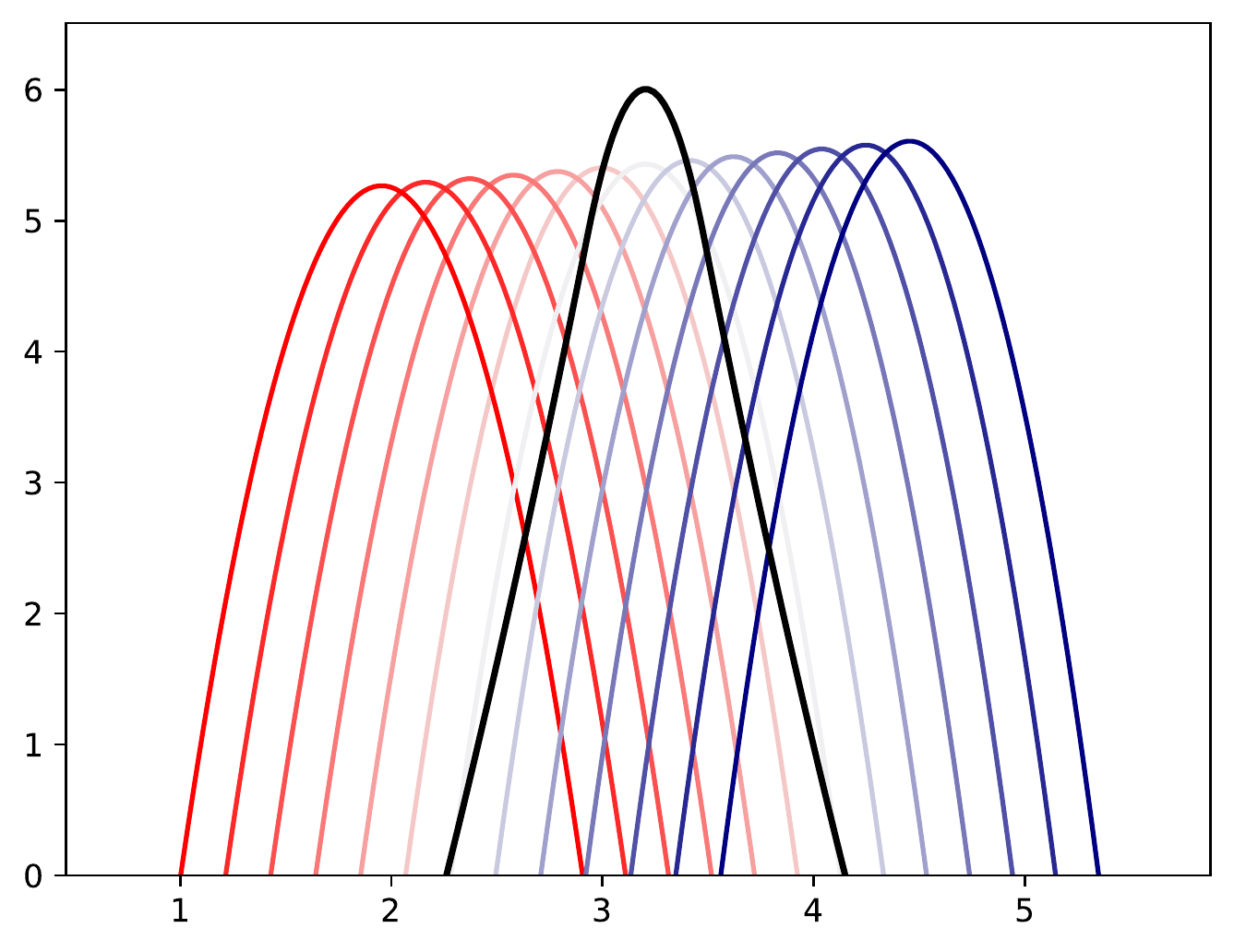}
	\end{subfigure}%
	\begin{subfigure}{0.5\linewidth}
		\centering
		\includegraphics[width=\linewidth]{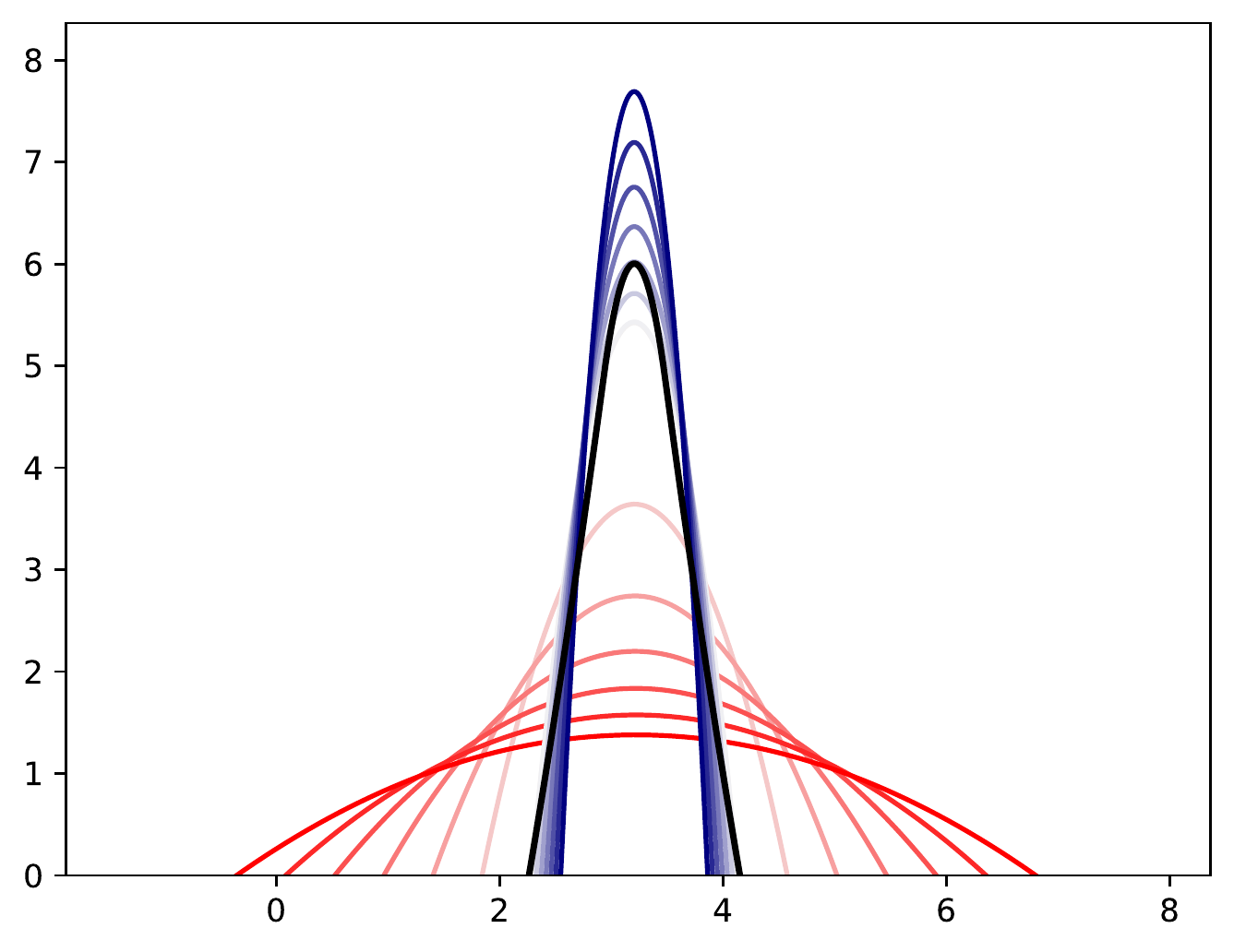}
	\end{subfigure}
	\caption{Top row: Simulated data (left) and comparison between Wasserstein barycentres (right): the red line is the barycentre in $\Wcal_2(\R)$ while the green one is the barycentre in $\Wcal_2(\S_1)$.
	Bottom row: first two principal directions in $\Wcal_2(\R)$ for the Gaussian measures, the solid black line is the barycentre.}
	\label{fig:gaussian_data}
\end{figure}

\begin{figure}[t]
	\centering
	\begin{subfigure}{\linewidth}
		\centering
		\includegraphics[width=\linewidth]{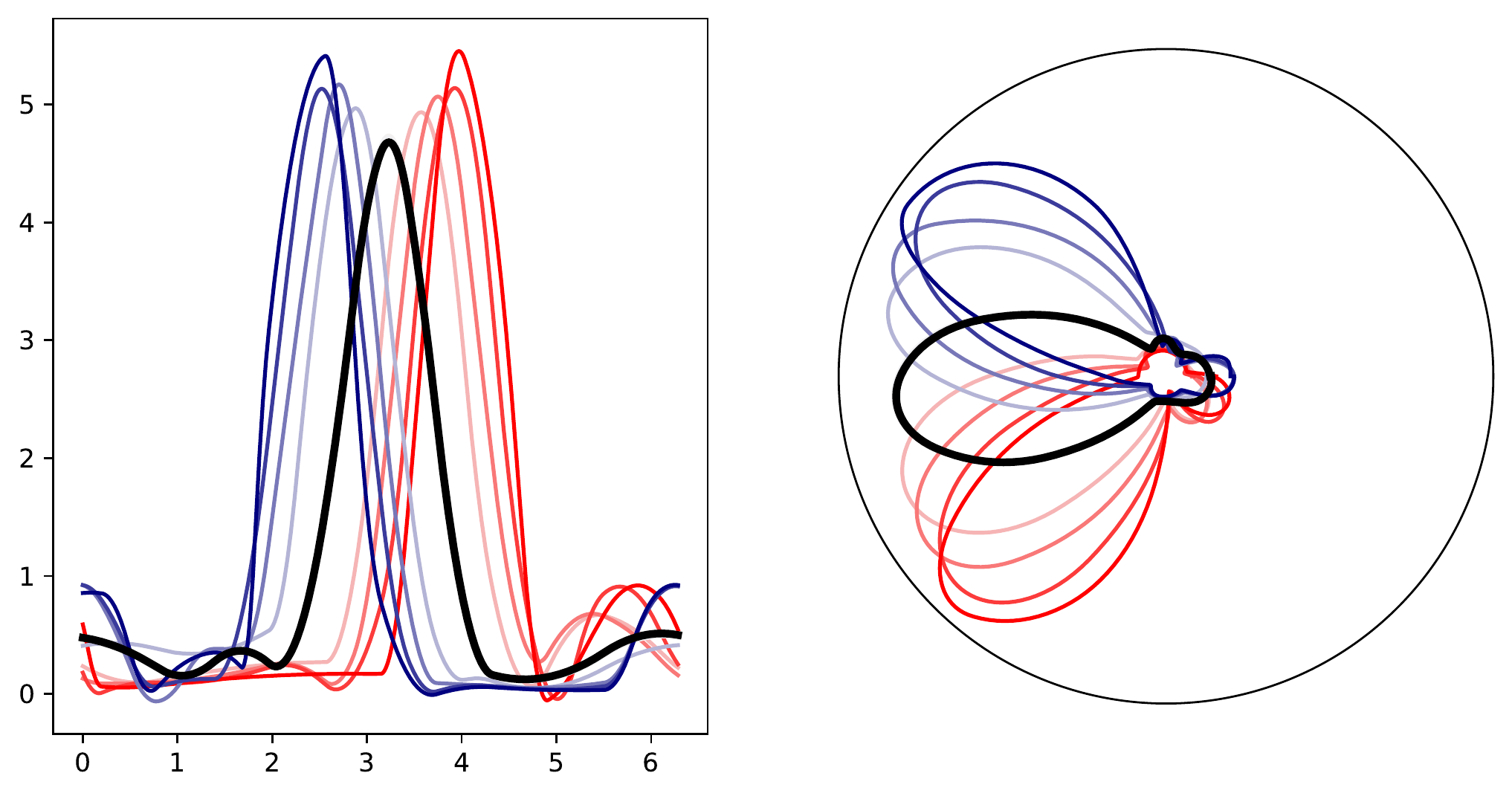}
	\end{subfigure}
	\begin{subfigure}{\linewidth}
		\centering
		\includegraphics[width=\linewidth]{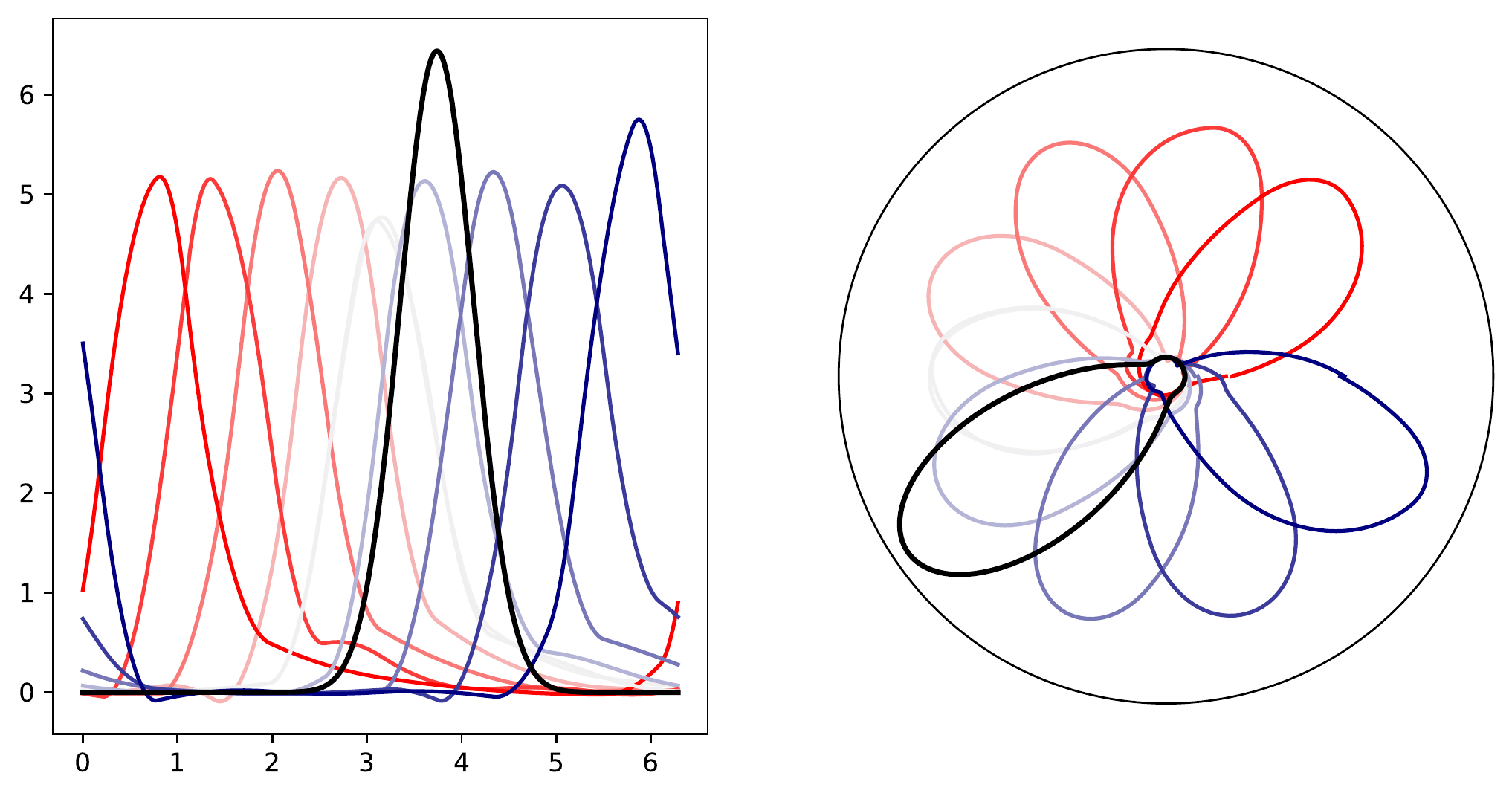}
	\end{subfigure}
	\caption{Top row: 
	Bottom row: }
	\label{fig:gaussian_pds}
\end{figure}

The first two principal directions in $\Wcal_2(\R)$ are displayed in \Cref{fig:gaussian_data} (bottom row), these clearly separate the effect of the location and the one of the scale as expected. 
When performing PCA in $\Wcal_2(\S_1)$, the scale and location's effects are not clearly separated as shown in Figure \ref{fig:gaussian_pds} (top row): the mass located at the minor modes in the barycentre needs to be moved to match the unimodal measures we generated.
In particular, moving along the first principal direction results in 
a less pronounced change in the location, compared to the PCA on $\R$ and
in densities having a more evident mode around the origin. Note that the mode close to $0$ tends to get closer to the main mode. Both these effects combine with the second principal direction (not shown in the plots), to approximate the unimodal distributions in the data set.

We also consider a different point $\bar \mu$ where to centre the PCA, namely one of the observations, displayed in the solid black line in \Cref{fig:gaussian_pds} (bottom). In this case, we set $\log_{\bar \mu}(\mu_0)$ equal to the empirical mean of $\log_{\bar \mu}(\mu_i)$. 
This is clearly unimodal and results in a significantly different first principal direction. Indeed, moving along this direction, the densities are unimodal and their location changes quite substantially, although they present a bit of skewness (see the red densities in \Cref{fig:gaussian_pds} (bottom).
One could argue that this direction is more interpretable than the one found when centering the PCA at the barycentre.
However, the reconstruction error using this direction is two times higher than the ones using the ``original'' direction from the barycentre.
We argue that, when centering the PCA in other points than the barycentre, the interpretation of the principal directions as the main ``sources of variability'' might be misleading, as shown in this case.
However, it is true that this procedure might yield practically relevant insights on the data set under analysis: in this case, it is true that the location of the measures changes significantly in the data.

Hence, if the ultimate goal of performing PCA is to gain insight on the data set, and the barycentre presents some features that are not displayed in the datapoints (as in this case, the barycentre is trimodal while all datapoints are unimodal), it might be worth to consider other candidate points where to centre the PCA.
If the goal is to perform dimensionality reduction instead, we argue that the barycentre is the only sensible candidate to centre the PCA, as it will surely lead to smaller reconstruction errors, which are synonymous to a smaller loss of information in the reduced data.

\section{Additional Plots for the Eye Dataset}\label{sec:app_eye2}

\Cref{fig:eye_7clus} reports the refined clusters of the eye dataset obtained by cutting the dendrogram at height 0.3.
When considering the refined clusters in \Cref{fig:eye_7clus}, we see that all clusters are characterised by slightly different shapes: going from left to right, from the first row to the second, the first one has a bump on the left
side and is relatively flat on the right side;
the second one has a bump on the left side and a smaller bump on the right side; the third one is more stretched along the vertical axis; the fourth one has an even bigger bump on the right side and is flatter on the right side compared to the first cluster; the fifth one is similar in shape to the second one, but has more mass in the lower-left area; the sixth one is much thicker in the upper section than all the other clusters; the last one is quite thick as well in the   
upper section but with a small bump on the right side and a flatter profile on the left side compared to the sixth cluster.

\begin{figure}[t]
	\includegraphics[width=\linewidth]{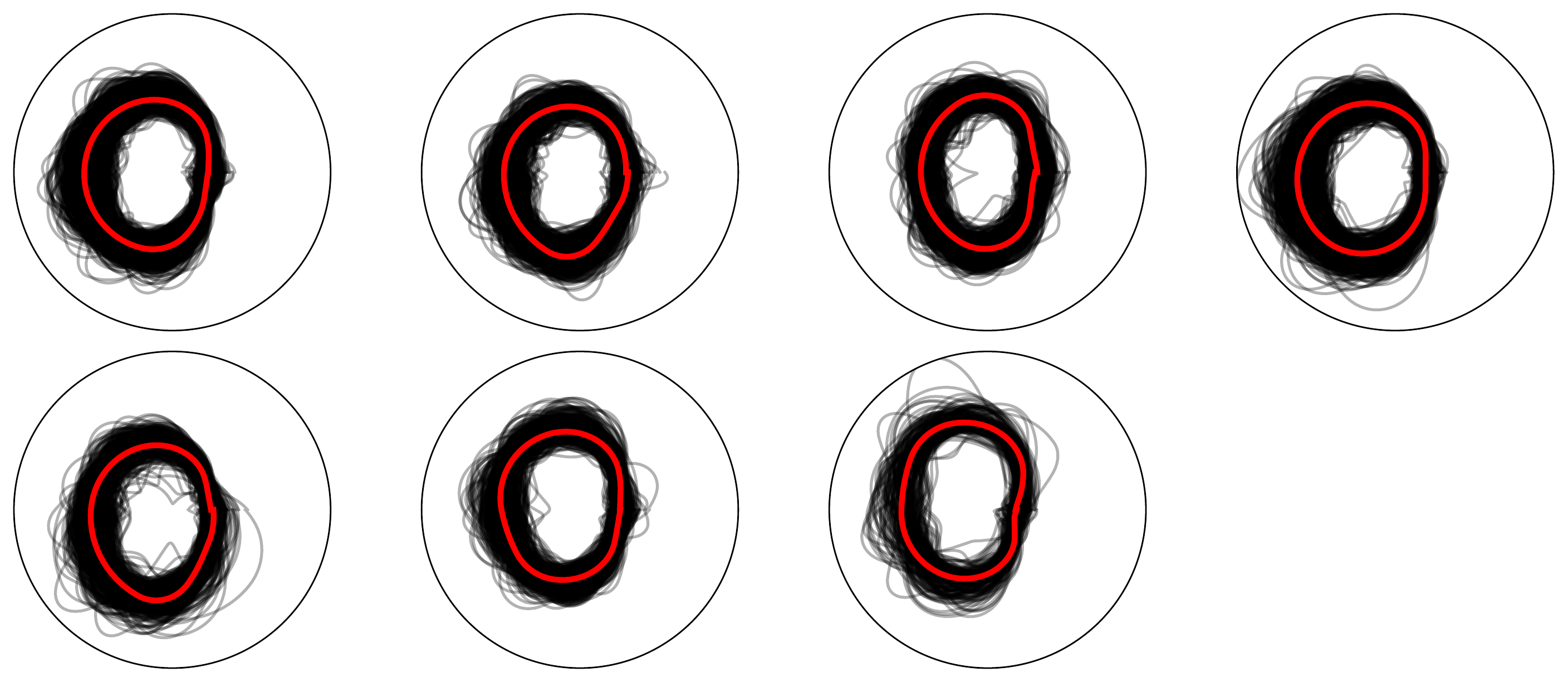}
    \caption{Eye dataset subdivided in seven clusters.}
	\label{fig:eye_7clus}
\end{figure}

\begin{figure}[t]
    \centering
	\begin{subfigure}{0.5\linewidth}
		\includegraphics[width=\linewidth]{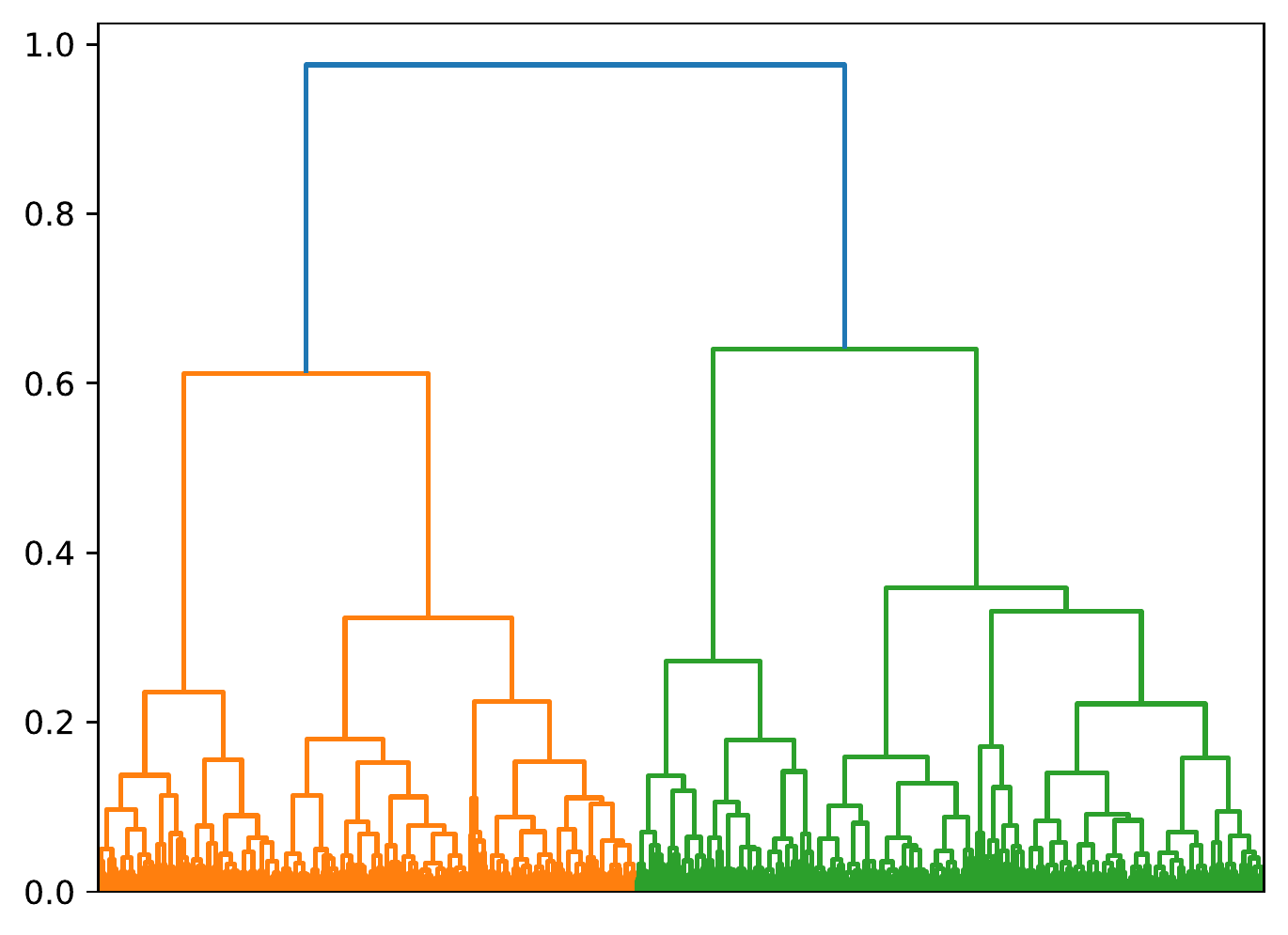}
	\end{subfigure}%
	\begin{subfigure}{0.5\linewidth}
		\includegraphics[width=\linewidth]{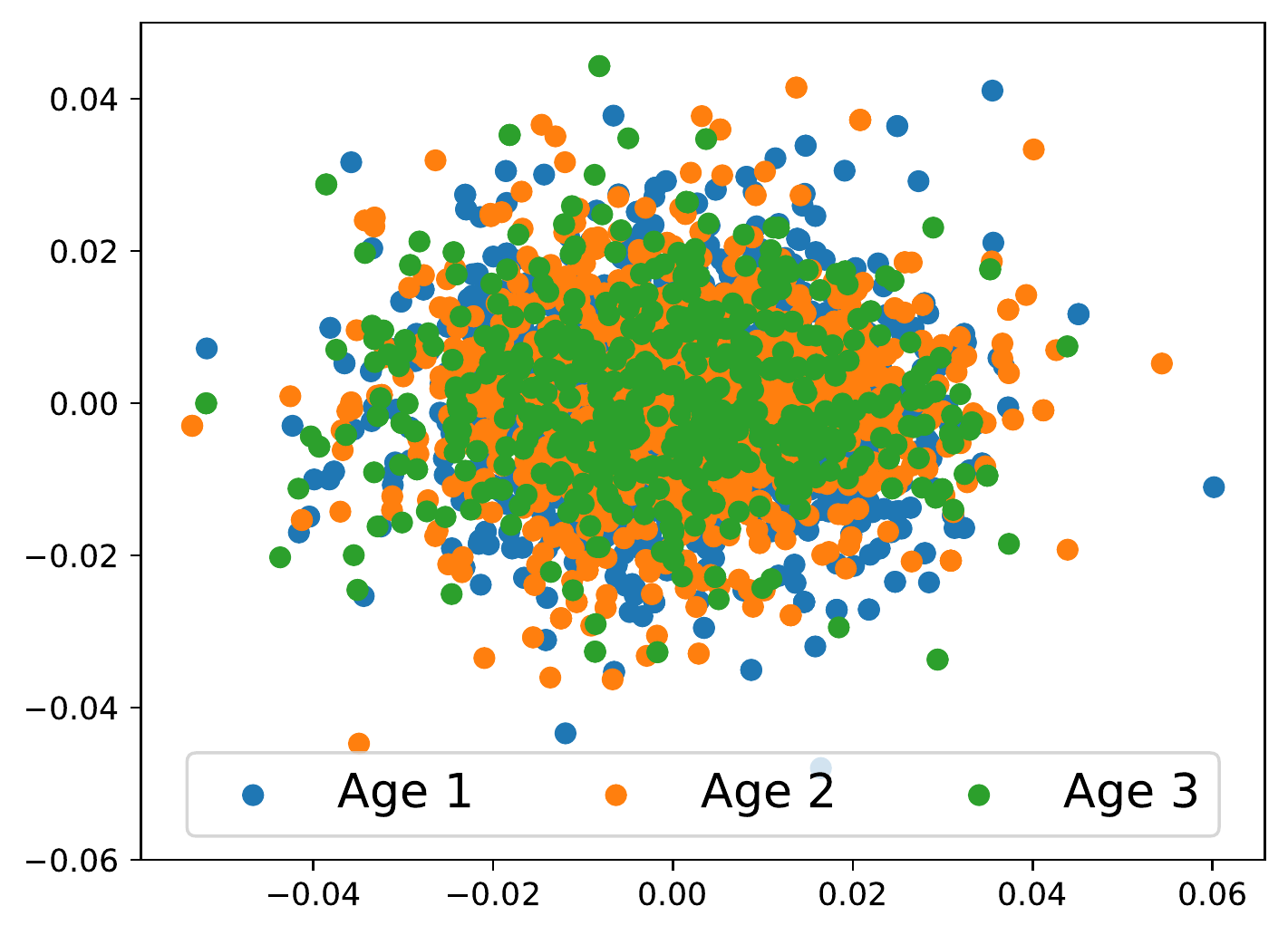}
	\end{subfigure}
	\caption{Hierarchical clustering dendrogram (left) and scatter plot of the scores along the first two principal directions, stratified by age group (colour) for the analysis of the OCT measurements dataset analysed in \Cref{sec:eye}.}
	\label{fig:eye_dend_scores}
\end{figure}

\end{document}